\documentclass[a4paper, 11pt]{article}
\usepackage{jheppub}

\usepackage[utf8]{inputenc}
\usepackage{graphicx}
\usepackage{amsmath}
\usepackage{amssymb}
\usepackage{mathtools}
\usepackage{systeme}
\usepackage{amsfonts}
\usepackage{mathrsfs}

\usepackage{amsthm}
\usepackage{wasysym}
\usepackage[english]{babel}
\usepackage{enumitem}
\usepackage{yfonts}
\usepackage{tikz-cd}
\usepackage{comment}
\usepackage{hyperref}
\usepackage{bm}

\usepackage[toc]{appendix}

\newtheoremstyle{theoremstyle}
  {8pt} 
  {8pt} 
  {\itshape} 
  {} 
  {\bfseries} 
  {.} 
  {.5em} 
  {} 

\newtheoremstyle{definitionstyle}
  {8pt} 
  {8pt} 
  {} 
  {} 
  {\bfseries} 
  {.} 
  {.5em} 
  {} 

\newtheoremstyle{remarkstyle}
  {8pt} 
  {8pt} 
  {} 
  {} 
  {\itshape} 
  {.} 
  {.5em} 
  {} 

\swapnumbers{}

\theoremstyle{theoremstyle}
\newtheorem{theorem}{Theorem}[section]

\newcommand{\thistheoremname}{}
\newtheorem{genericthm}[theorem]{\thistheoremname}

\theoremstyle{theoremstyle}
\newtheorem{lemma}[theorem]{Lemma}
\newtheorem{proposition}[theorem]{Proposition}
\newtheorem{corollary}[theorem]{Corollary}
\newtheorem{conjecture}[theorem]{Conjecture}

\theoremstyle{definitionstyle}
\newtheorem{definition}[theorem]{Definition}

\theoremstyle{remarkstyle}

\newcommand{\B}{\mathcal{B}}
\newcommand{\R}{\mathcal{R}}
\newcommand{\Q}{\mathcal{Q}}
\newcommand{\N}{\mathbb{N}}
\newcommand{\Z}{\mathbb{Z}}
\newcommand{\C}{\mathbb{C}}
\newcommand{\T}{\mathcal{T}}
\renewcommand{\O}{\mathcal{O}}

\renewcommand{\P}{\mathcal{P}}
\newcommand{\F}{\mathcal{F}}
\newcommand{\D}{\mathcal{D}}
\newcommand{\A}{\mathcal{A}}
\newcommand{\U}{\mathcal{U}}
\newcommand{\J}{\mathcal{J}}
\renewcommand{\L}{\mathscr{L}}
\newcommand{\id}{\mathrm{id}}
\renewcommand{\vec}[1]{\mathbf{#1}}
\newcommand{\Y}{\mathcal{Y}}

\let\temp\epsilon
\let\epsilon\varepsilon
\let\varepsilon\temp

\let\temp\phi
\let\phi\varphi
\let\varphi\temp

\arxivnumber{}

\title{\boldmath The quantum group structure of long-range integrable deformations}
\author{Koen Schouten}
\author{and Marius de Leeuw}
\affiliation{School of Mathematics \& Hamilton Mathematics Institute, Trinity College Dublin, Dublin, Ireland}

\emailAdd{kschoute@tcd.ie}
\emailAdd{mdeleeuw@maths.tcd.ie}
\abstract{Quantum integrable spin chains are known to possess a large family of long-range deformations generated by the local, boost and bilocal operators. Although these deformations are well-understood on the level of the pairwise commuting charges, the underlying quantum group structures had not yet been recognised. In this paper, we provide a quantum group-theoretical description for the family of long-range deformations of arbitrary homogeneous Yang-Baxter integrable spin chains up to first order in the deformation parameter. In particular, we show that the deformations are obtained via a twist of the algebraic structure of the underlying quantum group. This twisting results in a generally non-associative algebra that has a non-trivial Drinfeld associator. The Drinfeld associator is then shown to encode the information about the long-range interaction terms for the integrable spin chain. Moreover, the deformed quantum group is shown to contain a large perturbatively associative substructure, thus ensuring the perturbative integrability of the long-range model. The deformed quantum group provides explicit expressions for the Lax operators and $R$-matrices of the long-range deformed models, which manifestly satisfy the RLL relation and the Yang-Baxter equation up to first order in the deformation parameter. In order to derive the results, we introduce algebra elements that we call the algebraic charge densities. As a side result, we provide a conjecture for the explicit expressions of the undeformed charge densities in terms of these algebra elements.
}

\begin{document}
\hspace{0cm}
\maketitle

\newpage
\section{Introduction}
Quantum integrable models, characterised by the existence of a ``sufficiently large'' set of independent mutually-commuting conserved charges, have been of both mathematical and physical interest for nearly a century. The knowledge of such a set of pairwise commuting charges allows one to decompose the Hilbert space into smaller sectors, thus reducing the problem of the spectral decomposition into simpler components. This was first utilised by Bethe in 1931 to find the exact solutions of the Heisenberg spin chain, using what is now called the \textit{Bethe Ansatz}  \cite{Bethe1931zur}. One of the main ideas behind the Bethe Ansatz, and therefore quantum integrability, is that the elastic many-body scattering factorises into a product of two-body scatterings \cite{Caux2011}. In this way, all the physical excitations can be described in terms of the two-body scattering states. For this factorisation to be consistent, the scattering matrix is required to satisfy the \textit{Yang-Baxter equation} \cite{Yang1967SomeExactResults, Baxter1972PartitionFunction, TakhtadshanFaddeev1997}. Quantum integrable models could therefore be (partially) classified by the solutions, called the \textit{$R$-matrices}, of the Yang-Baxter equation \cite{Kulish1982SolutionstoYangBaxter, Sogo1982Classification, Khachatryan2013YangBaxterEightVertex, Vieira2018YangBaxterDifferential, Hietarinta1993Solving2DYBE, Dye2003UnitaryYBE4, Pourkia2018CYBE, deLeeuw2019Classifyingintegrable, deLeeuw2020YangBaxterandBoost, CorcorandeLeeuw2024, Garkun2024Newspectral, deleeuw20264x4solutions}. Moreover, the Yang-Baxter equation naturally gave rise to the development by the Leningrad-St. Petersburg school of an algebraic framework to describe quantum integrability, known as the \textit{quantum inverse scattering method} \cite{STF1979QuantumInverseScattering, TakhtadshanFaddeev1997, faddeev1984integrable,  Sklyanin1982QuantumversionInverseScattering, Faddeev1982IntegrableModels, sklyanin1982some, kulish1983quantum, Sklyanin1992quantuminversescatteringmethod, Korepin1993quantum} (see \cite{Faddeev1995QISMHistory} for an overview). In particular, using the $R$-matrix and a corresponding \textit{Lax operator}, one can construct a generating function for the mutually commuting charges, and additionally derive their eigenstates and spectra using the \textit{Algebraic Bethe Ansatz} \cite{Faddeev1995AlgebraicAspects, faddeev1996algebraicbetheansatzwork}.\\

The realisations of the Yang-Baxter equation led to new algebraic objects that could be viewed as quantisations of the classical Lie groups and Lie algebras, known as \textit{quantum groups} \cite{sklyanin1982some, kulish1983quantum}. Drinfeld and Jimbo independently formalised the description of these quantum groups in the language of Hopf algebras (see e.g. \cite{CPAGtQG, Majid1995FoundationsQG, KasselQuantumGroups, QGAR}), which can be considered as one-parameter deformations of universal enveloping algebras of complex semisimple Lie groups \cite{Drinfled1985, Drinfled1986, Jimbo1986qanalogue}. The \textit{universal $R$-matrix} for a (quasi-triangular) quantum group is then defined through the braiding properties of the coproduct structure, and satisfies the algebraic Yang-Baxter equation if the coproduct is \textit{coassociative}. In this way, quantum integrable models can be obtained through specific representations of the quantum groups. For example, the Heisenberg spin chain corresponds to a representation of the \textit{Yangian algebra} \cite{Drinfled1985, Drinfled1986} (see \cite{Loebbert2016LecturesonYangianSymmetry, MolevYangiansClassicalLie, Bernard1993} for introductions). Instead of quantisations of the universal enveloping algebras, Faddeev, Reshetikhin and Takhtajan, following the spirit of non-commutative geometry  \cite{ConnesNoncommutativeGeometry1990}, later considered the quantisation of the function algebra over a Lie group, known as the \textit{FRT-bialgebra} \cite{FaddeevReshetikhinTakhtajan1989quantAlg, FaddeevReshetikhinTakhtajan1990}. These FRT-bialgebras can be viewed as the \textit{dual objects} to the quantised universal enveloping algebra, and therefore describe the same underlying quantum group structure \cite{QGAR}. The commuting charges, as well as the \textit{creation operators} for the eigenstates, are then realised as representations of elements of the FRT-bialgebra \cite{Drinfled1988}. In particular, the FRT-bialgebras describe the full quantum group structure behind the quantum inverse scattering method.\\

Integrable structures have also been discovered to be present in the AdS/CFT correspondence \cite{Maldacena1997}. In particular, in the planar limit \cite{tHooft1974Planardiagram} of the $\mathcal{N}=4$ SYM model, it was shown that the dilatation operator acting on single trace primary operators can be identified with an integrable nearest-neighbour spin chain Hamiltonian at one-loop order \cite{Minahan2002BetheAnsatz}. Integrability techniques could therefore be used to compute the scaling dimensions of primary local operators at one-loop order \cite{Beisert2003integrableSuper, Beisert2005LongRangeBetheAnsatz}. At higher-loop corrections, the dilatation operator was shown to be given by an integrable spin chain Hamiltonian involving higher-range interactions at each perturbation order of the `t Hooft coupling \cite{Beisert2003Thedilatation, Beisert2004NovelLongRange, Beisert2004su23} (see \cite{Rej2011ReviewLongrange} for a review). This finding later led to a more general class of \textit{long-range deformations} of integrable spin chains, in which the range of interactions grows linearly with the perturbation order \cite{Beisert2006Longrange, Bargheer2008BoostingNearestNeighbour, Bargheer2009}. In particular, by deforming the integrable charges using a \textit{deformation equation}, one obtains a new set of perturbatively integrable long-range charges. A family of such long-range deformations was found to be generated by the \textit{local}, \textit{boost}, and \textit{bilocal} operators \cite{Bargheer2008BoostingNearestNeighbour, Bargheer2009}. Other known long-range integrable spin chains, such as the Haldane-Shastry \cite{Haldane1988Exact, Shastry1988HeisenbergLongRange} and the Inozemtsev spin chain \cite{inozemtsev1995heisenberg, Inozemtsev2003}, also fall under this family of long-range models \cite{Serban2004Planar}. Since there is a large tower of mutually commuting charges, it is also expected that one could derive the eigenvalues and eigenstates for the charges using the Bethe Ansatz. In the thermodynamic limit, i.e. on the doubly-infinite spin chain, this is indeed known to be possible using the \textit{Thermodynamic Bethe Ansatz} \cite{Gromov2010ExactSpectrum, Arutyunov2009ThermodynamicBetheAnsatz, Bombardelli2009Thermo, Bajnok2012ReviewTBA}. However, on finite-size spin chains, it has remained unknown how to perform a Bethe Ansatz for which \textit{wrapping corrections} are included, i.e. when the interaction range exceeds the length of the spin chain. Similarly, up until now, it was not known whether these long-range deformed models have a corresponding underlying quantum group structure. \\

In this paper, we provide, for the first time, a quantum group-theoretical description for the long-range deformations of arbitrary homogeneous Yang-Baxter integrable spin chains up to first order in the deformation parameter. Since the long-range integrable models correspond to deformations of the commuting charges, it is also expected that these are obtained through deformations of the underlying quantum group. Some progress in understanding this quantum group structure was already made in \cite{Gombor2021, Gombor2022WrappingCorrections, deLeeuwRetore2023}, where some general properties for the long-range Lax operators were described. In this paper, we will improve on these results by explicitly describing the corresponding deformations of the FRT-bialgebra up to first order in the deformation parameter, thus providing a concrete foundation for the understanding of the quantum group structure. In particular, the long-range deformed quantum groups are shown to be obtained by means of \textit{twisting}, which deforms the multiplication such that the quantum group becomes \textit{non-associative} and acquires a non-trivial \textit{Drinfeld associator} \cite{Drinfeld1989QuasiHopf}. This breaking of the associativity is necessary for the long-range interactions to appear. Namely, for associative quantum groups, the $R$-matrix factorises into range-two terms, such that the associated integrable model only contains nearest-neighbour interactions. The non-associativity and the appearance of the non-trivial Drinfeld associator in the deformed quantum algebra, however, also allow for higher-range corrections. In particular, the Drinfeld associator exactly encodes the information about the long-range deformation. Moreover, in order to ensure the (perturbative) integrability of the long-range deformed model, the deformed quantum group retains a large perturbatively associative substructure up to first order in the deformation parameter. The long-range deformed quantum groups then provide explicit expressions for the Lax operators and $R$-matrices for the deformed models, which manifestly satisfy the RLL-relations and the Yang-Baxter equation due to this (perturbatively) associative substructure. This then also shows that the integrable long-range spin chains are Yang-Baxter integrable up to first order in the deformation parameter.\\

In order to describe the deformations, we introduce specific algebra elements, which we call the \textit{algebraic charge densities} and which are obtained through derivations of the universal $R$-matrix. Similar to the $R$-matrix, these algebraic charge densities have particularly nice algebraic properties, and, as a side result, are conjectured to give explicit expressions for the charge densities of the undeformed integrable spin chains. We remark that these results are completely general and may be applied to any homogeneous Yang-Baxter integrable spin chain. Moreover, the algebraic charge densities are shown to provide new techniques for deriving the \textit{generalised Sutherland equations}, and also for proving the Yang-Baxter integrability of the long-range deformations.

\subsection*{General overview}
Let us provide a general overview of the framework that is used to describe the long-range deformations in this paper. Throughout the paper, we will consider arbitrary $R$-matrices $\R(u,v) \in \mathrm{End}(V\otimes V)$ that satisfy the Yang-Baxter equation
\begin{equation}
    \R_{12}(u_1,u_2)\R_{13}(u_1,u_3)\R_{23}(u_2,u_3) = \R_{23}(u_2,u_3)\R_{13}(u_1,u_3)\R_{12}(u_1,u_2)
\end{equation}
and are regular, i.e. $\R(u,u) = \P$ with $\P$ the permutation operator. Recall here that in specific cases these $R$-matrices arise from representations of (pseudo)-quasitriangular quantised universal enveloping algebras, like the Yangian or quantum affine algebra. One could use this interpretation to define ``coproducts'' of the $R$-matrices, even if there is no clear underlying quantum universal enveloping algebra. For example, we have
\begin{equation}
    \begin{split}
        \R_{(12)3}(u_1,u_2; u_3) &\coloneq (\Delta\otimes \id)(\R) = \R_{13}(u_1, u_3)\R_{23}(u_2, u_3)\\
        \R_{1(23)}(u_1; u_2, u_3) &\coloneq (\id\otimes \Delta)(\R) = \R_{13}(u_1,u_3)\R_{12}(u_1,u_2).
    \end{split}
\end{equation}
More generally, one can consider $N-1$ consecutive coproducts on the first leg of the $R$-matrix, and $M-N-1$ coproducts on the second leg, to get
\begin{equation}
    \R_{(1,..., N)(N+1,..., M)}(u_1,..., u_{M})\coloneq \overset{N}{\underset{i=1}{\overrightarrow{\prod}}}\overset{M}{\underset{\ j=N+1}{\overleftarrow{\prod}}}\R_{ij}(u_i,u_j),
\end{equation}
for $0<N<M$. Formally, as we will describe in Section~\ref{sec:coquasitriangularityandLfunctionals}, these ``coproducts'' correspond to specific evaluations of the \textit{universal $r$-form} of the associated FRT-bialgebra. One can now also consider the associated Lax operator $\L_{a1}(u) \coloneq \R_{a1}(u,0)$ and monodromy matrix $\T_a(u) \coloneq \L_{aL}(u)\cdots \L_{a1}(u)$, which correspond to specific representations of the associated FRT-bialgebra. The transfer-matrix $t(u) \coloneq \mathrm{tr}_{a}\T_a(u)$ is then well-known to generate a tower of mutually-commuting local charges given by
\begin{equation}\label{eq:overviewlocalcharges}
    \mathbb{Q}_k \coloneq \frac{d^{k-2}}{du^{k-2}}\left(t(u)^{-1}\frac{d}{du}t(u)\right) = \sum_{n=1}^{L} \mathfrak{q}^{(k)}_{n,...,n+k-1},
\end{equation}
for $k\geq 2$ and where $\mathfrak{q}^{(k)} \in \mathrm{End}(V^{\otimes k})$ are the corresponding charge densities.\\

The central objects throughout this paper are the \textit{algebraic charge densities}, which correspond to derivatives of the $R$-matrices. In particular, let us define  a derivative of the $R$-matrix on the spectral parameters corresponding to the first leg:
\begin{equation}
    \begin{split}
        \D_1\R_{(1,...,N)(N+1,...,M)}(u_1,...,u_M) &\coloneq \sum_{i=1}^{N}\frac{d}{du_i}\R_{(1,...,N)(N+1,...,M)}(u_1,...,u_M).
    \end{split}
\end{equation}
For $k\geq 2$, the algebraic charge densities are then defined to be given by
\begin{equation}
    \Q^{(k)}_{(1,...,N)(N+1,...,M)} \coloneq \D_1^{k-2}\left(\R^{-1}_{(1,...,N)(N+1,...,M)}\D_1\R_{(1,...,N)(N+1,...,M)}\right),
\end{equation}
where we have left out the explicit dependence on the spectral parameter in the notation of the above equation. Remark here the similarity with \eqref{eq:overviewlocalcharges}. In Section~\ref{sec:algebraicchargedensities}, we provide a more formal definition in terms of bilinear maps on the FRT-bialgebra. As examples, we have
\begin{equation}
    \begin{alignedat}{2}
        \Q^{(2)}_{12}(u,v) &= \R^{-1}_{12}(u,v)\frac{d}{du}\R_{12}(u,v), \quad &\Q^{(2)}_{1(23)}(u; v,w)  &= \R^{-1}_{1(23)}(u; v,w)\frac{d}{du}\R_{1(23)}(u;v,w)\\
        \Q^{(3)}_{12}(u,v) &= \frac{d}{du}\Q^{(2)}_{12}(u,v), \quad &\Q^{(3)}_{1(23)}(u;v,w) &= \frac{d}{du}\Q^{(2)}_{1(23)}(u;v,w).
    \end{alignedat}
\end{equation}
One may now recognise that the second algebraic charge density evaluated at zero spectral parameters is exactly equal to the charge density for the second charge, i.e. \linebreak $\mathfrak{q}^{(2)}_{12} =  \Q^{(2)}_{12}(0,0)$. Similarly, as we show in Section~\ref{sec:algebraicchargedensities}, the third charge density is given by a difference of algebraic charge densities, that is $\mathfrak{q}^{(2)}_{123} = \Q^{(3)}_{1(23)}(0,0,0) - \Q^{(3)}_{23}(0,0)$. In particular, this leads us to conjecture that all the charge densities are given by
\begin{equation}
    \mathfrak{q}^{(k)}_{1,...,k} = \Q^{(k)}_{1(2,...,k)}(0,...,0) - \Q^{(k)}_{2(3,...,k)}(0,...,0).
\end{equation}
This then provides a closed formula for all the charge densities of arbitrary Yang-Baxter integrable spin chains in terms of the $R$-matrix and derivations thereof. The algebraic charge densities also provide a generalised formula for the Sutherland equation. In particular, let $\L_{a(1,...,k)} = \L_{ak}\cdots \L_{a1}$. As we show in Section~\ref{sec:higherorderSutherlandequation}, the Sutherland equation for the $k$-th charge density is given by
\begin{equation}\label{eq:overviewSutherland}
    [\mathfrak{q}^{(k)}_{1,...,k}, \L_{a(1,...,k)}] = \L_{a(2,...,k)}{}^r\Q^{(k)}_{1,a,(2,...,k-1)}\L_{a1} - \L_{a(3,...,k)}{}^r\Q^{(k)}_{2,a,(3,...,k)}\L_{a(12)},
\end{equation}
where we introduce
\begin{equation}
    {}^r\Q^{(k)}_{1,2,(3,...,k)}(u) \coloneq \Q^{(3)}_{1(2,...,k)}(0;u,0,...,0) - \R^{-1}_{12}(0,u)\Q^{(3)}_{1(3,...,k)}(0,...,0)\R_{12}(0,u)
\end{equation} 
as the \textit{reduced algebraic charge densities} as defined in Section~\ref{sec:algebraicchargedensities}. Remark that ${}^r\Q^{(k)}_{1,2,(3,...,k)}(0) = \mathfrak{q}^{(k)}_{1,...,k}$ since $\R_{12}(0,0) = \P_{12}$.\\

These generalised Sutherland equations now enable us to provide explicit expressions for the Lax operators and $R$-matrices corresponding to the long-range deformations up to first order in the deformation parameter. In particular, consider the monodromy matrix on the doubly-infinite spin chain $\L_{a,(...)}\coloneq \overleftarrow{\prod}_{m\in \Z} \L_{a,m}$ and on the half-infinite spin chain $\L_{a,(n,...)}\coloneq \overleftarrow{\prod}_{m\in \Z_{\geq n}} \L_{a,m}$ and $\L_{a,(...,n)}\coloneq \overleftarrow{\prod}_{m\in \Z_{\leq n}} \L_{a,m}$. As discussed in Section~\ref{sec:LaxforBQ3derivation}, it follows from \eqref{eq:overviewSutherland} that the commutator between the monodromy matrix on the doubly-infinite spin chain and the boost operator $\B[\mathbb{Q}_k] \coloneq \sum_{n\in \Z}n \cdot \mathfrak{q}^{(k)}_{n,...,n+k-1}$ is given by
\begin{equation}\label{eq:overviewcommutatorboost}
    [\B[\mathbb{Q}_k], \L_{a,(...)}] = \sum_{n\in \Z}n\cdot[\mathfrak{q}^{(k)}_{n,...,n+k-1}, \L_{a,(...)}] = \sum_{n\in \Z}\L_{a(n+1,...)}{}^r\Q^{(k)}_{n,a,(n+1,...,n+k-2)}\L_{a,(...,n)},
\end{equation}
for $k\geq 3$. Similarly, the Sutherland equation can be used to compute the commutator between the bilocal operators and the monodromy matrix, see Section~\ref{sec:LaxforBQ3derivation}. It is known that these boost and bilocal operators generate long-range deformations of the spin chain \cite{Bargheer2008BoostingNearestNeighbour, Bargheer2009}. It can be seen in \eqref{eq:overviewcommutatorboost} that the commutator with the boost operator is given by a sum over insertions of the reduced algebraic charge density at position $n$. As discussed in Section~\ref{sec:LaxforBQ3derivation}, such a deformed monodromy matrix on the doubly-infinite spin chain may also be obtained through the Lax operator
\begin{equation}
    \L^{\lambda}_{(a, b_1,...,b_{k-2})n} (u)\coloneq \left(1 + \lambda\cdot  {}^r\Q^{(k)}_{n,a,(b_1,...,b_{k-2})}(u)\right)\L_{(a, b_1,...,b_{k-2})n}(u) + \O(\lambda^2),
\end{equation}
where $\L_{(a, b_1,...,b_{k-2})n}(u) = \L_{an}(u)\P_{b_1,n}\cdots \P_{b_{k-2},n}$. Similarly, we also derive expressions for the Lax operators corresponding to the bilocal deformation in Section~\ref{sec:LaxforBQ3derivation}. We therefore obtain general expressions for the Lax operators of the long-range deformations of arbitrary integrable spin chains. Moreover, as we discuss in Section~\ref{sec:deformationFRTalgebra}, these deformed Lax operators correspond to representations of \textit{twisted} FRT-bialgebras. The explicit descriptions of these twisted algebras then also provide expressions for the corresponding $R$-matrices. For example, as derived in Section~\ref{sec:deformationFRTalgebra}, the $R$-matrix for the $\B[\mathbb{Q}_3]$ deformation is given by
\begin{equation}
    \begin{split}
        \R^\lambda_{(12)(34)}(u,0;v,0) &\coloneq \R_{(12)(34)}(u,0;v,0)  + \O(\lambda^2)\\
        &\quad + \lambda\left({}^r\Q^{(3)}_{4,1,2}(u) \R_{(12)(34)}(u,0;v,0) - \R_{(12)(34)}(u,0;v,0){}^r\Q^{(3)}_{2,3,4}(v)\right),
    \end{split}
\end{equation}
where $\R_{(12)(34)}(u,0;v,0) \coloneq \R_{14}(u,0)\R_{13}(u,v)\R_{24}(0,0)\R_{23}(0,v)$. It is shown in Section~\ref{sec:deformationFRTalgebra} that these Lax operators and $R$-matrices correspond to representations of an associative subalgebra of the twisted quantum algebra, which ensures that the RLL relation and the Yang-Baxter equation hold. We also provide $R$-matrices for the more general boost deformations and the bilocal deformations. One should remark here that, for the long-range deformations, the Lax operators and the $R$-matrices are generally not equal to one another. The origin of this difference may be traced back to the description of the underlying FRT-bialgebra, which is given by a twisted double-crossed product bialgebra as explained in Section~\ref{sec:deformationFRTalgebra}.

\subsection*{Outline}
The outline of this paper is as follows:
\begin{itemize}[leftmargin=*]
    \item \textbf{In Section~\ref{sec:FRTalgebrasandintegrablespinchains}}, we provide a short review of FRT-bialgebras and discuss how they describe the associated spin chains. In particular, we begin our discussion by giving the definition of an $R$-matrix and the corresponding FRT-bialgebra in Section~\ref{sec:definitionFRTalgebra}. Then, in Section~\ref{sec:dualdescription}, we discuss how these FRT-bialgebras can be interpreted as the dual of the more usual quantised universal enveloping algebras, like the Yangian. In Section~\ref{sec:coquasitriangularityandLfunctionals}, we review the concept of coquasitriangularity of these FRT-bialgebras, and show how the Lax operators arise as representations induced from the $\ell$-functionals of the FRT-bialgebras. In Section~\ref{sec:commutingcharges}, we discuss the tower of commuting charges that are generated in the FRT-bialgebra. Then, in Section~\ref{sec:algebraicchargedensities}, we will introduce the algebraic charge densities, which are shown to be related to the charge densities, and we derive some of their properties. In particular, we use the algebraic charge densities to give a conjecture for a general formula for the charge densities of the mutually commuting charges. Lastly, in Section~\ref{sec:higherorderSutherlandequation}, we use the algebraic charge densities to derive the higher-order variants of the Sutherland equation.

    \item \textbf{In Section~\ref{sec:longrangedeformations}}, we will discuss the long-range deformations of integrable spin chains. In particular, we will begin in Section~\ref{sec:classificationLongrangeDeformations} by recalling the three families of long-range deformations as described in \cite{Bargheer2008BoostingNearestNeighbour, Bargheer2009}. Namely, as we will review, the long-range deformations are generated by the local operators, boost operators and bilocal operators. In Section~\ref{sec:LaxforBQ3derivation}, we then use the higher-order Sutherland equations to derive the Lax operators for the three families of long-range deformations up to first order in the deformation parameter. Afterwards, in Section~\ref{sec:deformationFRTalgebra}, we discuss the deformations of the FRT-bialgebra that correspond to the long-range deformations. In particular, we will give explicit twisting elements up to first order in the deformation parameter and show that the twisted Lax operators correspond to the long-range deformation, and that there exist $R$-matrices such that the RLL-relation and the Yang-Baxter equation are satisfied. In Section~\ref{sec:someexamples}, we then explicitly work out the examples for the boost $\B[\mathbb{Q}_3]$ and bilocal $[\mathbb{Q}_2|\mathbb{Q}_3]$ long-range deformations of the XXX Heisenberg spin chain. Lastly, in Section~\ref{sec:longrangedeformedYangian}, we discuss the corresponding long-range deformation of the Yangian algebra.

    \item \textbf{In Appendix~\ref{sec:proofsfordeformation}}, we will provide the proofs of the various statements that are made in the main text. Namely, the proofs are somewhat lengthy, and in order to maintain the readability of the main text, we therefore provide the proofs in the appendix.
    
    \item \textbf{In Appendix~\ref{sec:Thetamorphismrelation}}, we will discuss the relation between the boost $\B[\mathbb{Q}_3]$ deformation of the XXX Heisenberg spin chain and the $\Theta$-morphism that was described in \cite{gromov2014theta}. In particular, in that paper, the $\Theta$-morphism was used to explicitly derive the transfer matrix and the $N$-magnon creation operators for the long-range deformed XXX spin chain. We show that, in our framework, we obtain the same results. We then use this relation to identify the $N$-magnon creation operators within our deformed FRT-bialgebra, thus bringing us a small step closer to understanding the long-range Algebraic Bethe Ansatz.
\end{itemize}

\noindent \textbf{Conventions.} In this paper, the tensor product $\otimes$ will denote the ordinary tensor product over $\C$-vector spaces or $\C$-algebras. For an algebra $\A$ and an element $x\in \A$, we use the leg-numbering notation for its inclusion in an $n$-fold tensor product. In particular, we will write $x_i \coloneq 1^{\otimes(i-1)}\otimes x\otimes 1^{\otimes(n-i)} \in \A^{\otimes n}$ for $1\leq i\leq n$. Similarly, for $x\in \A\otimes \A$, the notation $x_{ij}$ denotes the element as considered in $\A^{\otimes n}$ in the $i$-th and $j$-th factor. This notation generalises to elements $x_{i_1,...,i_m}$ in the higher tensor product $\A^{\otimes m}\subseteq \A^{\otimes n}$ for $m \leq n$. We will denote by $\C[\![u]\!]$ the formal power series over $\C$, and by $\C(\!(u)\!)$ the Laurent series over $\C$. Similarly, $\A[\![u]\!]\coloneq \A\otimes \C[\![u]\!]$ and $\A(\!(u)\!) \coloneq \A\otimes \C(\!(u)\!)$ denote the formal power series and the Laurent series with coefficients in the algebra $\A$, respectively. We remark that for most of this paper, algebras will be considered to be unital and associative. However, in Section~\ref{sec:deformationFRTalgebra}, we will also encounter non-associative algebras. It will be explicitly stated when such non-associative algebras are considered. 

\section{Review of FRT-bialgebras}\label{sec:FRTalgebrasandintegrablespinchains}
We will begin our discussion with a review of FRT-bialgebras and the associated integrable spin chains, which is mostly meant to introduce the basic notation and concepts necessary for the description of the long-range deformed algebras that will be discussed in Section~\ref{sec:longrangedeformations}. The FRT-bialgebras, introduced by Faddeev, Reshetikhin and Takhtajan, can be naturally viewed as deformations/quantisations of the algebra of functions on an (affine) Lie group/algebra \cite{FaddeevReshetikhinTakhtajan1989quantAlg, FaddeevReshetikhinTakhtajan1990}, and are defined through the solutions to the Yang-Baxter equation. These FRT-bialgebras can be considered as the dual description of the usual Drinfeld-Jimbo approach to quantum groups, which are quantisations of the corresponding (affine) Lie algebras \cite{Drinfled1985, Drinfled1986, Jimbo1986qanalogue}. Moreover, the FRT-bialgebras describe the underlying algebraic structure behind the quantum inverse scattering method \cite{STF1979QuantumInverseScattering, TakhtadshanFaddeev1997, faddeev1984integrable,  Sklyanin1982QuantumversionInverseScattering, Faddeev1982IntegrableModels, sklyanin1982some, kulish1983quantum, Sklyanin1992quantuminversescatteringmethod, Korepin1993quantum, Faddeev1995QISMHistory}. In Section~\ref{sec:definitionFRTalgebra}, we will review the definition of the $R$-matrix and the associated FRT-bialgebra. Then, in Section~\ref{sec:dualdescription}, we will discuss how one can interpret these FRT-bialgebras as the algebra of functions over the corresponding Drinfeld-Jimbo type quantum groups. In Section~\ref{sec:coquasitriangularityandLfunctionals}, we then review the concept of coquasitriangularity and $\ell$-functionals. In Section~\ref{sec:commutingcharges}, we review the tower of commuting charges that can be obtained from the representations of the FRT-bialgebra, thus constructing integrable spin chains. Then, in Section~\ref{sec:algebraicchargedensities}, we introduce new algebra elements, the algebraic charge densities, which are conjectured to be related to the charge densities. Lastly, in Section~\ref{sec:higherorderSutherlandequation}, we use the algebraic charge densities to derive the higher-order Sutherland equations, corresponding to the higher charge densities. 

\subsection{The FRT-bialgebra}\label{sec:definitionFRTalgebra}
In defining the FRT-bialgebras, we will closely follow the discussions as given in, for example, \cite{EtingofSchiffmann1998, Nazarov2019doubleYangian} (see also \cite{Majid1995FoundationsQG}). Before discussing the FRT-bialgebra, we will first introduce some notation. Throughout the paper, we let $V$ be a finite-dimensional (complex) vector space. Moreover, we consider an element $\R(u,v)\in \mathrm{End}(V\otimes V)[\![u,v]\!]$, called the \textit{$R$-matrix}, that is analytic (and non-constant) in $u$ and $v$ and satisfies the \textit{Yang-Baxter equation} \cite{Yang1967SomeExactResults, Baxter1972PartitionFunction}
\begin{equation}
    \label{eq:matrixYBE}
    \R_{12}(u_1,u_2)\R_{13}(u_1,u_3)\R_{23}(u_2,u_3) = \R_{23}(u_2,u_3)\R_{13}(u_1,u_3)\R_{12}(u_1,u_2),
\end{equation}
for all $u_1,u_2,u_3\in \C$. Moreover, we will assume that the $R$-matrix is \textit{regular}, i.e. $\R(u,u) = \P$, where $\P \colon v\otimes w \mapsto w\otimes v$ is the permutation operator on $V\otimes V$, and that it is \textit{unitary}, i.e. it satisfies the braiding unitarity equation $\R_{12}(u,v)\R_{21}(v,u) = I_{V\otimes V}$.\footnote{Remark here that the regularity of the $R$-matrix, together with the Yang-Baxter equation, already imply braiding unitarity. Indeed, for $u_3 = u_1$, the Yang-Baxter equation \eqref{eq:matrixYBE} implies that $\R_{12}(u_1,u_2)\R_{21}(u_2, u_1) = \R_{23}(u_2, u_1)\R_{32}(u_1, u_2)$, which can only hold if $\R_{12}(u_1,u_2)\R_{21}(u_2, u_1) = \rho(u,v)I_{V\otimes V}$ for some function $\rho(u,v)$. We then always assume the $R$-matrix to be normalised such that $\rho(u,v) = 1$. Also, note that braiding unitarity generally only holds up to a discrete set of points in $\C$, as there might be specific points at which the $R$-matrix becomes singular.} The variables $u$ and $v$ will be referred to as the \textit{spectral parameters}. Various classifications of solutions to the Yang-Baxter equation can be found in \cite{Kulish1982SolutionstoYangBaxter, Sogo1982Classification, Khachatryan2013YangBaxterEightVertex, Vieira2018YangBaxterDifferential, Hietarinta1993Solving2DYBE, Dye2003UnitaryYBE4, Pourkia2018CYBE, deLeeuw2019Classifyingintegrable, deLeeuw2020YangBaxterandBoost, CorcorandeLeeuw2024, Garkun2024Newspectral, deleeuw20264x4solutions}.\\

We will fix an orthonormal basis $\{e_i\}_{i=1}^{N}$ for $V$, where $N  = \dim V$, and additionally consider the corresponding standard basis $\{E_{ij}\}_{i,j=1}^{N}$ for $\mathrm{End}(V)$. Next, we let  $\C\langle t^{(-k)}_{ij}\rangle$ be the free algebra generated by the elements $t^{(-k)}_{ij}$ for $k\in \Z_{\geq 1}$ and $1 \leq i,j \leq N$. We define the corresponding (formal) generating function for the algebra elements to be given by \cite{Nazarov2019doubleYangian}
\begin{equation}\label{eq:matrixelementsexpansion}
    t_{ij}(u) \coloneq \delta_{ij} +t_{ij}^{(-1)} + \sum_{k=1}^\infty t^{(-k-1)}_{ij} u^k \in \C\langle t^{(-k)}_{ij}\rangle[\![u]\!].
\end{equation} 
In order to lighten up the notation in the rest of this paper, we introduce the $N\times N$ matrix $\T(u) \coloneq (t_{ij}(u))_{i,j=1}^{N}$ with matrix coefficients $t_{ij}(u)$. Alternatively, one may use a similar notation as given in e.g. \cite{MolevNazarovOlshansky1996Yangians, MolevYangiansClassicalLie} and write 
\begin{equation}\label{eq:matrixcorepresentationT}
    \T(u) = \sum_{i,j = 1}^{N}E_{ij}\otimes t_{ij}(u) \in \mathrm{End}(V)\otimes\C\langle t_{ij}^{(-k)}\rangle[\![u]\!].
\end{equation}
This matrix may be interpreted as defining the \textit{corepresentation} $\phi_u \colon V\to V\otimes \C\langle t^{(-k)}_{ij}\rangle[\![u]\!]$ given by $\phi_u(e_j) = \sum_{i=1}^{N} e_i\otimes t_{ij}(u)$. The matrix $\T(u)$ is then also called the \textit{matrix corepresentation} of $\C\langle t^{(-k)}_{ij}\rangle$. Of course, this is only a corepresentation if $\C\langle t^{(-k)}_{ij}\rangle$ is equipped with the appropriate coalgebra structure \cite{CPAGtQG}, which we provide in Definition~\ref{def:FRTalgebradef}. Suppose now that we have some tensor product $V^{\otimes m}$ for $m\geq 1$, then we define the matrix of algebra elements at position $n$ in this tensor product to be given by
\begin{equation}\label{eq:matrixatpositionn}
    \T_n(u) \coloneq \sum_{i,j = 1}^{N}\left[I_V^{\otimes(n-1)}\otimes E_{ij}\otimes I_V^{\otimes (m-n)}\right]\otimes t_{ij}(u) \in \mathrm{End}(V^{\otimes m})\otimes \C\langle t_{ij}^{(-k)}\rangle[\![u]\!]
\end{equation}
for $1\leq n \leq m$. Alternatively, this may be interpreted as applying the corepresentation $\phi_u$ in the $n$-th coordinate of the tensor product $V^{\otimes m}$. Similarly, we define the tensor product of matrices to be given by
\begin{equation}\label{eq:tensorproductofmatrices}
    \T(u)\otimes \T(v) \coloneq \sum_{i,j,k,\ell = 1}^{n} E_{ij}E_{k\ell}\otimes \left[t_{ij}(u)\otimes t_{k\ell}(v)\right] \in \mathrm{End}(V)\otimes \C\langle t_{ij}^{(-k)}\rangle[\![u]\!]\otimes \C\langle t_{ij}^{(-k)}\rangle[\![v]\!].
\end{equation}
This tensor product corresponds to the successive application of the corepresentation, i.e. $(\phi_u\otimes \id)\circ \phi_v$. Remark that $E_{ij}E_{k\ell} = \delta_{jk}E_{i\ell}$, after which one recognises the ordinary coproduct rule for $\C\langle t^{(-k)}_{ij}\rangle$ \cite{QGAR} (see also Definition~\ref{def:FRTalgebradef}). Note here that, in principle, the definitions in \eqref{eq:matrixatpositionn} and \eqref{eq:tensorproductofmatrices} may be combined to consider tensor products of a more general form like $\T_{n}(u)\otimes \T_{m}(v)$, which is then defined naturally by the successive application of the corepresentation at different coordinates.\\

Based on \cite{FaddeevReshetikhinTakhtajan1989quantAlg, FaddeevReshetikhinTakhtajan1990} (see also e.g. \cite{EtingofSchiffmann1998, Nazarov2019doubleYangian}), we have the following definition:
\begin{definition}\label{def:FRTalgebradef}
    Let $\R(u,v)\in \mathrm{End}(V\otimes V)[\![u,v]\!]$ be a regular solution to the Yang-Baxter equation \eqref{eq:matrixYBE}. The associated \textit{FRT-bialgebra}, denoted by $\A(\R)$, is the unital and associative algebra generated by the elements $t_{ij}^{(-k)}$ for $k\in \Z_{\geq 1}$ and $1\leq i,j\leq N$, subject to the relations
    \begin{equation}
        \label{eq:RTTrelations}
        \R_{12}(u,v)\T_{1}(u)\T_2(v) = \T_2(v)\T_1(u)\R_{12}(u,v),
    \end{equation}
    and with a bialgebra structure given by\footnote{Recall that the coproduct is a homomorphism $\Delta \colon \A(\R) \to \A(\R)\otimes \A(\R)$. Here, $\Delta(\T(u))$ is defined as taking the coproduct at each matrix element of the matrix corepresentation $\T(u)$. Similarly, for the counit $\epsilon \colon \A(\R)\to \C$, we define $\epsilon(\T(u))$ as taking the counit at each matrix element.}
    \begin{equation}
        \Delta(\T(u)) = \T(u)\otimes \T(u)\quad\text{and} \quad \epsilon(\T(u)) = I_V.
    \end{equation}
\end{definition}
\noindent We refer to \cite{QGAR, CPAGtQG, KasselQuantumGroups, Majid1995FoundationsQG} for some standard literature on bialgebras, Hopf algebras and quantum groups. From now on forward, we will denote by $\T(u)$ the corresponding matrix in $\mathrm{End}(V)\otimes \A(\R)[\![u]\!]$, i.e. view $\T(u)$ as a matrix corepresentation of $\A(\R)$. Remark here that $\A(\R)$ is a bialgebra, but \textit{not} a Hopf algebra. In particular, to define an antipode on $\A(\R)$, the inverse of the matrix $\T(u)$ needs to be well-defined, which is only true in a completion of $\A(\R)$. If we let $\A^\circ(\R)$ be the formal completion of $\A(\R)$ relative to the descending filtration defined by $\deg t^{(-k)}_{ij} = k$, then $\T(u)$ is invertible in $\text{End}(V)\otimes \A^\circ(\R)$, such that $\A^\circ(\R)$ is a Hopf algebra with antipode $S(\T(u)) = \T(u)^{-1}$  \cite{Nazarov2019doubleYangian}.\\

Lastly, we remark that our discussed FRT-bialgebras are not necessarily the same as the usual RTT-realisations of some known quantum groups. For example, for an $R$-matrix corresponding to the Yangian (see e.g. \cite{Drinfled1985, MolevYangiansClassicalLie, MolevNazarovOlshansky1996Yangians}), the associated FRT-bialgebra $\A(\R)$ corresponds to the \textit{dual Yangian} \cite{Nazarov2019doubleYangian}, which can be seen from the expansion in $u$ of $t_{ij}(u)$ as defined in \eqref{eq:matrixelementsexpansion}, which is around $u=0$ instead of around $u=\infty$. Therefore, one should not confuse the generators $t^{(-k)}_{ij}$ for the FRT-bialgebra $\A(\R)$ with the generators for the Yangian. Moreover, for FRT-bialgebras related to deformations of specific classical Lie algebras, one often also imposes additional constraints on the algebra. For example, one might impose the quantum determinant to be equal to one, or impose additional orthogonality or symplectic conditions \cite{MolevYangiansClassicalLie}. We will not impose such additional constraints to maintain generality. In the case of quantum affine algebras $\U_q(\widehat{\mathfrak{g}})$, it is usual to have a (trigonometric) $R$-matrix that depends on (multiplicative) spectral parameters $z_1, z_2$ such that $\R_q(z_1,z_2) =\R_q(z_1/z_2)$ (see e.g. \cite{DingFrenkel1992}). However, under this parametrisation, the $R$-matrix is not analytic at $z_1,z_2 = 0$, such that the above construction does not work. Therefore, one instead needs to consider the reparametrisation $z = q^{u}$, such that the $R$-matrix $\R(u_1,u_2) \coloneq \R_q(q^{u_1 - u_2})$ is analytic around $u_1,u_2 = 0$. Note that in that case, the corresponding FRT-bialgebra $\A(\R)$ in Definition~\ref{def:FRTalgebradef} differs from the usual constructions as given in, for example, \cite{FaddeevReshetikhinTakhtajan1989quantAlg, FaddeevReshetikhinTakhtajan1990, Drinfled1988, DingFrenkel1992}, as the algebra generators depend on the choice of parametrisation. In order to maintain generality in our discussion, we will always assume that $\R(u_1,u_2)$ is analytic around $u_1,u_2 = 0$, and therefore consider the appropriate corresponding parametrisations of the $R$-matrix and the ``Yangian-like'' FRT-bialgebras. However, all the techniques developed in this paper could, for example, also be directly applied to the ``true'' trigonometric FRT-bialgebra corresponding to the quantum affine algebra.

\subsection{The dual description}
\label{sec:dualdescription}
It is well-known that the above-described FRT-algebraic construction for (affine) quantum groups corresponds to the dual description of the Drinfeld-Jimbo construction \cite{Drinfled1988}, i.e. the quantisation of the universal enveloping algebras. In particular, there exists a dual pairing of bialgebras between them. Here, we will give a brief exposition of this duality, using the specific realisation of the Yangian as a simple example. Of course, the following generalises to more general quantum affine algebras.\\

Recall the definition of the Yangian: Let $\mathfrak{g}$ be a finite-dimensional complex simple Lie algebra and $\hbar > 0$ some deformation parameter. The Yangian $\mathcal{Y}_\hbar (\mathfrak{g})$ is the algebra generated by elements $x, \J(x)$ for $x\in \mathfrak{g}$, with $[x, \J(y)] = \J([x,y])$ for $x,y\in \mathfrak{g}$, together with the \textit{Serre relations} (see e.g. \cite{CPAGtQG, Drinfled1986, Drinfled1988}), which we will not denote explicitly. The Yangian may be viewed as a deformation of the universal enveloping algebra of the loop algebra $\U(\mathfrak{g}[u])$. In particular, the Yangian $\mathcal{Y}_\hbar (\mathfrak{g})$ has a Hopf algebra structure given by 
\begin{equation}\label{eq:coproductYangian}
    \begin{alignedat}{2}
        \Delta(x) &= x\otimes 1 + 1\otimes x, \quad  &\Delta(\J(x)) &= \J(x)\otimes 1 + 1\otimes \J(x) + \frac{\hbar}{2}[x\otimes 1, t]
     \end{alignedat}
 \end{equation}
 and $\epsilon(x) = \epsilon(\J(x)) = 0$, $S(x) = -x$, $S(\J(x)) = -\J(x) + \frac{1}{4}cx$ for $x\in \mathfrak{g}$. Here, $c$ is the eigenvalue of the Casimir element $t$ in the adjoint representation, which we regard as an element of $\U(\mathfrak{g})^{\otimes 2}$ \cite{CPAGtQG}. Moreover, there exists an automorphism $\B_u \colon \mathcal{Y}_{\hbar}(\mathfrak{g})\to \mathcal{Y}_{\hbar}(\mathfrak{g})$ given by $\B_u(x) = x$ and $\B_u(\J(x)) = \J(x)+ux$ for $x\in \mathfrak{g}$, which corresponds to a reparametrisation of the spectral parameter $u$.\\
 
It was shown in \cite{Drinfled1985} that the Yangian is \textit{pseudo-triangular}: there exists an element $\R(u) = 1+\sum_{k=1}^{\infty}\R_k\left(\frac{\hbar}{u}\right)^k$, called the \textit{universal $R$-matrix}, with $\R_k \in \mathcal{Y}_\hbar(\mathfrak{g})\otimes \mathcal{Y}_\hbar(\mathfrak{g})$ such that
 \begin{equation}
    \label{eq:quasitriangularityconditions}
    \begin{gathered}
        (\Delta \otimes \mathrm{id})(\mathcal{R}(u)) = \mathcal{R}_{13}(u)\mathcal{R}_{23}(u), \quad (\mathrm{id} \otimes \Delta)(\mathcal{R}(u)) = \mathcal{R}_{13}(u)\mathcal{R}_{12}(u),\\
        \mathcal{R}(u) (\mathcal{B}_u\otimes \mathrm{id})(\Delta(x)) = (\mathcal{B}_u\otimes \mathrm{id})(\Delta^{\mathrm{cop}}(x))\mathcal{R}(u),
    \end{gathered}
\end{equation}
for all $x\in \mathcal{Y}_\hbar(\mathfrak{g})$, where $\Delta^{\mathrm{cop}} = \sigma\circ \Delta$ with $\sigma$ the flip operator $\sigma\colon x\otimes y \mapsto y\otimes x$. Moreover, the $R$-matrix is\textit{ unitary}, i.e. $\R_{12}(u)\R_{21}(-u) = 1\otimes 1$ and it satisfies $(\B_u\otimes \B_{v})\R(w) = \R(w+u-v)$. In the following, we will often use the (abuse of) notation $\R \coloneq \R(0)$.\footnote{Although $\R(0)$ is not well-defined as an algebraic element in $\mathcal{Y}_\hbar(\mathfrak{g})\otimes \mathcal{Y}_\hbar(\mathfrak{g})$, it does have a well-defined expression in specific representations of the Yangian.} It follows immediately from \eqref{eq:quasitriangularityconditions} (see e.g. \cite{Drinfled1985, Drinfled1988}) that the $R$-matrix satisfies the (now algebraic) Yang-Baxter equation
\begin{equation}
    \R_{12}(u_1-u_2)\R_{13}(u_1-u_3)\R_{23}(u_2-u_3) = \R_{23}(u_2-u_3)\R_{13}(u_1-u_3)\R_{12}(u_1-u_2),
\end{equation}
for $u_1,u_2,u_3\in \C$. Note here that this $R$-matrix is of difference form, i.e. satisfies $\R(u,v) = \R(u-v)$.\\

Any representation of the Yangian algebra will now give rise to a matrix solution of the Yang-Baxter equation \eqref{eq:matrixYBE}. Therefore, let $(\rho,V)$ be a (finite-dimensional) representation of $\mathcal{Y}_\hbar(\mathfrak{g})$ such that $(\rho\otimes \rho)(\R(0)) = \P$ and define $\rho_u \coloneq \rho\circ \B_u$, which we call an \textit{evaluation representation} at $u$.\footnote{In practice, this representation $\rho$ should correspond to a fundamental representation of $\mathfrak{g}$ in order for the dual pairing to be possibly non-degenerate. For example, for $\mathfrak{g} = \mathfrak{sl}_2$, $\rho$ should correspond to a spin-1/2 representation. It can be noted here that if we instead choose $\rho$ to correspond to, for example, a spin-1 representation, the associated FRT-bialgebra is a subalgebra of the bialgebra for the spin-1/2 representation. Namely, any matrix element for the spin-1 representation can be written as a sum of products of matrix elements for the spin-1/2 representation.} The $R$-matrix $\R_\rho(u,v) \coloneq (\rho\otimes \rho)(\R(u-v))$ is now a regular solution to the Yang-Baxter equation \eqref{eq:matrixYBE} and therefore defines an associated FRT-bialgebra $\A(\R_\rho)$. Importantly, there is a \textit{dual pairing} as bialgebras between $\A(\R_\rho)$ and $\mathcal{Y}_\hbar(\mathfrak{g})$, given by the bilinear map \cite{Drinfled1988}
\begin{equation}
    \label{eq:dualpairingYangian}
    \langle \cdot,\cdot\rangle \colon  \A(\R_\rho)\times \mathcal{Y}_\hbar(\mathfrak{g}) \to \C, \quad \langle t_{ij}(u),x\rangle \coloneq \langle e_i|\rho_u(x)|e_j\rangle,
\end{equation}
which satisfies
\begin{equation}
    \label{eq:dualpairingprops}
    \begin{alignedat}{2}
        \langle a\otimes b, \Delta_\mathcal{Y}(x) \rangle &= \langle ab, x \rangle, \quad &\langle 1_\A, x\rangle &= \epsilon_\mathcal{Y}(x),\\
        \langle \Delta_\A(a), x\otimes y \rangle &= \langle a, xy\rangle \quad &\langle a, 1_\mathcal{Y}\rangle& = \epsilon_\A(a),
    \end{alignedat}
\end{equation}
for $a,b\in \A(\R_\rho)$ and $x,y\in \mathcal{Y}_\hbar(\mathfrak{g})$ and where $\langle a\otimes b,x\otimes y\rangle = \langle a,x\rangle\langle b,y\rangle$ \cite{QGAR}. In particular, this shows that the elements $t_{ij}(u)$ in the FRT-bialgebra $\A(\R_\rho)$ correspond to matrix elements for the representation $\rho$ of $\mathcal{Y}_\hbar(\mathfrak{g})$ \cite{Drinfled1988}. Through the dual pairing \eqref{eq:dualpairingYangian}, one can then view elements in $\A(\R_\rho)$ as functions $\mathcal{Y}_\hbar(\mathfrak{g}) \to \C$, and elements in $\mathcal{Y}_\hbar(\mathfrak{g})$ as functions $\A(\R_\rho)\to \C$, i.e. $\A(\R_\rho)\subseteq \mathcal{Y}_\hbar(\mathfrak{g})^*$ and $\mathcal{Y}_\hbar(\mathfrak{g})\subseteq \A(\R_\rho)^*$, where $\A(\R_\rho)^*$ and $\mathcal{Y}_{\hbar}(\mathfrak{g})^*$ are the bialgebra duals \cite{QGAR} (see also Section~\ref{sec:coquasitriangularityandLfunctionals}). From \eqref{eq:dualpairingprops}, it can then be remarked that the product on $\A(\R_\rho)$ corresponds to the coproduct on $\mathcal{Y}_\hbar(\mathfrak{g})$ and vice versa. This can also explicitly be seen by comparing \eqref{eq:RTTrelations} and \eqref{eq:quasitriangularityconditions}, where the $R$-matrix defines a braiding relation on the product of $\A(\R_{\rho})$, while it gives a braiding relation on the \textit{co}product of $\mathcal{Y}_{\hbar}(\mathfrak{g})$. Similarly, the unit of one bialgebra corresponds to the counit of the other. Lastly, if we extend the dual pairing to the completion $\A^\circ(\R)$, it can be shown \cite{QGAR} that the dual pairing \eqref{eq:dualpairingYangian} implies that $\langle S_\A(a), x\rangle = \langle a, S_\mathcal{Y}(x)\rangle$, such that the antipodes of the two Hopf algebras also correspond to each other. Remark that this duality is equivalent to the usual construction of the dual Yangian as in, for example, \cite{Nazarov2019doubleYangian}.\\

Of course, the above construction is more generally true for affine quantum algebras and extends beyond Yangians. Although for a general $R$-matrix and FRT-bialgebra $\A(\R)$, it is not necessarily known what the corresponding Drinfeld-Jimbo type quantum algebra is, one can still interpret the FRT-bialgebra as being dual to one, and think of the elements $t_{ij}(u)$ as corresponding to matrix elements of a given (evaluation) representation. In the rest of this paper, we will solely work in terms of the FRT-bialgebras, while frequently coming back to this interpretation. Note here that in general, the dual pairing \eqref{eq:dualpairingYangian} is not non-degenerate. In order to obtain a non-degenerate pairing, one would need to quotient out by some ideal, which corresponds to imposing additional constraints on $\A(\R)$. For example, one may require the quantum determinant to be equal to one, or impose additional orthogonality or symplectic conditions \cite{Drinfled1988}. Again, we will not consider such quotients to maintain generality.

\subsection{Coquasitriangularity and $\ell$-functionals}
\label{sec:coquasitriangularityandLfunctionals}
Since the Drinfeld-Jimbo construction of the quantum algebra admits a (pseudo) universal $R$-matrix, the duality as described in the previous section also implies the existence of a dual notion of the universal $R$-matrix for the FRT-bialgebra. In particular, the dual bracket \eqref{eq:dualpairingYangian} and the universal $R$-matrix $\R$ of $\mathcal{Y}_{\hbar}(\mathfrak{g})$ induce a bilinear map $\mathbf{r} \colon \A(\R_\rho)\times\A(\R_\rho) \to \C$, given by $\mathbf{r}(a,b) \coloneq \langle a\otimes b, \mathcal{R}\rangle$ for $a,b\in \A(\R_\rho)$. The function $\mathbf{r}$ is then the dual object to the universal $R$-matrix, which we call the \textit{universal $r$-form}. The properties that such a universal $r$-form satisfies are similar to the properties of the universal $R$-matrix, but where the roles of the comultiplication and multiplication are interchanged by the dual pairing through \eqref{eq:dualpairingprops}.\\

To state the dual notion of quasitriangularity for general bialgebras, for which we follow \cite{QGAR}, we introduce some notation. For a general bialgebra $\A$, we denote the coproduct for an element $a\in \A$ by the Sweedler notation $\Delta(a) = \sum a_{(1)}\otimes a_{(2)}$. Similarly, for a higher coproduct, one writes $\Delta^{(n)}(a) = \sum a_{(1)}\otimes a_{(2)}\otimes \cdots \otimes a_{(n+1)}$ since the coproduct is coassociative. For linear maps $f,g \colon \A\to \C$, we then define the convolution product $\ast$ by $(f\ast g)(a) \coloneq \sum f(a_{(1)})g(a_{(2)})$ for $a\in \A$. Similarly, one can define a ``coproduct'' on these functions by $\Delta(f)\colon \A\otimes \A\to \C, (a,b)\mapsto f(ab)$ for $a,b\in \A$. Of course, this convolution product and coproduct on the space of functions just correspond to the ordinary product and coproduct in the dual algebra in the context of the dual pairing \eqref{eq:dualpairingprops}. Lastly, we define $\epsilon(f) = f(1)$, after which the space of functions $\A\to \C$ admits a bialgebra-like\footnote{\label{footnote:coproductdualalgebra}Note here that, currently, the ``coproduct'' on the space of functions is a map $\Delta \colon \A^* \to (\A\otimes \A)^*$. In order for $\A^*$ to be a proper bialgebra, one requires the image of $\Delta$ to lie in $\A^*\otimes \A^* \subset (\A\otimes \A)^*$, which is not true in general for infinite-dimensional bialgebras. When we refer to the coproduct of $\A^*$, we will thus refer to this weaker notion of $\Delta$.} structure, which we denote by $\A^*$.\\

We now have the following definition \cite{QGAR, KasselQuantumGroups}:
\begin{definition}
    A bialgebra $\A$ is called \textit{coquasitriangular}\footnote{In literature, coquasitriangular bialgebras may also be called \textit{dual quasitriangular bialgebras} \cite{Majid1993BraidedGroups}, \textit{braided bialgebras} \cite{Larson1991Twodualclasses, Doi1993braidedbialgebras}, or \textit{cobraided bialgebras} \cite{KasselQuantumGroups}.} if there is a bilinear form $\mathbf{r} \colon \A\times \A\to \C$, called the \textit{universal $r$-form}, such that:
    \begin{itemize}[leftmargin=*]
        \item $\mathbf{r}$ is invertible under the convolution product, i.e. there exists another bilinear form $\mathbf{r}^{-1} \colon \A\times\A\to \C$ such that $(\mathbf{r}\ast \mathbf r^{-1})(a,b) = \epsilon(a)\epsilon(b) = (\mathbf{r}^{-1}\ast \mathbf r)(a,b)$ for all $a,b\in \A$.
        \item The universal $r$-form determines a braiding structure on the algebra, i.e.
        \begin{equation}
            \label{eq:braidingrform}
            \sum \mathbf{r}(a_{(1)}, b_{(1)})a_{(2)}b_{(2)} = \sum b_{(1)}a_{(1)}\mathbf{r}(a_{(2)}, b_{(2)}).
        \end{equation}
        \item The universal $r$-form factorises under multiplication, i.e.
        \begin{equation}
            \label{eq:factorisationrform}
            \begin{split}
                \mathbf{r}(a\cdot b, c) &= (\mathbf{r}_{13}\ast \mathbf{r}_{23})(a,b,c) = \sum \mathbf{r}(a, c_{(1)})\mathbf{r}(b, c_{(2)})\\
                \mathbf{r}(a, b\cdot c) &= (\mathbf{r}_{13}\ast \mathbf{r}_{12})(a,b,c)=\sum \mathbf{r}(a_{(1)}, c)\mathbf{r}(a_{(2)}, b),
            \end{split}
        \end{equation}
        for $a,b,c\in \A$, where we introduced the notation $\mathbf{r}_{12}(a,b,c) \coloneq \mathbf{r}(a,b)$, and analogously for $\mathbf{r}_{13}$ and $\mathbf{r}_{23}$.
    \end{itemize}
    Moreover, the bialgebra $\A$ is called \textit{cotriangular} if $\mathbf{r}^{-1} = \mathbf{r}_{21}$, where $\mathbf{r}_{21}(a,b) \coloneq \mathbf{r}(b,a)$.
\end{definition}
\noindent Any $r$-form can be shown to satisfy $\mathbf{r}(1,a) = \mathbf{r}(a,1) = \epsilon(a)$ \cite{QGAR}. Also, one can recognise these properties to be the dual analogue of \eqref{eq:quasitriangularityconditions}. It can be noted that the factorisation property \eqref{eq:factorisationrform} of the $r$-form corresponds with the associativity of the algebra. In particular, since $(ab)c = a(bc)$, the pictorially depicted braiding $(ab)c \to c(ab)$ using the $r$-form in \eqref{eq:braidingrform} can be factorised into the braiding of $bc \to cb$ and $ac\to ca$. We emphasise here the necessity of associativity. In Section~\ref{sec:deformationFRTalgebra}, we will (weakly) break the associativity of the algebra, such that the factorisation property will not hold anymore in general. This breaking of the associative, and therefore non-factorisability, will be the main mechanism that gives rise to long-range deformations. Lastly, by combining $\eqref{eq:braidingrform}$ and $\eqref{eq:factorisationrform}$, it follows that the $r$-form satisfies the dual Yang-Baxter equation
\begin{equation}\label{eq:dualYangBaxtereq}
    \sum \mathbf{r}(a_{(1)}, b_{(1)})\mathbf{r}(a_{(2)}, c_{(1)})\mathbf{r}(b_{(2)}, c_{(2)}) = \sum \mathbf{r}(b_{(1)}, c_{(1)})\mathbf{r}(a_{(1)}, c_{(2)})\mathbf{r}(a_{(2)}, b_{(2)}),
\end{equation}
for $a,b,c\in \A$.\\

Of course, any FRT-bialgebra $\A(\R)$ is cotriangular by construction, where the universal $r$-form is simply given by $\mathbf{r}\left(\T_1(u), \T_{2}(v)\right) \coloneq \R_{12}(u,v)$. Here, the cotriangularity (as opposed to coquasitriangularity) is a consequence of the braiding unitarity of the $R$-matrix. The evaluation of the $r$-form for products of elements in $\A(\R)$ can then be obtained through the factorisation property \eqref{eq:factorisationrform} of the $r$-form, for which the Yang-Baxter equation \eqref{eq:matrixYBE} provides a consistency condition to the associativity of the underlying algebra. Here, a product of elements of $\A(\R)$ in $\mathbf{r}$ will correspond to the coproduct in the universal $R$-matrix in the dual picture, and we will use the notation\footnote{Here, $\overrightarrow{\prod}_{i=1}^{N}a_i \coloneq a_{1}\cdots a_N$, while $\overleftarrow{\prod}_{i=1}^{N}a_i \coloneq a_{N}\cdots a_1$.}
\begin{equation}\label{eq:Rmatrixfactorisation}
    \begin{split}
        \R_{(1,..., N)(N+1,..., M)}(u_1,..., u_{M})&\coloneq \mathbf{r}\left(\T_1(u_1)\cdots \T_{N}(u_N), \T_{N+1}(u_{N+1})\cdots \T_{M}(u_{M})\right)\\
        &= \overset{N}{\underset{i=1}{\overrightarrow{\prod}}}\overset{M}{\underset{\ j=N+1}{\overleftarrow{\prod}}}\R_{ij}(u_i,u_j),
    \end{split}
\end{equation}
for $0<N< M$. On the right, one can recognise the same formula as the coproduct of a universal $\R$-matrix of a quasitriangular bialgebra, where the brackets in the subscript of $\R_{(1,\cdots, N)(N+1, ..., M)}$ denote the $(N-1)$-th and $(M-N-1)$-th coproduct on the first and second coordinate of $\R$, respectively. In particular, in the dual picture for the universal $R$-matrix, we have $\R_{(1,...,N)(N+1,...,M)} \coloneq (\Delta^{N-1}\otimes \Delta^{M-N-1})(\R)$.\\

Importantly, the universal $r$-form naturally induces a homomorphism from a bialgebra $\A$ to its dual. Indeed, consider the following definition:
\begin{definition}
    Let $\A$ be a cotriangular bialgebra and $a\in \A$. The \textit{$\ell$-functional} is the map $\bm{\ell}(a) \colon \A \to \C$ given by $\bm{\ell}(a) \coloneq \mathbf{r}(a,\cdot)$.\footnote{Note that if $\A$ is coquasitriangular, there also exists another $\ell$-functional $\bm{\ell}^- \coloneq \mathbf{r}^{-1}(\cdot, a)$, which coincides with $\bm{\ell}$ in the cotriangular case. Also, we use a different convention from, for example \cite{QGAR}, where the element $a$ is inserted in the second coordinate of $\mathbf{r}$. Instead, we use a convention which better aligns with most physics literature.}
\end{definition}
\noindent Now let $\U \subseteq \A^*$ be the bialgebra generated by these functionals $\bm{\ell}(a)$ for $a\in \A$. In that case, the map $\bm{\ell} \colon a\mapsto \bm{\ell}(a)$ is a bialgebra homomorphisms from $\A^{\mathrm{cop}}$ to $\U$, where $\A^{\mathrm{cop}}$ is the coopposite bialgebra of $\A$, i.e. the bialgebra with coproduct $\Delta^\mathrm{cop} \coloneq \sigma\circ \Delta$, where $\sigma$ the flip operator. Indeed, it follows from \eqref{eq:factorisationrform} that
\begin{equation}
    \begin{split}
        \langle c, \bm{\ell}(ab)\rangle &= \sum \bm{\ell}(a)(c_{(1)})\bm{\ell}(b)(c_{(2)}) = \langle \Delta(c), \bm{\ell}(a)\otimes \bm{\ell}(b)\rangle = \langle c, \bm{\ell}(a) \bm{\ell}(b)\rangle\\
        \langle bc, \bm{\ell}(a)\rangle &= \sum \bm{\ell}(a_{(1)})(c)\bm{\ell}(a_{(2)})(b) = \langle b\otimes c, \bm{\ell}(\Delta^{\mathrm{cop}}(a))\rangle = \langle b\otimes c, \Delta(\bm{\ell}(a))\rangle,
    \end{split}
\end{equation}
for $a,b,c\in \A$, where we used the notation $\langle b, \bm{\ell}(a)\rangle \coloneq \bm{\ell}(a)(b) = \mathbf{r}(a,b)$.\\

In the context of the FRT-bialgebra $\A(\R_\rho)$, as described in Section~\ref{sec:dualdescription}, under the dual bracket \eqref{eq:dualpairingYangian}, the algebra $\U$ can be identified with a subalgebra of a certain completion of $\mathcal{Y}_\hbar(\mathfrak{g})$, which we denote by $\hat{\mathcal{Y}}_\hbar(\mathfrak{g})$. Therefore, the function $\bm{\ell}(a) \coloneq (\langle a, \cdot\rangle\otimes \id)(\R)$ gives a bialgebra homomorphism
\begin{equation}
    \bm{\ell} \colon \A(\R_\rho)^\mathrm{cop}\to \hat{\mathcal{Y}}_\hbar(\mathfrak{g}).
\end{equation}
In particular, any representation of $\mathcal{Y}_\hbar(\mathfrak{g})$ induces a representation of $\A(\R_\rho)^{\mathrm{cop}}$ through the $\ell$-functional. Of course, this could now, in principle, be a representation different from $\rho$. However, for our purposes, we will choose the representation to be the same. In that case, the elements $t_{ij}(u) \in \A(\R_\rho)$ correspond to matrix elements of $\rho_u$, such that evaluation of $\bm{\ell}(a)$ at these matrix elements $t_{ij}(u)$ corresponds to the matrix elements of the matrix $\rho_u(\bm{\ell}(a))$ in $\mathrm{End}(V)$. For a general FRT-bialgebra $\A(\R)$, one defines the evaluation representation $(\rho_u,V)$ of $\A(\R)^\mathrm{cop}$ to be given by $\rho_u(a) \coloneq \bm{\ell}(a)(\T(u)) \in \mathrm{End}(V)$ for $a\in \A(\R)$.\\

By restricting to the elements in the matrix corepresentation $\T(u)$, one can now define an operator $\L(u,v) \in \mathrm{End}(V)\otimes \mathrm{End}(V)$, called the \textit{Lax operator}, given by
\begin{equation}
    \mathscr{L}_{a1}(u,v) \coloneq \rho_v(\T_a(u)) = \bm{\ell}(\T_{a}(u))(\T_1(v)) = \mathbf{r}(\T_a(u), \T_1(v)) = \R_{a1}(u,v),
\end{equation}
which just corresponds to the evaluation of the representation $\rho_v$ at the matrix elements of $\T_a(u)$, where $a$ denotes a coordinate on the \textit{auxiliary space} of the matrix corepresentation, while $1$ is a coordinate on the actual representation space $V$. Note here that the Lax operator therefore corresponds to a representation of $\A(\R)^{\mathrm{cop}}$ that is induced from the $\ell$-functional. Of course, this Lax operator is now simply equal to the $R$-matrix. However, we will see in Section~\ref{sec:longrangedeformations} that this is no longer true for the long-range deformations. For this paper, we will solely be interested in evaluation-zero representations, i.e. we will set $v=0$, and write $\L(u)\coloneq \L(u,0)$. Since $\mathscr{L}$ corresponds to a representation of $\A(\R)^{\mathrm{cop}}$, it satisfies the \textit{RLL-relation}
\begin{equation}
    \R_{ab}(u,v)\mathscr{L}_{a1}(u)\mathscr{L}_{b1}(v) = \mathscr{L}_{b1}(v)\mathscr{L}_{a1}(u)\R_{ab}(u,v),
\end{equation}
which ensures the integrability of the associated spin chain \cite{Faddeev1982IntegrableModels} (see also Section~\ref{sec:commutingcharges}).

\subsection{The Bethe subalgebra and charge densities}\label{sec:commutingcharges}
The RTT relation \eqref{eq:RTTrelations} for the FRT-bialgebra $\A(\R)$ gives rise to a large abelian subalgebra, which corresponds to an (infinite) tower of commuting charges in the associated physical spin chain for a given representation. In particular, the trace over $\mathrm{End}(V)$ in the matrix corepresentation $\T(u)$ \eqref{eq:matrixcorepresentationT} defines the elements $t^{(-k)}\in \A(\R)$ given by
\begin{equation}
    t(u) \coloneq \text{tr}_{\mathrm{End}(V)} \T(u) = 1+\sum_{k=0}^{\infty} t^{(-k-1)} u^k \in \A(\R)[\![u]\!].
\end{equation}
By taking the trace over the RTT-equation \eqref{eq:RTTrelations}, it then immediately follows that all these elements commute, i.e. $[t(u), t(v)] = 0$, or $[t^{(-k)}, t^{(-\ell)}] = 0$ for $k,\ell \geq 1$ \cite{STF1979QuantumInverseScattering, Sklyanin1982QuantumversionInverseScattering, TakhtadshanFaddeev1997}. The elements $t^{(-k)}$ therefore generate a (large) abelian subalgebra of $\A(\R)$, called the \textit{Bethe subalgebra} \cite{Nazarov1996BetheSubalgebras}. To make contact with the \textit{local charges} that are generated by the Bethe subalgebra on specific spin chains, one often defines the elements
\begin{equation}\label{eq:definitioncharges}
    q_k \coloneq \frac{d^{k-2}}{du^{k-2}}\left(t(u)^{-1}\frac{d}{du}t(u)\right)\bigg|_{u=0} \in \A(\R) \quad \text{for } k\geq 2,
\end{equation}
such that the Bethe subalgebra is now generated by $t^{(-1)}$ and $q_k$ for $k\geq 2$, and the elements $q_k$ are called the \textit{charges}.\footnote{Often, one also considers the first charge $q_1 = \log(1+t^{(-1)})$, where the logarithm should be interpreted as the Taylor expansion of $\log(1+x)$ around $x = 0$. Note that these charges are only well-defined in the completion $\A^\circ(\R) $ of the FRT-bialgebra as given under Definition~\ref{def:FRTalgebradef}.} In the context of spin chains, one of these charges in the corresponding representation is then defined to be the \textit{Hamiltonian} $H$ of a system. In most cases, this Hamiltonian is specifically identified with the second charge $q_2$. Since the FRT-bialgebra $\A(\R_\rho)$ corresponds to the dual of a Drinfeld-Jimbo algebra, this also implies that the Hamiltonian is given by an element of the dual of the quantised universal enveloping algebra. For example, for the Yangian, one has $H \in \mathcal{Y}_\hbar (\mathfrak{g})^*$ \cite{Drinfled1988}. Importantly, the existence of a large tower of pair-wise commuting charges helps with finding the spectral decomposition of the Hamiltonian, which lies at the root of the quantum inverse scattering method \cite{STF1979QuantumInverseScattering, TakhtadshanFaddeev1997, faddeev1984integrable,  Sklyanin1982QuantumversionInverseScattering, Faddeev1982IntegrableModels, sklyanin1982some, kulish1983quantum, Sklyanin1992quantuminversescatteringmethod, Korepin1993quantum, Faddeev1995QISMHistory}.\\

Let us now consider these charges in the representation given by an $L$-fold tensor product of evaluation $0$ representations, i.e. $\rho_0^{\otimes L}$, on $V^{\otimes L}$. Firstly, note that the evaluation of $\T(u)$ under this representation, which we denote by $\T^{\rho}(u)$, is given by
\begin{equation}\label{eq:monodromymatrix}
    \T_a^{\rho}(u) \coloneq\rho_0^{\otimes L}(\T_a(u)) = \mathbf{r}(\T_a(u), \T_1(0)\cdots \T_L(0)) = \L_{a,L}(u)\L_{a,L-1}(u)\cdots \L_{a,1}(u),
\end{equation}
where $1,...,L$ denote the coordinates on the \textit{physical space} $V^{\otimes L}$, while $a$ denotes the coordinate on the auxiliary space $V$ of the matrix corepresentation $\T(u)$. The matrix $\T_a^{\rho}(u) \in \mathrm{End}(V)\otimes \mathrm{End}(V^{\otimes L})$ is the \textit{monodromy matrix}. Similarly, one can consider the representation of the Bethe subalgebra, which is obtained by tracing the monodromy matrix over the auxiliary space. That is, one has
\begin{equation}\label{eq:transfermatrix}
    t^\rho(u) \coloneq \rho_0^{\otimes L}(t(u)) = \text{tr}_a \T^\rho_a(u),
\end{equation}
which is the \textit{transfer matrix}. Since also in this representation one has $[t^\rho(u), t^\rho(v)] = 0$, the charges $q_k$ under this representation correspond to a set of commuting charges on $V^{\otimes L}$, and therefore give rise to a quantum integrable system. Moreover, since the $R$-matrix is regular, these charges are \textit{local}. In particular, since $\L_{ak}(0) = \P_{ak}$, it can be noted that $t^\rho(0) = \P_{1,L}\cdots \P_{1,2}$ is equal to the \textit{shift operator} $\U \colon e_{1}\otimes e_2\otimes \cdots \otimes e_L \mapsto e_L\otimes e_1\otimes \cdots \otimes e_{L-1}$ on $V^{\otimes L}$, such that, for example, the second charge is given by  (see e.g. \cite{deLeeuw2020YangBaxterandBoost})
\begin{equation}
    \label{eq:localHamiltonian}
    \mathbb{Q}_2 \coloneq \rho_0^{\otimes L}(q_2) = \U^{-1}\sum_{k=1}^{L}\text{tr}_a\left(\L_{a,L}(0)\cdots \dot \L_{a,k}(0)\cdots \L_{a,1}(0)\right) =\sum_{k=1}^{L}\mathfrak{q}^{(2)}_{k,k+1},
\end{equation}
where $\mathfrak{q}^{(2)} \coloneq \L^{-1}(0)\dot {\L}(0)\in \mathrm{End}(V^{\otimes 2})$ denotes the \textit{charge density}, with $\dot{\L}(u) \coloneq \frac{d}{du}\L(u)$. Moreover, $\mathfrak{q}^{(2)}_{i,j}$ denotes the action of the charge density on the $i$-th and $j$-th components of the tensor product $V^{\otimes L}$. The second charge $\mathbb{Q}_2$ in this representation is therefore simply a sum of the local densities $\mathfrak{q}^{(2)}$ of range $2$ over the full spin chain $V^{\otimes L}$. \\

In fact, it is more generally true that the charges $\mathbb{Q}_k \coloneq \rho_0^{\otimes L}(q_k)$ can be written as a sum over local densities $\mathfrak{q}^{(k)}\in \mathrm{End}(V^{\otimes k})$ of range $k$  \cite{deLeeuw2020YangBaxterandBoost},\footnote{Note here that the $k$-th charge will be denoted with a subscript of $k$, i.e. $q_k$, while the $k$-th charge density is denoted with a superscript of $k$ in \texttt{mathfrak}, i.e. $\mathfrak{q}^{(k)}$.} where it should be noted that such local charge densities are not unique. Indeed, one may add terms to the charge density that vanish when summed over the whole spin chain. We then say that two charge densities $\mathfrak{q}^{(k)}, \tilde{\mathfrak{q}}^{(k)} \in \mathrm{End}(V^{\otimes k})$ are equivalent, denoted by $\mathfrak{q}^{(k)}\sim \tilde{\mathfrak{q}}^{(k)}$, if they generate the same (global) charges on any spin chain of length $L$, i.e. if $\sum_{n=1}^{L} \mathfrak{q}^{(k)}_{n,...,n+k-1} = \sum_{n=1}^{L} \tilde{\mathfrak{q}}^{(k)}_{n,...,n+k-1}$ for all $L\geq k$. In particular, we have the following definition:
\begin{definition}\label{def:equivalencelocalchargedensities}
    Let $k\geq 2$. Two charge densities $\mathfrak{q}^{(k)}, \tilde{\mathfrak{q}}^{(k)} \in \mathrm{End}(V^{\otimes k})$ are said to be equivalent, denoted by $\mathfrak{q}^{(k)} \sim \tilde{\mathfrak{q}}^{(k)}$, if there exists a matrix $\mathfrak{m}\in \text{End}(V^{\otimes (k-1)})$ such that $\tilde{\mathfrak{q}}^{(k)} = \mathfrak{q}^{(k)} + I_V\otimes \mathfrak{m} - \mathfrak{m}\otimes I_V$.
\end{definition}
\noindent Note here that, by definition, we can only have equivalence between charge densities $\mathfrak{q}^{(k)}$ and $\tilde{\mathfrak{q}}^{(\ell)}$ if $k = \ell$.\\

We remark here that the origin of the locality of the charges is two-fold. Firstly, it is a consequence of the fact that the $R$-matrix is regular. Secondly, the charges of $\mathbb{Q}_k$ are of range $k$ precisely because of the factorisation property \eqref{eq:factorisationrform} of the $r$-form, which allows the monodromy matrix \eqref{eq:monodromymatrix} to be given in terms of the ``range-two'' Lax operators. In particular, all the higher charge densities $\mathfrak{q}^{(k)}$ are also sums and commutants of range-two terms. For example, the third charge density is given by (see e.g. \cite{deLeeuw2020YangBaxterandBoost})
\begin{equation}\label{eq:thirdchargedensity}
    \mathfrak{q}^{(3)}_{123} = [\mathfrak{q}^{(2)}_{23}, \mathfrak{q}^{(2)}_{12}]+\frac{d}{du}(\L^{-1}_{12}(u)\dot{\L}_{12}(u))\big|_{u=0}.
\end{equation}
The fact that the origin of the local charges comes from regularity and the factorisation property of the $r$-form now already gives us a hint towards what the correct description for the long-range deformations should be. In particular, one expects that long-range deformations will appear if the Lax operator cannot fully factorise into a product of range-two terms, which, in light of the discussion in Section~\ref{sec:coquasitriangularityandLfunctionals}, happens if the associativity of the FRT-bialgebra is explicitly broken. Indeed, as we will see in the Section~\ref{sec:longrangedeformations}, the breaking of associativity is exactly what gives rise to the long-range deformations.

\subsection{Algebraic charge densities}\label{sec:algebraicchargedensities}
Before we can describe the long-range deformations of the FRT-bialgebras, we will introduce some algebra elements that we will call the \textit{algebraic charge densities}, and which are given in terms of derivations of the $R$-matrix. In fact, these objects can already be used to give us more information about the FRT-bialgebra and the corresponding integrable model as well. In particular, as we will see, these algebraic charge densities can be thought of as algebraic generalisations of the ordinary charge densities. In particular, we will give a conjecture for the explicit expressions of the charge densities using the algebraic charge densities. Moreover, these objects will be used to derive the higher-order Sutherland equations in Section~\ref{sec:higherorderSutherlandequation} and in the description of the bialgebra deformations in Section~\ref{sec:longrangedeformations}.\\

Firstly, let us consider the differential operator $\D \colon \A(\R)^*\to \A(\R)^*$ given by
\begin{equation}
    \label{eq:definitionderivative}
    (\D f)\left(\T_1(u_1)\cdots \T_{N}(u_N)\right) \coloneq \sum_{n=1}^{N}\frac{d}{dv_n}f\left(\T_1(v_1)\cdots \T_{N}(v_N)\right)\bigg|_{(v_1,...,v_N)=(u_1,..., u_N)}
\end{equation}
for $f\in \A(\R)^*$. In the context of the dual pairing \eqref{eq:dualpairingYangian} as described in Section~\ref{sec:dualdescription}, this derivation corresponds to an operator on the quantised universal enveloping algebra $\mathcal{Y}_\hbar(\mathfrak{g})$. In particular, one may define the corresponding operator $\D \colon \mathcal{Y}_\hbar(\mathfrak{g})\to \mathcal{Y}_\hbar(\mathfrak{g})$ given by $\D(x) = \frac{d}{du}\B_u(x) \big|_{u=0}$ for $x\in  \mathcal{Y}_\hbar(\mathfrak{g})$, where $\B_u$ is the boost automorphism. It can be easily verified that $\D$ is indeed a \textit{derivation}, i.e. $\D(xy) = \D(x)y + x\D(y)$. Moreover, it is a \textit{coderivation}, as under the coproduct of $\mathcal{Y}_\hbar(\mathfrak{g})$, one has $\Delta(\D(x)) = (\D\otimes \id + \id \otimes \D)\Delta(x)$, which, after repeated coproducts, gives the same formula as \eqref{eq:definitionderivative} in the dual picture. We will write $\D_i$ for the derivative on the $i$-th coordinate. Note here that, since the universal $R$-matrix for the Yangian is of difference-form, it satisfies $(\D_1 + \D_2)(\R) = 0$. For general $R$-matrices that are not of difference form, this obviously does not have to hold.\footnote{\label{footnote:derivation}Remark that $\D$ is only a well-defined map $\A(\R)^*$ to $\A(\R)^*$ as bialgebras if the $R$-matrix is of difference form such that it commutes with the derivation. Otherwise, $\D$ should be regarded as a derivation on $\A(\R)^*$ as an algebra. Note that it is well-defined as a derivation over a bialgebra when restricted to the subspace of functions over $\A_0\subseteq \A(\R)$ generated by elements $t_{ij}(0)$ (see Section~\ref{sec:deformationFRTalgebra}).}\\

We can now use the derivation on $\A(\R)^*$ to define the following bilinear maps:
\begin{definition}
    For $k\geq 2$, the \textit{algebraic charge densities} are the bilinear maps \linebreak $\Q^{(k)}, \tilde \Q^{(k)}\colon  \A(\R)\times \A(\R)\to \C$ given by 
    \begin{equation*}
        \Q^{(k)} \coloneq \D^{k-2}_1\left(\mathbf{r}^{-1}\ast\D_1(\mathbf{r})\right) \quad  \text{and} \quad \tilde\Q^{(k)} \coloneq -\D^{k-2}_2\left(\mathbf{r}^{-1}\ast\D_2(\mathbf{r})\right),
    \end{equation*}
    where $\D_1 \coloneq \D\otimes \id$ and $\D_2 \coloneq \id \otimes \D$.
\end{definition}
\noindent Note that $\tilde \Q^{(k)}$ is defined with an additional minus sign, which will be argued for later. In the dual picture as described in Section~\ref{sec:dualdescription}, the above bilinear maps correspond to the elements $(-1)^{m-1}\D^{k-2}_i\left(\R^{-1}\D_m(\R)\right)$ for $m = 1,2$, where $\R$ is the universal $R$-matrix as given above \eqref{eq:quasitriangularityconditions}. Since $\mathbf{r}(1,a) = \epsilon(a) = \mathbf{r}(a,1)$, it can also be remarked that $\Q^{(k)}(a,1)= 0 = \Q^{(k)}(1,a)$ and $\tilde \Q^{(k)}(a,1)= 0 = \tilde \Q^{(k)}(1,a)$ for all $a\in \A(\R)$. By evaluating the algebraic charge densities at products of the matrix corepresentations $\T_i(u_i)$, they define new matrices given by 
\begin{equation}
    \label{eq:Qkalgebraicchargedensity}
    \begin{split}
        \Q^{(k)}_{(1,...,N)(N+1,...,M)}(u_1,..., u_M) &\coloneq \Q^{(k)}\left(\T_1\cdots \T_N, \T_{N+1}\cdots \T_M\right) \in \mathrm{End}(V^{\otimes M}),\\
        \tilde \Q^{(k)}_{(1,...,N)(N+1,...,M)}(u_1,..., u_M) &\coloneq \tilde \Q^{(k)}\left(\T_1\cdots \T_N, \T_{N+1}\cdots \T_M\right)\in \mathrm{End}(V^{\otimes M})
    \end{split}
\end{equation}
for $0<N<M$, where we used the shorthand notation $\T_i \coloneq \T_i(u_i)$ on the right-hand side. In the dual picture, these matrix functions correspond to specific evaluation representations of coproducts of the elements $\Q^{(k)} \coloneq \D_1^{k-2}(\R^{-1}\D_1(\R))$ and $\tilde \Q^{(k)} \coloneq -\D^{k-2}_2(\R^{-1}\D_2(\R))$. From now on, we will be less explicit about the correspondence between these matrix functions and elements in a representation of the corresponding dual algebra. Moreover, we will often suppress the explicit dependence on the spectral parameters to simplify notation.\\

In order to get some intuition behind the algebraic charge densities, let us first restrict to $k=2$, i.e. the elements $\Q^{(2)}$ and $\tilde \Q^{(2)}$. Using braiding unitarity $\R_{12}\R_{21} = 1$, it can be noted that
\begin{equation}
    \tilde \Q_{12}^{(2)}  = -\R_{12}^{-1}\D_2(\R_{12}) = -\R_{21}\D_2(\R_{21}^{-1}) = \R_{12}^{-1}\Q^{(2)}_{21}\R_{12},
\end{equation}
which shows that $\Q^{(2)}$ and $\tilde \Q^{(2)}$ are in fact related by conjugation. Moreover, since $\Q^{(2)}_{12} = \R^{-1}_{12}\D_1(\R_{12})$, it immediately follows that $\Q^{(2)}_{12}(\vec 0) = \mathfrak{q}^{(2)}_{12} = \tilde \Q^{(2)}_{12}(\vec 0)$, where $\vec 0 = (0,0)$ and the charge density $\mathfrak{q}^{(2)}_{12}$ was defined under \eqref{eq:localHamiltonian}. This gives the reason why we refer to $\Q^{(2)}$ and $\tilde \Q^{(2)}$ as algebraic charge densities; they restrict to the actual charge densities that appear in the tower of commuting charges, but are now defined for arbitrary spectral parameters $u_i$ and arbitrary coproducts. Note that, in the case the $R$-matrix $\R$ is of difference form, such that $(\D_1+\D_2)(\R) = 0$, we have $\tilde \Q = \Q$, thus motivating the additional minus sign in the definition of $\tilde \Q$.\\

Let us now consider the coproducts of the algebraic charge densities $\Q^{(2)}$ and $\tilde \Q^{(2)}$. Firstly, the factorisation property of the $r$-form \eqref{eq:factorisationrform}, or dually, the factorisation property of the universal $R$-matrix \eqref{eq:quasitriangularityconditions}, allows us to write the coproducts of $\Q^{(2)}$ and $\tilde \Q^{(2)}$ in terms of the ``two-legged'' algebraic charge density. In particular, we have
\begin{equation}
    \label{eq:coproductCharge2}
    \begin{alignedat}{2}
        \Q^{(2)}_{1(23)} &= \R^{-1}_{1(23)}\D_1(\R_{1(23)}) = \R^{-1}_{12}\R^{-1}_{13}\D_1(\R_{13}\R_{12}) &=&\ \R^{-1}_{12}\Q^{(2)}_{13}\R_{12} + \Q^{(2)}_{12},\\
        \Q^{(2)}_{(12)3} &= \R^{-1}_{(12)3}(\D_1 + \D_2)(\R_{(12)3}) = \R^{-1}_{23}\R^{-1}_{13}(\D_1+\D_2)(\R_{13}\R_{23})\ &=&\ \R^{-1}_{23}\Q^{(2)}_{13}\R_{23} + \Q^{(2)}_{23},
    \end{alignedat}
\end{equation}
and similarly
\begin{equation}
    \tilde \Q^{(2)}_{1(23)} = \R^{-1}_{12}\tilde \Q^{(2)}_{13}\R_{12} + \tilde \Q^{(2)}_{12}\quad \text{and} \quad \tilde \Q^{(2)}_{(12)3} = \R^{-1}_{23}\tilde \Q^{(2)}_{13}\R_{23} + \tilde \Q^{(2)}_{23},
\end{equation}
where we recall the notation for the brackets in the subscript for the coproduct as introduced under \eqref{eq:Rmatrixfactorisation}. Note here that the regularity of the $R$-matrix implies that the coproducts of the algebraic charge density are given by $\Q^{(2)}_{1(23)}(\vec 0) = \tilde \Q^{(2)}_{(12)3}(\vec 0) = \mathfrak{q}^{(2)}_{12} + \mathfrak{q}^{(2)}_{23}$, such that the coproducts evaluated at trivial spectral parameters gives a sum over the charge densities.\\

We can now generalise the formula for the coproduct of the algebraic charge densities to arbitrary values of $k\geq 2$. In particular, as shown in Appendix~\ref{sec:proofsforalgebraicchargedensities}, one has the following result:
\begin{proposition}\label{prop:coproductChargeDensities}
    For $k\geq 2$, the coproduct relations for the algebraic charge densities $\Q^{(k)}$ and $\tilde \Q^{(k)}$ are given by
    \begin{equation*}
        \begin{split}
            \Q^{(k)}_{1(23)} &= \Q^{(k)}_{12} + \R^{-1}_{12}\Q^{(k)}_{13}\R_{12}\\
            &\quad + \sum_{n=2}^{k-1}\ \sum_{\ell_1+\cdots+\ell_n = k-1}c^{k-1}_{\ell_1,...,\ell_n}\left[\left[\cdots \left[\left[\R^{-1}_{12}\Q^{(\overline{\ell_n})}_{13}\R_{12}, \Q^{(\overline{\ell_1})}_{12}\right],\Q^{(\overline{\ell_2})}_{12}\right],\cdots \right], \Q^{(\overline{\ell_{n-1}})}_{12}\right],
        \end{split}
    \end{equation*}
    and
    \begin{equation*}
        \begin{split}
            \tilde \Q^{(k)}_{(12)3} &= \tilde \Q^{(k)}_{23} + \R^{-1}_{23}\tilde \Q^{(k)}_{13}\R_{23}\\
            &\quad + \sum_{n=2}^{k-1}\ \sum_{\ell_1+\cdots+\ell_n = k-1}c^{k-1}_{\ell_1,...,\ell_n}\left[\tilde \Q^{(\overline{\ell_{n-1}})}_{23}, \left[\tilde \Q^{(\overline{\ell_{n-2}})}_{23},\left[\cdots , \left[\tilde \Q_{23}^{(\overline{\ell_1})}, \R^{-1}_{23}\tilde \Q_{13}^{(\overline{\ell_n})}\R_{23}\right]\cdots\right]\right]\right],
        \end{split}
    \end{equation*}
    where $\overline{\ell_i} = \ell_i+1$ and $\ell_i \geq 1$ and $c^{k-1}_{\ell_1,...,\ell_n} \coloneq \prod_{m=2}^n\dbinom{\sum_{i=1}^{m}\ell_i  - 1}{\sum_{i=1}^{m-1}\ell_i}$.
\end{proposition}
\noindent We remark that (very) similar coproduct relations may be derived in an analogous way as in Appendix~\ref{sec:proofsforalgebraicchargedensities} for $\Q^{(k)}_{(12)3}$ and $\tilde \Q^{(k)}_{1(23)}$.\footnote{Remark that these coproducts are then only well-defined as ``true'' coproducts. i.e. are intertwined by the $R$-matrix, if the $R$-matrix is of difference form (see footnote~\ref{footnote:derivation}). The coproducts as given in Proposition~\ref{prop:coproductChargeDensities} are well-defined even if the $R$-matrix is not of difference form, since here the derivation and the coproduct in $\Q^{(k)}$ and $\tilde{\Q}^{(k)}$ act on the two different legs.} However, we will not need the explicit expressions for this paper. \\

Let us now consider the explicit case for $k=3$, in which case it follows from Proposition~\ref{prop:coproductChargeDensities} that the coproducts are given by
\begin{equation}
    \label{eq:coproductsQ3}
    \begin{split}
        \Q^{(3)}_{1(23)} &= \R_{12}^{-1}\Q^{(3)}_{13}\R_{12} + \Q^{(3)}_{12} + [\R^{-1}_{12}\Q^{(2)}_{13}\R_{12}, \Q^{(2)}_{12}],\\
        \tilde \Q^{(3)}_{(12)3} &= \R^{-1}_{23}\tilde \Q^{(3)}_{13}\R_{23} + \tilde \Q^{(3)}_{23} + [\tilde \Q^{(2)}_{23}, \R^{-1}_{23}\tilde \Q^{(2)}_{13}\R_{23}].
    \end{split}
\end{equation}
If we now set all the spectral parameters to $0$, in which case the $R$-matrix $\R$ reduces to the permutation operator, it can be directly seen from \eqref{eq:thirdchargedensity} that
\begin{equation}\label{eq:charge3representatives}
    \left(\Q^{(3)}_{1(23)}(\vec 0) - \Q^{(3)}_{23}(\vec 0)\right) \sim \mathfrak{q}^{(3)}_{123} \sim \left(\tilde \Q^{(3)}_{(12)3}(\vec 0) - \tilde \Q^{(3)}_{12}(\vec 0)\right),
\end{equation}
where $\sim$ is the equivalence relation for charge densities as given in Definition~\ref{def:equivalencelocalchargedensities}. As already suggested in the naming, the algebraic charge densities for $k=3$ are therefore directly related to the charge densities for the commuting charges as described in Section~\ref{sec:commutingcharges}. Remark here that \eqref{eq:charge3representatives} gives two different representatives for the third charge density.\\

It is now expected that a similar formula as \eqref{eq:charge3representatives} holds true for all higher-order charge densities as well. In particular, we have the following conjecture:
\begin{conjecture}\label{conj:chargedensities}
    For $k\geq 2$, the charge densities $\mathfrak{q}^{(k)}_{1,...,k}$ are given by
    \begin{equation*}
        \left(\Q^{(k)}_{1(2,...,k)}(\vec 0) - \Q^{(k)}_{2(3,...,k)}(\vec 0)\right) \sim \mathfrak{q}^{(k)}_{1,...,k} \sim \left(\tilde \Q^{(k)}_{(1,...,k-1)k}(\vec 0) - \tilde \Q^{(k)}_{(1,...,k-2)k-1}(\vec 0)\right),
    \end{equation*}
with $\sim$ the equivalence relation as given in Definition~\ref{def:equivalencelocalchargedensities}.
\end{conjecture}
\noindent As of this moment, we do not have a proof for this conjecture. Proving the result would require one to explicitly calculate each charge density from the transfer matrix \eqref{eq:transfermatrix} or using the boost formalism \cite{TetelmanLorentzGroupforTwoDimensional, Links2001LadderOperator} (see also e.g. \cite{Loebbert2016LecturesonYangianSymmetry, deLeeuw2019Classifyingintegrable, deLeeuw2020YangBaxterandBoost}) and then compare the result to the equation of the coproducts of charge densities using Proposition~\ref{prop:coproductChargeDensities}. Calculating each of these charge densities explicitly seems to be a rather tedious endeavour, which is why we have not attempted to prove the conjecture at this point. We do, however, have strong evidence for the conjecture to be true. For example, we have explicitly verified the equality in \texttt{Mathematica} for specific $R$-matrices/models. More importantly, it is shown in Corollary~\ref {cor:commutingcharges} that the charges that are defined through the above conjecture all commute with the transfer matrix. Moreover, the long-range deformations that we will describe in Section~\ref{sec:longrangedeformations} are expected to exactly correspond to the family of long-range deformations generated by the boost and bilocal operators (see Section~\ref{sec:classificationLongrangeDeformations} or \cite{Bargheer2008BoostingNearestNeighbour, Bargheer2009}). This equivalence between the two deformation equations on the level of charge densities is only true if Conjecture~\ref{conj:chargedensities} holds. We remark that none of the results in the following sections depend on this conjecture. Of course, assuming the conjecture to be true, Proposition~\ref{prop:coproductChargeDensities} then allows one to give a recurrence relation for all the charge densities in terms of lower charge densities:
\begin{corollary}
    The charge densities $\mathfrak{q}^{(k)}$ for $k\geq 2$ satisfy the recurrence relation
    \begin{equation*}
        \mathfrak{q}^{(k)}_{1,...,k}  = \Q^{(k)}_{12} + \sum_{n=2}^{k-1}\ \sum_{\ell_1+\cdots+\ell_n = k-1}c^{k-1}_{\ell_1,...,\ell_n}\left[\left[\cdots \left[\left[\mathfrak{q}^{(\overline{\ell_n})}_{2,3,...,\overline{\ell_n}+1}, \Q^{(\overline{\ell_1})}_{12}\right],\Q^{(\overline{\ell_2})}_{12}\right],\cdots \right], \Q^{(\overline{\ell_{n-1}})}_{12}\right],
    \end{equation*}
    where $\overline{\ell_i} = \ell_{i}+1$, $\ell_i \geq 1$, $\Q^{(\overline{\ell})}_{12} \coloneq \Q^{(\overline{\ell})}_{12}(\vec 0)$ and $c^{k-1}_{\ell_{1},...,\ell_{n}}$ as given in Proposition~\ref{prop:coproductChargeDensities}.
\end{corollary}

The appearance of the difference $\Q^{(k)}_{1(2,...,k)}(\vec 0) - \Q^{(k)}_{2(3,...,k)}(\vec 0)$ in the expression for the charge density in Conjecture~\ref{conj:chargedensities} is not accidental. In fact, it is related to an important property of the algebraic charge densities that will be used many times in the various proofs in this paper. In order to describe this property, we will first need to introduce some notation. Namely, in the following, we will often encounter situations where we have a mixture of trivial and non-trivial insertions of spectral parameters. In order to explicitly make this distinction and simplify notation, we will denote a coordinate with a non-trivial spectral parameter by an underline. That is, $\T_{\underline 1} \coloneq \T_{\underline{1}}(u_1)$, while $\T_1 \coloneq \T_{1}(0)$. Moreover, for a vector $\vec a \coloneq (a_1,..., a_{|\vec a|})$ we write $\T_{\underline{\vec a}}(u_{\vec a}) \coloneq \T_{\underline{a}_1}(u_{1})\T_{\underline{a}_2}(u_2)\cdots \T_{\underline{a}_{|\vec a|}}(u_{|\vec a|})$. We then consider the specific elements
\begin{equation}
    \Q^{(k)}_{1(\underline{\vec a},2,...,N,\underline{\vec b})} \coloneq \Q^{(k)}(\T_1, \T_{\underline{\vec a}}\T_2\cdots \T_N \T_{\underline{\vec b}}), \quad \tilde \Q^{(k)}_{(\underline{\vec a},1,...,N-1,\underline{\vec b})N} \coloneq \tilde \Q^{(k)}(\T_{\underline{\vec a}}\T_1\cdots \T_{N-1}\T_{\underline{\vec b}}, \T_N),
\end{equation}
for $N\geq 2$ and where we used the simplified notation $\T_{\underline{\vec a}} \coloneq \T_{\underline{\vec a}}(u_{\underline{\vec a}})$, $\T_{\underline{\vec b}} \coloneq \T_{\underline{\vec b}}(u_{\underline{\vec b}})$ and $\T_i \coloneq \T_i(0)$. A similar notation is used for $\tilde \Q^{(k)}$. As shown in Appendix~\ref{sec:proofsfordeformation}, we have the following result:

\begin{lemma}\label{lemma:simplifycharges}
    For $k\geq 2$, we have
    \begin{equation*}
        \begin{split}
            \Q^{(k)}_{1(\underline{\vec a},2,...,k-1,\underline{\vec b})} - \R^{-1}_{1\underline{\vec a}}\Q^{(k)}_{1(2,...,k-1, \underline{\vec b})}\R_{1\underline{\vec a}} &= \Q^{(k)}_{1(\underline{\vec a},2,...,k-1)} - \R^{-1}_{1\underline{\vec a}}\Q^{(k)}_{1(2,...,k-1)}\R_{1\underline{\vec a}},\\
            \tilde \Q^{(k)}_{(\underline{\vec a},1,...,k-1,\underline{\vec b})k} - \R^{-1}_{\underline{\vec b} k}\tilde \Q^{(k)}_{(\underline{\vec a},1,...,k-1)k}\R_{\underline{\vec b}k} &= \tilde \Q^{(k)}_{(1,...,k-1,\underline{\vec b})k} - \R^{-1}_{\underline{\vec b} k}\tilde \Q^{(k)}_{(1,...,k-1)k}\R_{\underline{\vec b}k}
        \end{split}
    \end{equation*}
    where $\vec a = (a_1,..., a_{|\vec a|})$ and $\vec b = (b_1,..., b_{|\vec b|})$. Moreover, the underline in the above equation denotes a non-trivial insertion of a spectral parameter, while the other coordinates have a trivial spectral parameter $0$.
\end{lemma}
\noindent In particular, this shows that such differences of algebraic charge densities have, in some sense, a ``finite length'' $k$. In order to simplify notation even more, we introduce:
\begin{definition}\label{def:reducedalgebraicchargedensities}
    For $k\geq 2$, the \textit{reduced algebraic charge densities} are the trilinear maps ${}^r\Q(k), {}^r\tilde{\Q}^{(k)} \colon \A(\R)\times \A(\R)\times \A(\R)\to \C$ given by
    \begin{equation*}
        \begin{split}
            {}^r\Q^{(k)}(a,b,c) &\coloneq \Q^{(k)}(a,b\cdot c) - (\mathbf{r}^{-1}_{12}\ast\Q^{(k)}_{13}\ast \mathbf{r}_{12})(a,b,c)\\
            {}^r\tilde{\Q}^{(k)}(a,b,c) &\coloneq \tilde{\Q}^{(k)}(a\cdot b,c) - (\mathbf{r}^{-1}_{23}\ast\Q^{(k)}_{13}\ast \mathbf{r}_{23})(a,b,c).
        \end{split}
    \end{equation*}
\end{definition}
\noindent Note here that ${}^r\Q^{(k)}(1,b,c) = 0 ={}^r\Q^{(k)}(a,1,c)$ and similarly ${}^r\tilde{\Q}^{(k)}(a,b,1) = 0 ={}^r\tilde{Q}^{(k)}(a,1,c) $ for all $a,b,c\in \A(\R)$. Moreover, we have ${}^r\Q^{(k)}(a,b,1) = \Q^{(k)}(a,b)$ and ${}^r\tilde{\Q}^{(k)}(1,b,c) = \tilde{\Q}^{(k)}(b,c)$. Importantly, Lemma~\ref{lemma:simplifycharges} now immediately implies that
\begin{equation}\label{eq:reducedAlgebraicchargedensityproperty}
    \begin{split}
        {}^r\Q^{(k)}_{1,\underline{\vec a},(2,...,k-1,\underline{\vec b})} = {}^r\Q^{(k)}_{1,\underline{\vec a},(2,...,k-1)} \quad \text{and} \quad {}^r\tilde{\Q}^{(k)}_{(\underline{\vec a},1,...,k-1),\underline{\vec b}, k} = {}^r\tilde{\Q}^{(k)}_{(1,...,k-1),\underline{\vec b}, k},
    \end{split}
\end{equation}
such that the reduced algebraic charge densities now have, in some sense, a ``finite length''. These reduced algebraic charge densities will play an important role in the derivation of the generalised Sutherland equations in Section~\ref{sec:higherorderSutherlandequation} and in the proofs for the twists of the quantum algebras in Section~\ref{sec:longrangedeformations}.

\subsection{Generalised Sutherland equations}\label{sec:higherorderSutherlandequation}
The higher-order algebraic charge densities now allow us to derive higher-order variants of the Sutherland equation \cite{Sutherland1970Two-Dim}, which will play an important role in the integrable long-range deformations. In particular, as we will see in Section~\ref{sec:LaxforBQ3derivation}, they can be used to derive the Lax operators for the long-range deformations up to first order in the deformation parameter. Moreover, as we will see in Section~\ref{sec:deformationFRTalgebra}, the equations appear in the proofs for the associativity of the corresponding deformed algebras.\\

Let us first recall the Sutherland equation for the second charge. By taking the derivation of the Yang-Baxter equation, one can derive
\begin{equation}
    \label{eq:algebraicSutherlandEquation}
    \begin{split}
        &\R_{23}^{-1}\D_2\left(\R_{12}\R_{13}\R_{23} - \R_{23}\R_{13}\R_{12}\right)=0\\
        \Rightarrow\quad & [\Q^{(2)}_{23}, \R_{13}\R_{12}] = \R_{13}\Q^{(2)}_{21}\R_{12} - \R^{-1}_{23}\Q^{(2)}_{21}\R_{23}\R_{13}\R_{12},
    \end{split}
\end{equation}
By considering the trivial spectral parameters $0$ in coordinates $2$ and $3$, such that $\R_{23}$ becomes the permutation operator, this equation reduces to the known \cite{Sutherland1970Two-Dim} (see also e.g. \cite{Loebbert2016LecturesonYangianSymmetry, deLeeuw2019Classifyingintegrable, deLeeuw2020YangBaxterandBoost}) Sutherland equation 
\begin{equation}
    [\mathfrak{q}^{(2)}_{12}, \L_{\underline{a}(12)}] = \L'_{\underline{a}2}\L_{\underline{a}1} - \L_{\underline{a}2}\L'_{\underline{a}1},
\end{equation}
where $\L'_{\underline{a}1}\coloneq \frac{d}{dv}\L_{\underline{a}1}(u,v)\big|_{v=0}$.\\

We can now use the algebraic charge densities as described in the previous section to derive the higher-order Sutherland equations. Let $\vec a = (a_1,..., a_{|\underline{\vec a}|})$ such that the $R$-matrix in the coordinates $\vec a$ is given by $\R_{\underline{\vec a}1} = \R_{\underline{a}_11}\cdots \R_{\underline{a}_{|\vec a|}1}$, where we use the underline to emphasise a non-trivial insertion of a spectral parameter. For $k\geq 2$, we then have
\begin{equation}
    \begin{split}
        [\Q^{(k)}_{12}, \R_{\underline{\vec a}(12)}] &= \Q^{(k)}_{12}\R_{{\underline{\vec a}}(12)} - \R_{\underline{\vec a}(12)}\Q^{(k)}_{12} - \Q^{(k)}_{1(2\underline{\vec a})}\R_{\underline{\vec a}(12)} + \R_{\underline{\vec a}2}\Q^{(k)}_{1(\underline{\vec a}2)}\R_{\underline{\vec a}1}\\
        &= \R_{\underline{\vec a}2}\left[\Q^{(k)}_{1(\underline{\vec a}2)} - \R_{1\underline{\vec a}}^{-1}\Q^{(k)}_{12}\R_{1\underline{\vec a}}\right]\R_{\underline{\vec a}1} - \left[\Q^{(k)}_{1(2\underline{\vec a})} - \Q^{(k)}_{12}\right]\R_{\underline{\vec a}(12)},
    \end{split}
\end{equation}
where in the first equality we used that $\R_{\underline{\vec a}2}\Q^{(k)}_{1(\underline{\vec a}2)} = \Q^{(k)}_{1(2{\underline{\vec a}})}\R_{{\underline{\vec a}}2}$. Now let $k\geq 3$. By taking a $(k-2)$-th coproduct on the second coordinate, one readily obtains
\begin{equation}
    \label{eq:Sutherlandderivation1}
    \begin{split}
        [\Q^{(k)}_{1(2,...,k)}, \R_{\underline{\vec a}(1,...,k)}] 
        &= \R_{{\underline{\vec a}}(2,..,k)}\left[\Q^{(k)}_{1(\underline{\vec a},2,...,k)}- \R^{-1}_{1\underline{\vec a}}\Q^{(k)}_{1(2,...,k)}\R_{1\underline{\vec a}}\right]\R_{\underline{\vec a}1}\\
        &\quad - \left[\Q^{(k)}_{1(2,...,k,\underline{\vec a})} - \Q^{(k)}_{1(2,...,k)}\right]\R_{\underline{\vec a}(1,...,k)},
    \end{split}
\end{equation}
and similarly
\begin{equation}
    \label{eq:Sutherlandderivation2}
    \begin{split}
        [\Q^{(k)}_{2(3,...,k)}, \R_{\underline{\vec a}(1,...,k)}] &=\R_{\underline{\vec a}(3,...,k)}\left[\Q^{(k)}_{2(\underline{\vec a},3,...,k)} - \R_{2\underline{\vec a}}^{-1}\Q^{(k)}_{2(3,...,k)}\R_{2\underline{\vec a}}\right]\R_{\underline{\vec a}(12)}\\
        &\quad - \left[\Q^{(k)}_{2(3,...,k,\underline{\vec a})} - \Q^{(k)}_{2(3,...,k)}\right]\R_{\underline{\vec a}(1,...,k)}.
    \end{split}
\end{equation}
Analogous expressions can be obtained for a commutation relation between the $R$-matrix $\R$ and algebraic charge density $\tilde \Q^{(k)}$. Let us now use Conjecture~\ref{conj:chargedensities} as a definition for the charge densities. That is, we define two representatives for the $k$-th charge density by
\begin{equation}\label{eq:chargedensitiesfromalgebraic}
    \mathfrak{q}^{(k)}_{1,...,k} \coloneq \Q^{(k)}_{1(2,...,k)}(\vec 0) - \Q^{(k)}_{2(3,...,k)}(\vec 0) \quad \text{and} \quad \tilde{\mathfrak{q}}^{(k)}_{1,...,k} \coloneq \tilde{\Q}^{(k)}_{(1,...,k-1)k}(\vec 0) - \tilde{\Q}^{(k)}_{(1,...,k-2)k-1}(\vec 0).
\end{equation}
By taking the difference between \eqref{eq:Sutherlandderivation1} and \eqref{eq:Sutherlandderivation2}, considering trivial spectral parameters on the coordinates $1,...,k$, and by applying Lemma~\ref{lemma:simplifycharges}, we have the following immediate result (for $\tilde \Q^{(k)}$ mutatis mutandis):
\begin{proposition}\label{prop:higherorderSutherlandEquations}
    Let $k\geq 3$, then the \textit{generalised Sutherland equations} are given by
    \begin{equation*}
        \begin{split}
            [\mathfrak{q}^{(k)}_{1,...,k}, \L_{\underline{\vec a}(1,...,k)}] &= \L_{\underline{\vec a}(2,...,k)}{}^r\Q^{(k)}_{1,\underline{\vec a},(2,...,k-1)}\L_{\underline{\vec a}1} - \L_{\underline{\vec a}(3,...,k)}{}^r\Q^{(k)}_{2,\underline{\vec a},(3,...,k)}\L_{\underline{\vec a}(12)},\\
            [\tilde{\mathfrak{q}}^{(k)}_{1,...,k}, \L_{\underline{\vec a}(1,...,k)}] &= \L_{\underline{\vec a},(k-1,k)}{}^r\tilde{\Q}^{(k)}_{(1,...,k-2),\underline{\vec a}, k-1}\L_{\underline{\vec a}(1,...,k-2)} - \L_{\underline{\vec a}k}{}^r\tilde{\Q}^{(k)}_{(2,...,k-1), \underline{\vec a}, k}\L_{{\underline{\vec a}}(1,...,k-1)},
        \end{split}
    \end{equation*}
    where ${}^r\Q^{(k)}$ and ${}^r\tilde{\Q}^{(k)}$ are the reduced algebraic charge densities as given in Definition~\ref{def:reducedalgebraicchargedensities}.
\end{proposition}
\noindent Remark here that the above proposition generalises the $k=2$ Sutherland equation as given in \eqref{eq:algebraicSutherlandEquation}. Moreover, it is immediately clear that these generalised Sutherland equations depend on the chosen representative of the charge density $\mathfrak{q}^{(k)}_{1,...,k}$. For the purposes of this paper, we will only need the above two representatives for the charge densities.\\

Importantly, these Sutherland equations can now be used to show that the charges defined through Conjecture~\ref{conj:chargedensities} all commute with the transfer matrix. Indeed, using the fact that one obtains a (vanishing) telescoping sum when commuting the charges with the transfer matrix (see also Appendix~\ref{sec:proofsforalgebraicchargedensities}), we have the following immediate Corollary:
\begin{corollary}\label{cor:commutingcharges}
    For $k\geq 3$, let $\mathfrak{q}^{(k)}_{1,...,k} \coloneq \Q^{(k)}_{1(2,...,k)}(\vec 0) - \Q^{(k)}_{2(3,...,k)}(\vec 0)$ and  consider the corresponding charge operator $\mathbb{Q}_k \coloneq \sum_{n\in \Z}\mathfrak{q}^{(k)}_{n,...,n+k-1}$, then\footnote{\label{footnote:infiniteproduct}Remark here that $\overleftarrow{\prod}_{m\in \Z} \L_{\underline{a},m}(u)$ denotes the monodromy matrix \eqref{eq:monodromymatrix} on the doubly-infinite spin chain. Of course, one might wonder whether such an infinite product is well-defined. However, for the purposes of illustrating the result, we will assume that one could, in principle, make sense of this infinite product in some way, for which the result is not expected to depend on its specifics.}
    \begin{equation*}
        \left[\mathbb{Q}_k, \mathrm{tr}_a\left(\overleftarrow{\prod}_{m\in \Z} \L_{\underline{a},m}(u)\right)\right] = 0.
    \end{equation*}
\end{corollary}
\noindent In particular, this shows that the charges $\mathbb{Q}_k$ that are \textit{defined} through the charge densities as given by Conjecture~\ref{conj:chargedensities} all commute with the transfer matrix on the doubly-infinite spin chain. One could, in principle, generalise the above result to show that the charges commute with the transfer matrix on a finite spin chain. The fact that all these charges commute with the transfer matrix is, of course, great evidence that Conjecture~\ref{conj:chargedensities} holds true.

\section{Long-range deformed FRT-bialgebras}\label{sec:longrangedeformations}
We now have all the necessary machinery to discuss the deformations of the FRT-bialgebras that correspond to the long-range integrable deformations of integrable spin chains. Let us first recall the notion of the long-range integrable models. For some deformation parameter $\lambda$, the deformed local charges are given by $\mathbb{Q}_k(\lambda) = \sum_{n=0}^{\infty} \mathbb{Q}_k^{(n)}\lambda^n$ such that $[\mathbb{Q}_k(\lambda), \mathbb{Q}_{\ell}(\lambda)] = 0$, and where $\mathbb{Q}_k^{(0)}$ is the undeformed charge as discussed in Section~\ref{sec:commutingcharges}. Moreover, for the deformation to be long-range, we impose that the higher correction terms have charge densities of range greater than $k$. In particular, the charge density for charge $\mathbb{Q}_k(\lambda)$ can be formally expanded as\footnote{The corresponding operator $\mathbb{Q}_k(\lambda)$ is then often also required to be \textit{quasi-local}, see e.g. \cite{Ilievski2015Quasilocal, Ilievski2016Quasilocal}.}
\begin{equation}\label{eq:chargedensitiesdeformations}
    \mathfrak{q}^{(k)}_{1,...,k} + \sum_{n=1}^{\infty}\mathfrak{q}^{(k),(n)}_{1,...,k_n}\lambda^n,
\end{equation}
where $k_{n+1} > k_n$, $k_0 = k$ and $\mathfrak{q}^{(k),(n)}_{1,...,k_n} \in \mathrm{End}(V^{\otimes k_n})$. That is, the range of the charge density increases (generally linearly) with the order in the deformation parameter $\lambda$. A family of such integrable long-range deformations of spin chains have been described in \cite{Beisert2006Longrange, Bargheer2008BoostingNearestNeighbour, Bargheer2009}. Although these deformations are well-understood on the level of the charges, the corresponding quantum groups had not yet been constructed. Some progress was made in \cite{Gombor2021, Gombor2022WrappingCorrections} and \cite{deLeeuwRetore2023}, where the Lax operator for some long-range deformation was discussed. We will now improve upon these results by explicitly constructing the full FRT-type bialgebras that correspond to these long-range deformations at first order in $\lambda$. In particular, in Section~\ref{sec:classificationLongrangeDeformations}, we will briefly recall the families of long-range deformations that were found in \cite{Bargheer2008BoostingNearestNeighbour, Bargheer2009}. Then, in Section~\ref{sec:LaxforBQ3derivation}, we use the higher-order Sutherland equations to derive the Lax operators for the long-range deformations up to first order in $\lambda$. In Section~\ref{sec:deformationFRTalgebra}, we discuss the corresponding deformations of the FRT-bialgebra that give rise to all the long-range deformations for the family of long-range deformations up to first order in $\lambda$. In particular, we give the explicit twisting elements for these deformations and show that the corresponding Drinfeld associators encode the information for the long-range interactions. In Section~\ref{sec:someexamples}, we then apply this to the XXX Heisenberg spin chain to provide the explicit expressions for the Lax operators and $R$-matrices of the boost $\B[\mathbb{Q}_3]$ and the bilocal $[\mathbb{Q}_2|\mathbb{Q}_3]$ long-range deformation. Lastly, in Section~\ref{sec:longrangedeformedYangian}, we also discuss the long-range deformation for the Yangian algebra.

\subsection{The deformation equation}\label{sec:classificationLongrangeDeformations}
Let us begin by reviewing the families of integrable long-range deformations that were described in \cite{Beisert2006Longrange, Bargheer2008BoostingNearestNeighbour, Bargheer2009}. Here, we will consider the charges on the doubly-infinite spin chain, such that $\mathbb{Q}_k = \sum_{n\in \Z}\mathfrak{q}^{(k)}_{n,...,n+k-1}$. For a deformation parameter $\lambda$, we then consider the deformed charges $\mathbb{Q}_k(\lambda) \coloneq \sum_{n=0}^{\infty}\mathbb{Q}_k^{(n)}\lambda^n$, where $\mathbb{Q}^{(0)}_k = \mathbb{Q}_k$, and the charges satisfy the deformation equation
\begin{equation}\label{eq:deformationequation}
    \frac{d}{d\lambda}\mathbb{Q}_k(\lambda) = [\mathcal{X}(\lambda), \mathbb{Q}_k(\lambda)],
\end{equation}
for some yet to be determined operator $\mathcal{X}(\lambda)$ on the doubly-infinite spin chain. Since $[\mathbb{Q}_k, \mathbb{Q}_\ell] = 0$, it can be easily found that any new tower of charges $\mathbb{Q}_k(\lambda)$ that satisfy \eqref{eq:deformationequation} are also pairwise commuting, i.e. $[\mathbb{Q}_k(\lambda), \mathbb{Q}_\ell(\lambda)] = 0$. However, we are not interested in just any deformation of the charges, but we require the deformed charges to act \textit{locally} and \textit{homogeneously} on the spin chain as well. That is, the charge $\mathbb{Q}_k$ can be written as a sum over the local charge densities of the form given in \eqref{eq:chargedensitiesdeformations}. Three families of choices for $\mathcal{X}(\lambda)$ have been identified in \cite{Bargheer2008BoostingNearestNeighbour, Bargheer2009}:
\begin{itemize}[leftmargin=*]
    \item \textbf{Local operators:} Firstly, let $k \geq 1$ and consider a matrix $\mathfrak{m}_{1,...,k} \in \mathrm{End}(V^{\otimes k})$, which defines a local operator $\O_{\mathfrak{m}} \coloneq \sum_{n\in \Z}\mathfrak{m}_{n,...,n+k-1}$ on the doubly-infinite spin chain. If we then set $\mathcal{X}(\lambda) = \O_\mathfrak{m}$, this will give rise to local and homogeneous long-range deformations of the integrable charges $\mathbb{Q}_k(\lambda)$. In fact, since $\mathcal{X}(\lambda)$ does now not depend on $\lambda$, the corresponding deformation equation \eqref{eq:deformationequation} can now be integrated exactly to give \linebreak $\mathbb{Q}_k(\lambda) = \exp(\lambda \O_\mathfrak{m})\mathbb{Q}_k(0)\exp(-\lambda \O_\mathfrak{m})$. In particular, the deformed charge $\mathfrak{\Q}_k(\lambda)$ is related to the undeformed charge via conjugation. Such a deformation on the level of the quantum group can be simply achieved by a change in representation. However, in Section~\ref{sec:deformationFRTalgebra}, we will still consider a deformation of the FRT-bialgebra that leads to the local-operator long-range deformation, as it is an easy example to consider to gain some additional intuition behind the other long-range quantum group deformations.
    
    \item \textbf{Boost operators:} For a long-range deformed charge $\mathbb{Q}_k(\lambda)$, we write its corresponding charge density as $\mathfrak{q}^{(k)}_n(\lambda) \coloneq \sum_{m\in \Z}\mathfrak{q}^{(k),(m)}_{n,...,n+k_m - 1}\lambda^m$, where $k_{m+1} > k_m$ and $k_0 = k$ in line with \eqref{eq:chargedensitiesdeformations}, such that $\mathbb{Q}_k(\lambda) = \sum_{n\in \Z}\mathfrak{q}^{(k)}_n(\lambda)$. The \textit{boost operator} corresponding to $\mathbb{Q}_k(\lambda)$ is then given by
    \begin{equation}\label{eq:boostoperatordef}
        \B[\mathbb{Q}_k(\lambda)] \coloneq \sum_{n\in \Z}n\cdot \mathfrak{q}^{(k)}_n(\lambda).
    \end{equation}
    Setting $\mathcal{X}(\lambda) = \B[\mathbb{Q}_{k}(\lambda)]$ for any $k\geq 3$ will then give rise to a non-trivial long-range integrable local and homogeneous deformation of the charges. We remark here that $\mathcal{X}(\lambda)$ depends on the charge $\mathbb{Q}_{k}(\lambda)$, which itself depends on $\mathcal{X}(\lambda)$ through the deformation equation \eqref{eq:deformationequation}. In particular, one should now think of $\eqref{eq:deformationequation}$ as defining a recurrence relation for $\mathbb{Q}_{k}(\lambda)$, which can be solved iteratively at each order of $\lambda$. One might remark here that the charge densities $\mathfrak{q}^{(k)}_n(\lambda)$ are not actually unique, and only defined up to the equivalence as given in Definition~\ref{def:equivalencelocalchargedensities}. However, it can be easily verified that, for two different representatives $\mathfrak{q}^{(k)}$ and $\tilde{\mathfrak{q}}^{(k)}$ for the charge density, the corresponding boost operators \eqref{eq:boostoperatordef} differ by a local operator.
    
    \item \textbf{Bilocal operators:} Lastly, for $2\leq k < \ell$, we consider the \textit{bilocal operator}\footnote {We remark here that our definition for the bilocal operator is different from the one used in \cite{Bargheer2008BoostingNearestNeighbour, Bargheer2009}. However, it is easy to verify that the definitions differ from each other by a local operator. We use our definition, as it is better suited for the deformation of the FRT-algebra that we will describe in Section~\ref{sec:deformationFRTalgebra}.}
    \begin{equation}\label{eq:bilocaloperatordef}
        [\mathbb{Q}_k(\lambda)|\mathbb{Q}_\ell(\lambda)] \coloneq \sum_{n\leq m}\mathfrak{q}^{(k)}_{n+\ell-k}(\lambda)\mathfrak{q}^{(\ell)}_m(\lambda).
    \end{equation}
    Setting $\mathcal{X}(\lambda) = [\mathbb{Q}_l(\lambda)|\mathbb{Q}_\ell(\lambda)]$ then gives rise to non-trivial long-range integrable local and homogeneous charges. Again, $\mathcal{X}(\lambda)$ depends on the charges $\mathbb{Q}_k(\lambda), \mathbb{Q}_{\ell}(\lambda)$, which themselves in turn depend on $\mathcal{X}(\lambda)$ through the deformation equation \eqref{eq:deformationequation}. Therefore, also for the bilocal operator, the deformation equation \eqref{eq:deformationequation} defines a recurrence relation for the charges. Moreover, it can be noted that by choosing a different representative for the local charge densities, the corresponding bilocal operators will differ by a local operator. Lastly, by substituting $\mathfrak{q}^{(k)}_n \to 1$ in \eqref{eq:bilocaloperatordef}, it can be seen that one obtains the boost operator \eqref{eq:boostoperatordef}. Therefore, one can view the long-range deformation corresponding to the boost operator as a special case of the deformation corresponding to the bilocal operator. We will see in Section~\ref{sec:deformationFRTalgebra} that this remains true for the deformation of the quantum group.
\end{itemize}
Additionally, \cite{Bargheer2009} introduces a change of basis of the charges, i.e. $\mathbb{Q}_k(\lambda) \to \Lambda_k^\ell(\lambda)\mathbb{Q}_\ell(\lambda)$, as a long-range deformation. However, such a change of basis transformation does not realise a truly inequivalent tower of commuting charges, and we will therefore not consider it as a true deformation. In fact, such a deformation is also not realisable on the level of the FRT-bialgebra, as its generated charges should always increase in range. The closest that one can get to a change of basis is a reparametrisation of the spectral parameter, i.e. $u \to f_{\lambda}(u)$ for some $\lambda$ dependent function. The spin chain might also have some additional Lie algebra symmetry, which in principle provides additional homogeneous local charges. The deformation corresponding to the boost operator for these Lie algebra charges will, in general, give rise to a deformation of nearest-neighbour type, i.e. non-long-range. However, the bilocal operators involving such charges do give rise to additional long-range deformations \cite{BFLL2013IntegrableDeformations}.\\

The above families of long-range deformations can now be combined to obtain a moduli space of long-range deformations. In particular, one can define
\begin{equation}
    \mathcal{X}(\lambda) \coloneq \sum_{n=0}^{\infty}\lambda^n\left[\mathfrak{m}_n + \alpha^k_n\B[\mathbb{Q}_k(\lambda)] + \beta^{k|\ell}_n[\mathbb{Q}_k(\lambda) | \mathbb{Q}_\ell(\lambda)]\right]
\end{equation}
for a general long-range deformation, where $\mathfrak{m}_n$ is a (possibly zero) local operator, and the indices $k$ and $\ell$ are summed over. Again, one should remember here that $\mathcal{X}(\lambda)$ is dependent on the charges $\mathbb{Q}_k(\lambda)$, which themselves are dependent on $\mathcal{X}(\lambda)$. For the purposes of this paper, we will only be interested in the first-order correction to the charges. In particular, we will only consider $\mathcal{X}(0)$ in the deformation equation \eqref{eq:deformationequation}, which is generated by some local operator $\mathfrak{m}$, the boost operators $\B[\mathbb{Q}_k(0)]$ and the bilocal operators $[\mathbb{Q}_k(0) | \mathbb{Q}_\ell(0)]$.

\subsection{Lax operators for the long-range deformations}\label{sec:LaxforBQ3derivation}
To obtain some intuition behind the FRT-bialgebras that will describe the long-range deformations, we will first derive the corresponding Lax operators at first order in the deformation parameter $\lambda$. Let us first introduce some notation. We consider the monodromy matrix on the infinite spin chain, which we denote by the infinite (co)product $\L_{\underline{a},(...)}(u)\coloneq \overleftarrow{\prod}_{m\in \Z} \L_{\underline{a},m}(u)$.\footnote{See footnote~\ref{footnote:infiniteproduct}.} Similarly, we define the Lax operators on the half-infinite chains by $\L_{\underline{a},(n,...)}(u)\coloneq \overleftarrow{\prod}_{m\in \Z_{\geq n}} \L_{\underline{a},m}(u)$ and $\L_{\underline{a},(...,n)}(u)\coloneq \overleftarrow{\prod}_{m\in \Z_{\leq n}} \L_{\underline{a},m}(u)$. If one now considers the commutant between the boost operator $\B[\mathbb{Q}_2]$ for the second charge and the monodromy matrix, one obtains
\begin{equation}
    \label{eq:boostQ2deformation}
    \begin{split}
        [\B[\mathbb{Q}_2], \L_{\underline{a},(...)}] &= \sum_{n\in \Z}n\cdot [\mathfrak{q}^{(2)}_{n,n+1}, \L_{\underline{a},(...)}]\\
        &= \sum_{n\in \Z}n\cdot\L_{\underline{a}, (n+2,...)}\left[\L'_{\underline{a},n+1}\L_{\underline{a},n}-\L_{\underline{a},n+1}\L'_{\underline{a},n}\right]\L_{\underline{a},(...,n-1)} = -\L'_{\underline{a},(...)},
    \end{split}
\end{equation}
where $\L'_{\underline{a}n}\coloneq \frac{d}{dv}\L_{\underline{a}n}(u,v)\big|_{v=0}$, and in the last step we did the substitution $n\to n-1$ in the first term, which then cancels with the second term to obtain a sum over derivatives in the $n$-th coordinate, and hence a total derivative of the monodromy matrix. Of course, this is a well-known result (see e.g \cite{Loebbert2016LecturesonYangianSymmetry}), giving rise to the boost formalism for generating the higher-order charges from just the second charge \cite{TetelmanLorentzGroupforTwoDimensional, Links2001LadderOperator}.\\

Using the generalised Sutherland equations as given in Proposition~\ref{prop:higherorderSutherlandEquations}, one can now easily generalise this result to compute the commutants between general boost or bilocal operators and the monodromy matrix on the doubly-infinite spin chain. In particular, for $k\geq 3$, one readily obtains
\begin{equation}\label{eq:boostmonodromycommutator}
    \begin{split}
        [\B[\mathbb{Q}_k], \L_{\underline{a},(...)}] = \sum_{n\in \Z}n\cdot[\mathfrak{q}^{(k)}_{n,...,n+k-1}, \L_{\underline{a},(...)}] = \sum_{n\in \Z}\L_{\underline{a}(n+1,...)}{}^r\Q^{(k)}_{n,\underline{a},(n+1,...,n+k-2)}\L_{\underline{a},(...,n)}.
    \end{split}
\end{equation}
Similarly, for $2\leq k < \ell$, we have
\begin{equation}\label{eq:bilocalmonodromycommutator}
    \begin{split}
        &[[\mathbb{Q}_k | \mathbb{Q}_\ell], \L_{\underline{a}, (...)}] = \sum_{n\leq m}[\tilde{\mathfrak{q}}^{(k)}_{n+\ell-k,...,n+\ell-1}\cdot \mathfrak{q}^{(\ell)}_{m,...,m+\ell-1}, \L_{\underline{a},(...)}]\\
        =& \sum_{n\in \Z}\big(\tilde{\mathfrak{q}}^{(k)}_{n+\ell-k,...,n+\ell-1}\cdot\big[\sum_{m \geq n}\mathfrak{q}^{(\ell)}_{m,...,m+\ell-1}, \L_{\underline{a},(...)}\big]\\
        &\qquad \qquad \qquad \qquad\ \  + \big[\sum_{m\leq n}\tilde{\mathfrak{q}}^{(k)}_{m+\ell-k,...,m+\ell-1}, \L_{\underline{a},(...)}\big]\cdot \mathfrak{q}^{(\ell)}_{n,...,n+\ell-1}\big)\\
        =& \sum_{n\in \Z}\big(\L_{\underline{a},(n+\ell,...)}\tilde{\mathfrak{q}}^{(k)}_{n+\ell-k,...,n+\ell-1}\L_{\underline{a},(n+1,...,n+\ell-1)}{}^r\Q^{(\ell)}_{n, \underline{a},(n+1,...,n+\ell-2)}\L_{\underline{a},(...,n)}\\
        &\qquad - \L_{\underline{a},(n+\ell-1,...)}{}^r\tilde{\Q}^{(k)}_{(n+\ell-k+1,...,n+\ell-2), \underline{a}, n+\ell-1}\L_{\underline{a},(n,...,n+\ell-2)}\mathfrak{q}^{(\ell)}_{n,...,n+\ell-1}\L_{a,(...,n-1)}\big),
    \end{split}
\end{equation}
where we recall that ${}^r\Q^{(k)}$ and ${}^r\tilde{\Q}^{(k)}$ are the reduced algebraic charge densities as given in Definition~\ref{def:reducedalgebraicchargedensities}. Remark that we use the two different representatives $\tilde{q}^{(k)}$ and $\mathfrak{q}^{(k)}$ for the charge densities. This choice of representatives will lead to a more natural form of the FRT-bialgebra deformation as described in Section~\ref{sec:deformationFRTalgebra}.\\

Using the above formula, one can readily write down the corresponding Lax operators that generate the deformed monodromy matrices. Before writing down the corresponding Lax operators, let us first introduce some notation. In particular, as noted in \cite{Gombor2021, Gombor2022WrappingCorrections, deLeeuwRetore2023}, and as we will see in this section, to describe long-range charges in a Lax-operator formalism, one needs to increase the auxiliary space of the Lax operator. Therefore, for $m\geq 1$ we define 
\begin{equation}\label{eq:undeformedlongrangeLaxoperator}
    \L_{(\underline{a},b_1,...,b_m)n}(u) \coloneq \R_{(\underline{a},b_1,...,b_m)n}(u,0,...,0) = \L_{\underline{a}n}(u)\P_{b_1,n}\cdots\P_{b_m,n},
\end{equation}
where $\P$ is the permutation operator, and where we used the factorisation of the $R$-matrix as given in \eqref{eq:Rmatrixfactorisation}. Recall here the notation as introduced in Section~\ref{sec:algebraicchargedensities} that an underline in the coordinate corresponds to a non-trivial insertion of a spectral parameter, while no underline signifies a trivial spectral parameter $0$. In light of the above notation, we will also write $\P_{(b_1,...,b_m)n} \coloneq \P_{b_1,n}\cdots \P_{b_m,n}$ and $\P_{(b_1,...,b_m)(...)} \coloneq  \overleftarrow{\prod}_{n\in \Z} \P_{(b_1,...,b_m)n}$, which is just a shift operator on the auxiliary space and the doubly-infinite spin chain. The Lax operators corresponding to the long-range deformations are now given as follows:
\begin{itemize}[leftmargin=*]
    \item \textbf{Local operator:} let $k\geq 2$ and $\mathfrak{m} \in \mathrm{End}(V^{\otimes k})$, then the Lax operator corresponding to the local operator deformation up to first order in $\lambda$ is given by
    \begin{equation}\label{eq:localoperatorLaxoperator}
        \begin{split}
            \L^\lambda_{(\underline{a},b_1,...,b_{k-1})n} &\coloneq \L_{(\underline{a}, b_1,...,b_{k-1})n}+\lambda\Big(\R^{-1}_{\underline{a}(b_1,...,b_{k-1})}\mathfrak{m}_{n,b_1,...,b_{k-1}}R_{\underline{a}(b_1,...,b_{k-1})}\\
            &\qquad\qquad\qquad\qquad\quad - \R^{-1}_{n\underline{a}}\mathfrak{m}_{n,b_1,...,b_{k-1}}\R_{n\underline{a}}\Big)\L_{(\underline{a}, b_1,...,b_{k-1})n} + \O(\lambda^2).
        \end{split}
    \end{equation}
    Indeed, if we define the corresponding monodromy matrix on the doubly-infinite spin chain $\L^\lambda_{(\underline{a},b_1,...,b_{k-1})(...)}\coloneq \overleftarrow{\prod}_{n\in \Z} \L^\lambda_{(\underline{a},b_1,...,b_{k-1})n}$, then it can be readily calculated by bringing all the permutation operators to right side, that the first-order correction is given by $\L^{(1)}_{(\underline{a},b_1,...,b_{k-1})(...)} = [\O_\mathfrak{m}, \L_{\underline{a},(...)}]\times \P_{(b_1,...,b_{k-1})(...)}$. In particular, this shows that the first-order correction to the monodromy matrix is given by the commutator of the monodromy matrix with a local operator, times a shift operator. When calculating the charges, this shift operator will be insignificant, such that indeed all the charges will correspond to the local operator long-range deformation as described in Section~\ref{sec:classificationLongrangeDeformations}. 

    \item \textbf{Boost operator:} let $k\geq 3$, then the Lax operator corresponding to the boost operator deformation up to first order in $\lambda$ is given by
    \begin{equation}\label{eq:boostoperatorLaxoperator}
        \L^{\lambda}_{(\underline{a}, b_1,...,b_{k-2})n}\coloneq \L_{(\underline{a}, b_1,...,b_{k-2})n} + \lambda\left( {}^r\Q^{(k)}_{n,\underline{a},(b_1,...,b_{k-2})}\L_{(\underline{a}, b_1,...,b_{k-2})n}\right) + \O(\lambda^2).
    \end{equation}
    Again, one can consider the monodromy matrix $\L^\lambda_{(\underline{a},b_1,...,b_{k-2})(...)}\coloneq \overleftarrow{\prod}_{n\in \Z} \L^\lambda_{(\underline{a},b_1,...,b_{k-2})n}$, for which the first-order correction, by bringing all the permutation operators to the right, can be calculated to be given by $\L^{(1)}_{(\underline{a},b_1,...,b_{k-2})(...)} = [\B[\mathbb{Q}_k], \L_{\underline{a},(...)}]\times \P_{(b_1,...,b_{k-2})(...)}$, where the commutator between the boost operator and the monodromy matrix was given in \eqref{eq:boostmonodromycommutator}. Therefore, this Lax operator will indeed generate all the charges corresponding to the boost operator-generated long-range deformation.

    \item \textbf{Bilocal operator:} Let $2\leq k < \ell$, the Lax operator corresponding to the bilocal operator deformation up to first order in $\lambda$ is given by
    \begin{equation}\label{eq:bilocaloperatorLaxoperator}
        \begin{split}
             \L^{\lambda}_{(\underline{a}, b_1,...,b_{\ell-1})n} &\coloneq \L_{(\underline{a}, b_1,...,b_{\ell-1})n} + \O(\lambda^2)\\
             &+ \lambda\Big(\R^{-1}_{\underline{a}(b_1,...,b_{\ell-1})}\tilde{\mathfrak{q}}^{(k)}_{b_{\ell-k},...,b_{\ell-1}}\R_{\underline{a}(b_1,...,b_{\ell-1})}{}^r\Q^{(\ell)}_{n,\underline{a},(b_1,...,b_{\ell-1})}\L_{(\underline{a}, b_1,...,b_{\ell-1})n}\\
             &\qquad - \L_{(\underline{a}, b_1,...,b_{\ell-1})n}\R^{-1}_{\underline{a}(b_1,...,b_{\ell-1})}{}^r\tilde{\Q}^{(k)}_{(b_1,...,b_{\ell-1}),\underline{a}, n}\R_{\underline{a}(b_1,...,b_{\ell-1})}\mathfrak{q}^{(\ell)}_{b_1,...,b_{\ell-1},n}\Big).
        \end{split}
    \end{equation}
    Again, when we consider the monodromy matrix $\L^\lambda_{(\underline{a},b_1,...,b_{\ell-1})(...)}\coloneq \overleftarrow{\prod}_{n\in \Z} \L^\lambda_{(\underline{a},b_1,...,b_{\ell-1})n}$, it can be derived that $\L^{(1)}_{(\underline{a},b_1,...,b_{\ell-1})(...)} = [[\mathbb{Q}_k|\mathbb{Q}_\ell], \L_{\underline{a},(...)}]\times \P_{(b_1,...,b_{\ell-1})(...)}$, where the commutator between the bilocal operator and the monodromy matrix was given in \eqref{eq:bilocalmonodromycommutator}. Therefore, this Lax operator indeed generates the long-range deformed charges corresponding to the bilocal operator.
\end{itemize}

As discussed in the previous section, the undeformed model might possess some additional Lie algebra symmetry, which also provides local homogeneous charges. Bilocal operators of the form $[ \mathbb{Q}_k|J]$, where $J$ is a Lie algebra generator, then also give rise to long-range deformations \cite{BFLL2013IntegrableDeformations}. The Lax operators for these long-range deformations can, in principle, be derived analogously, using a Sutherland-type equation for the charge $J$ and the commutator $[\Delta(J)_{12}, \L_{\underline{\vec a}(12)}]$. We remark that these equations are model-dependent and will therefore not be discussed in this paper. However, we emphasise that all the results can easily be generalised to include these types of Lie algebra charges and corresponding long-range deformations for specific $R$-matrices.\\

It should be noted that the above Lax operators can now also be used to construct the charges on spin chains of finite length $L$, in which case all the corresponding charge densities will still have the same long-range corrections. Namely, the long-range charge densities on the finite-length chains can now appear due to the increased auxiliary space. By computing the monodromy matrix $\L^\lambda_{(\underline{a},b_1,...,b_m)L}\cdots \L^\lambda_{(\underline{a},b_1,...,b_m)1}$ and tracing over the increased auxiliary space for the transfer matrix, the \textit{wrapping corrections}, when the interaction range exceeds the length of the spin chain, of the long-range charge densities will occur \cite{Gombor2021, Gombor2022WrappingCorrections, deLeeuwRetore2023}. Lastly, we note that the deformed Lax operators are now, in general, not regular. For example, the Lax operator for the boost operator $\B[\mathbb{Q}_k]$ deformation evaluated at $u =0$ is given by \begin{equation}
    \L^\lambda_{(a,b_1,...,b_{k-2})n}(0) = \left(1+\lambda \mathfrak{q}^{(k)}_{n,a,b_1,...,b_{k-2}}\right)\P_{(a,b_1,...,b_{k-2})n} + \O(\lambda^2),
\end{equation}
which is therefore not exactly equal to a shift operator. Of course, this is expected as the commutator between the boost operator and the shift operator is non-zero, such that the first charge $\mathbb{Q}_1 = \log(t(0))$ should also get a correction. Therefore, these Lax operators do not fall into the class of long-range Lax operators that were described in \cite{Gombor2021, Gombor2022WrappingCorrections}. Although the Lax operator is not regular, it should be noted that locality of the charges is still ensured, as the Lax operator is regular at zeroth order in $\lambda$. Moreover, as we will see, the associated $R$-matrix is regular.

\subsection{Deformation of the FRT-bialgebra}
\label{sec:deformationFRTalgebra}
As described in Section~\ref{sec:FRTalgebrasandintegrablespinchains}, the Lax operators for Yang-Baxter integrable spin chains appear as representations of an associated FRT-bialgebra. Therefore, we would like to find the FRT-type bialgebras that give rise to the long-range deformed Lax operators that we described in the previous section. In particular, one expects the Lax operator to arise as a representation of a deformation of the FRT-bialgebra $\A(\R)$. Here, the form of the Lax operators as given in \eqref{eq:localoperatorLaxoperator}, \eqref{eq:boostoperatorLaxoperator}, and \eqref{eq:bilocaloperatorLaxoperator} already gives us some clues towards what such a deformation should look like. Firstly, the fact that the auxiliary space needs to be increased with coordinates with trivial spectral parameters also suggests that the algebra elements in $\A(\R)$ need to be multiplied by some ``trivial elements''. We will show that this increasing of the auxiliary space corresponds on the level of the FRT-bialgebra to taking a so called ``double-crossed product'' with another bialgebra. Moreover, one can see in \eqref{eq:bilocaloperatorLaxoperator} that the Lax operator looks very similar to one that is obtained from a (Drinfeld) twisted bialgebra \cite{Drinfeld1983Constantquasiclassical}. Indeed, we will see that the long-range deformations correspond to a twist of the double-crossed product algebra. However, as we discussed at the end of Section~\ref{sec:commutingcharges}, and as we will see in this Section, this algebra will not be (fully) associative in order to describe the long-range deformation.\\

Let us begin by discussing the double-crossed product construction, which naturally corresponds to increasing the size of the auxiliary space. Let $\A(\R)$ be an FRT-bialgebra and let $\A_0\subset \A(\R)$ be the free algebra generated by $t_{ij}^{(-1)}$, i.e. the subalgebra of $\A(\R)$ consisting of sums of products of elements $t_{ij}(0)$ that correspond to the spectral parameter set to $u=0$. Since $\A_0$ is a subalgebra of $\A(\R)$, it also naturally inherits a bialgebra structure and a Hopf algebra structure under the completion to $\A_0^\circ$. Moreover, the $r$-form is also well-defined over $\A_0$, and naturally induces a right-action of $\A(\R)$ on $\A_0$ and a left action of $\A_0$ on $\A(\R)$ given by
\begin{equation}\label{eq:matchedaction}
    a\triangleleft b \coloneq \sum \mathbf{r}^{-1}(a_{(1)}, b_{(1)})a_{(2)}\mathbf{r}(a_{(3)}, b_{(2)})\quad \text{and} \quad a\triangleright b \coloneq \sum \mathbf{r}^{-1}(a_{(1)}, b_{(1)})b_{(2)}\mathbf{r}(a_{(2)}, b_{(3)}),
\end{equation}
for $a\in \A_0$ and $b\in \A(\R)$. We now have the following definition \cite{Majid1990Physicsforalgebraists, Majid1990BicrossproductExamples} (see also e.g. \cite{Majid1995FoundationsQG, KasselQuantumGroups, QGAR}):

\begin{definition}\label{def:doublecrossedproduct}
    The \textit{double-crossed product bialgebra}\footnote{This particular double-crossed product bialgebra is nothing different from the \textit{quantum double} $\D(\A(\R), \A_0;\mathbf{r})$. Remark also that the double-crossed product bialgebra $\A(\R)\bowtie_{\mathbf{r}}\A_0$ can be obtained as a twist from the tensor algebra $\A(\R)\otimes \A_0$ with $\gamma(a\otimes b,c\otimes d) = \epsilon(a)\mathbf{r}^{-1}(b,c)\epsilon(d)$, see e.g. \cite{Majid1995FoundationsQG, KasselQuantumGroups, QGAR}.} $\A(\R)\bowtie_{\mathbf{r}}\A_0$ is the vector space $\A(\R)\otimes \A_0$ that becomes a bialgebra with coproduct and counit induced from the coalgebra structure of $\A(\R)\otimes \A_0$, and multiplication given by\footnote{Of course, one should ask whether this definition of the multiplication indeed leads to a well-defined bialgebra. This follows from the fact that $\A(\R)$ and $\A_0$, together with the actions given in \eqref{eq:matchedaction}, form a so-called \textit{matched pair of bialgebras}. See e.g. \cite{Majid1995FoundationsQG, KasselQuantumGroups, QGAR} for details.}
    \begin{equation}
        \label{eq:bicrossedproduct}
        (a\otimes b)(c\otimes d) = \sum a(b_{(1)}\triangleright c_{(1)})\otimes (b_{(2)}\triangleleft c_{(2)})d,
    \end{equation}
    for $a,c\in \A(\R)$ and $b,d\in \A_0$.
\end{definition}
\noindent Under the completion to the Hopf algebras $\A^\circ(\R)$ and $\A_0^\circ$, the double-crossed product algebra $\A^\circ(\R)\bowtie_{\mathbf{r}}\A_0^\circ$ is also a Hopf algebra with antipode given by $S(a\otimes b) = (1\otimes S(a))(S(b)\otimes 1)$.\\

Note here that the maps $a\mapsto a\otimes 1$ and $b\mapsto 1\otimes b$ give the (natural) inclusions of the bialgebras $\A(\R)\subseteq \A(\R)\bowtie_{\mathbf{r}}\A_0$ and $\A_0 \subseteq \A(\R)\bowtie_{\mathbf{r}}\A_0$, respectively. Moreover, it can be directly calculated that for the matrix corepresentations $\T(u)$ and $\T(0)$ in $\A(\R)$ and $\A_0$, the product \eqref{eq:bicrossedproduct} is given by
\begin{equation}
    \label{eq:specificbicrossedproduct}
    \left(\T_{\underline 1}(u)\otimes \T_{2}(0)\right)\left(\T_{\underline 3}(v)\otimes \T_{4}(0)\right) = \R^{-1}_{2\underline{3}}(0,v)\left[\T_{\underline 1}(u)\T_{\underline 3}(v)\otimes \T_{2}(0)\T_{4}(0)\right]\R_{2\underline{3}}(0,v),
\end{equation}
where we denote the matrix indices for elements in $\A(\R)$ with non-trivial spectral parameter with an underline $\underline i$. Note here that the appearance of the $R$-matrix is expected from the exchange in the order of the matrices $\T_{\underline 3}(v)$ and $\T_{2}(0)$, i.e. the double-crossed product \eqref{eq:specificbicrossedproduct} looks exactly the same as just the product of elements in $\A(\R)$. Moreover, the algebra $\A(\R)\bowtie_{\mathbf{r}}\A_0$ is also cotriangular, with the universal $r$-form given by \cite{Majid1995FoundationsQG, KasselQuantumGroups, QGAR}
\begin{equation}\label{eq:rformdoubled}
    \mathbf{r}(a\otimes b, c\otimes d) = \sum \mathbf{r}(a_{(1)}, d_{(1)})\mathbf{r}(a_{(2)}, c_{(1)})\mathbf{r}(b_{(1)}, d_{(2)})\mathbf{r}(b_{(2)}, c_{(2)}),
\end{equation}
for $a,c\in \A(\R)$ and $b,d\in \A_0$. Note that the right-hand side is given by the same expression as the factorisation of $\mathbf{r}(ab, cd)$ for the $r$-form of $\A(\R)$, signifying again that the double-crossed product exactly resembles the ordinary product on $\A(\R)$.\\

Of course, one might wonder why we need to consider such a double-crossed product at all if it resembles the same structure as $\A(\R)$. However, we will see that this algebra deforms more naturally into the quantum group that corresponds to the long-range deformation. For example, it can be noted here that the operator
\begin{equation}
    \label{eq:undeformedincreasedauxiliaryLaxoperator}
    \L_{(\underline{a},b_1,...,b_m)1}(u) \coloneq \mathbf{r}\left(\T_{\underline{a}}(u)\otimes \T_{b_1}(0)\cdots \T_{b_{m}}(0), 1\otimes \T_1(0)\right) = \R_{\underline{a}1}(u)\P_{b_1,1}\cdots \P_{b_m,1},
\end{equation}
for $m\geq 1$ is exactly the undeformed Lax operator on the increased auxiliary space as given in \eqref{eq:undeformedlongrangeLaxoperator}. In particular, in the light of Section~\ref{sec:coquasitriangularityandLfunctionals}, this Lax operator corresponds to the representation of the matrix corepresentation $\T_{\underline a}(u)\otimes \T_{b_1}(0)\cdots \T_{b_m}(0)$ of $\left(\A(\R)\bowtie_{\mathbf{r}}\A_0\right)^{\mathrm{cop}}$.\\

The correct long-range deformed FRT-type bialgebra now corresponds to deformations of the double-crossed product bialgebra $\A(\R)\bowtie_{\mathbf{r}}\A_0$. In particular, one introduces a twist to the multiplication. Let us first consider such a twisting for general bialgebras \cite{Drinfeld1989QuasiHopf} (see also e.g. \cite{Majid1995FoundationsQG, KasselQuantumGroups, QGAR}):

\begin{definition}
    Let $\A$ be a bialgebra and $\gamma\colon \A\times \A\to \C$ a convolution-invertible bilinear map that is unital, i.e. $\gamma(a,1) = \epsilon(a) = \gamma(1,a)$. The \textit{twisted} (possibly non-associative) bialgebra $\A(\gamma)$ is the vector space $\A$ with the coalgebra structure induced from $\A$ and with twisted multiplication given by
    \begin{equation}
        a\cdot_\gamma b \coloneq \sum \gamma(a_{(1)}, b_{(1)})a_{(2)}\cdot b_{(2)}\gamma^{-1}(a_{(3)}, b_{(3)}),
    \end{equation}
    for $a,b\in \A(\gamma)$.
\end{definition}
\noindent As specified in the definition, the twisted bialgebra $\A(\gamma)$ is, in general, not associative. However, it is ``almost-associative'' \cite{Majid1995FoundationsQG} in the sense that\footnote{Such an ``almost-associative'' bialgebra for which the Drinfeld associator is non-trivial might also be called a \textit{quasibialgebra} \cite{Majid1995FoundationsQG}.}
\begin{equation}
    a\cdot_\gamma \left(b \cdot_{\gamma} c\right) = \sum \varphi\left(a_{(1)}, b_{(1)}, c_{(1)}\right)\left[\left(a_{(2)}\cdot_{\gamma}b_{(2)}\right)\cdot_{\gamma}c_{(2)}\right]\varphi^{-1}\left(a_{(3)}, b_{(3)}, c_{(3)}\right),
\end{equation}
where $\varphi\colon \A\times \A\times \A\to \C$ is the \textit{Drinfeld associator} given by \cite{Drinfeld1989QuasiHopf}
\begin{equation}
    \label{eq:definitionDrinfeldAssociator}
    \varphi(a,b,c) \coloneq \sum \gamma\left(b_{(1)}, c_{(1)}\right)\gamma\left(a_{(1)}, b_{(2)}c_{(2)}\right)\gamma^{-1}\left(a_{(2)}b_{(3)}, c_{(3)}\right)\gamma^{-1}\left(a_{(3)}, b_{(4)}\right).
\end{equation}
The bialgebra $\A(\gamma)$ is therefore associative if and only if the Drinfeld associator is trivial, i.e. if it is given by $\varphi(a,b,c) = \epsilon(a)\epsilon(b)\epsilon(c)$, in which case $\gamma$ is said to satisfy the \textit{$2$-cocycle condition} \cite{Majid1993Crossproduct}. If $\gamma$ is a such a 2-cocyle and $\A$ is a Hopf algebra, then $\A(\gamma)$ is also a Hopf algebra with antipode $S^\gamma(a) \coloneq \sum f(a_{(1)})S(a_{(2)})f^{-1}(a_{(3)})$ where $f(a) \coloneq \sum \gamma(a_{(1)}, S(a_{(2)}))$. In the dual picture, as described in Section~\ref{sec:dualdescription}, the twist corresponds to a modification of the coproduct of the corresponding dual quantised universal enveloping algebra, where the additional possibly non-trivial Drinfeld associator turns the Hopf-algebra into a quasi-Hopf algebra \cite{Drinfeld1989QuasiHopf}. We will discuss this in more detail in Section~\ref{sec:longrangedeformedYangian}.\\

If $\A$ is coquasitriangular, then $\A(\gamma)$ is also coquasitriangular (in a possibly non-associative sense). In particular, let $\mathbf{r}$ be the $r$-form for $\A$, then the $r$-form $\mathbf{r}^{\gamma}$ for $\A(\gamma)$ is given by \cite{Drinfeld1989QuasiHopf} (see also e.g. \cite{Majid1995FoundationsQG, KasselQuantumGroups, QGAR})
\begin{equation}
    \label{eq:deformedrform}
    \mathbf{r}^\gamma(a,b) \coloneq (\gamma_{21}\ast\mathbf{r}\ast \gamma^{-1})(a,b) = \sum\gamma(b_{(1)}, a_{(1)})\mathbf{r}\left(a_{(2)}, b_{(2)}\right)\gamma^{-1}\left(a_{(3)}, b_{(3)}\right),
\end{equation}
such that the braiding property $\sum \mathbf{r}^\gamma(a_{(1)}, b_{(1)})a_{(2)}\cdot_{\gamma}b_{(2)} = \sum b_{(1)}\cdot_{\gamma}a_{(1)}\mathbf{r}^\gamma(a_{(2)}, b_{(2)})$ holds for elements in $\A(\gamma)$ as well. Remark that if $\A$ is cotriangular, i.e. $\mathbf{r}^{-1} = \mathbf{r}_{21}$, then $\A(\gamma)$ is also cotriangular. As noted in Section~\ref{sec:coquasitriangularityandLfunctionals}, the factorisation property \eqref{eq:factorisationrform} is tightly linked to the associativity of the algebra. Since, after twisting, the algebra is no longer necessarily associative, one also does not expect the factorisation property to hold. Indeed, it can be verified that \cite{Drinfeld1989QuasiHopf}
\begin{equation}
    \label{eq:nonassociativefactorisation}
    \begin{split}
        \mathbf{r}^\gamma(a\cdot_{\gamma}b,c) &= \left(\varphi_{312}\ast \mathbf{r}^\gamma_{13}\ast \varphi^{-1}_{132}\ast \mathbf{r}^{\gamma}_{23}\ast \varphi_{123}\right)(a,b,c)\\
        \mathbf{r}^\gamma(a, b\cdot_{\gamma}c) &= \left(\varphi^{-1}_{231}\ast \mathbf{r}^\gamma_{13}\ast\varphi_{213}\ast\mathbf{r}^\gamma_{12}\ast \varphi^{-1}_{123}\right)(a,b,c).
    \end{split}
\end{equation}
Remark here in particular that, if the Drinfeld associator is trivial, and therefore the algebra is associative, one does obtain the ordinary factorisation rule \eqref{eq:factorisationrform}. It can be noted that the Drinfeld associator, and therefore a possible long-range term, appears in the product rule \eqref{eq:nonassociativefactorisation} for the $r$-form. In particular, as we will show shortly, this additional Drinfeld associator is exactly what gives rise to the long-range deformations of the Lax operator as given in Section~\ref{sec:LaxforBQ3derivation}. Note that the non-associative factorisation rule also modifies the Yang-Baxter equation. In particular, it follows that the twisted $r$-form satisfies the non-associative Yang-Baxter equation \cite{Drinfeld1989QuasiHopf}
\begin{equation}
    \label{eq:nonassociativeYBE}
    \mathbf{r}^\gamma_{12}\ast \varphi_{312}\ast \mathbf{r}^\gamma_{13}\ast \varphi^{-1}_{132}\ast \mathbf{r}^{\gamma}_{23}\ast \varphi_{123} = \varphi_{321}\ast \mathbf{r}^\gamma_{23}\ast \varphi^{-1}_{231}\ast \mathbf{r}^{\gamma}_{13}\ast \varphi_{213}\ast \mathbf{r}^{\gamma}_{12}.
\end{equation}
Therefore, for a general twisted bialgebra $\A(\gamma)$, the Yang-Baxter equation \eqref{eq:dualYangBaxtereq} does not hold. Of course, this then also means that the cotriangular bialgebra does not generate a tower of commuting charges, and therefore does not give rise to an integrable model. However, as we will see, the twisted bialgebras that we will be considering for the long-range deformations will still have some (perturbatively) associative subalgebra, such that the Yang-Baxter equation also still holds (perturbatively).\\

In order to describe these associative subalgebras, we will introduce some additional notation. In particular, we will introduce the following grading on $\A_0$: we define $\A^{(0)}_0 \coloneq \C$ and for $n\geq 1$ we define
\begin{equation}
    \A_0^{(n)} \coloneq \mathrm{span}\left\{t_{i_1,j_1}(0)\cdot \cdots t_{i_n, j_n}(0) : 1\leq i_1,..., i_n \leq \dim V \text{ and }  1\leq j_1,..., j_n \leq \dim V\right\}.
\end{equation}
That is, $\A_0^{(n)}\subseteq \A_0$ consists of monomials in $t_{ij}(0)$ of degree $n$. As a vector space, the algebra can be decomposed as $\A_0 = \bigoplus_{n=0}^{\infty}\A_0^{(n)}$. In addition, we consider the subsets of $\A_0$ with elements that have degree less or degree more than some $n\geq 0$, given respectively by
\begin{equation}
    \A_0^{(\leq n)} \coloneq \bigoplus_{k=0}^{n}\A_0^{(n)}\quad \text{and}\quad \A_0^{(\geq n)} \coloneq \bigoplus_{k=n}^{\infty}\A_0^{(n)}.
\end{equation} 
Similarly, we let $\A^{(\geq 1)}(\R)$ be the non-unital subset of $\A(\R)$ generated by the elements $t_{ij}(0)$ and $t^{(-k)}_{ij}$ for $k \geq 2$ and $1\leq i,j\leq N$. Note that this subset is equal to the non-unital \textit{reduced algebra} (see e.g. \cite{LodayVallette2012}), such that $\A(\R) = \C \oplus \A^{(\geq 1)}(\R)$ as a vector space. The grading can now be used to split the double-crossed product algebra $\A(\R)\bowtie_{\mathbf{r}}\A_0$ into two parts. In particular, for $m\geq 1$, one can write
\begin{equation}
    \label{eq:algebradecomposition}
    \A(\R)\bowtie_{\mathbf{r}}\A_0 = \left[\A_0 \oplus \left(\A^{(\geq 1)}(\R)\bowtie_{\mathbf{r}}\A_0^{(\geq m)}\right)\right]\oplus \left[\A^{(\geq 1)}(\R)\bowtie_{\mathbf{r}}\A_0^{(\leq m-1)}\right]
\end{equation}
as a vector space, where we identify $\A_0\cong  \C\otimes \A_0\subseteq \A(\R)\bowtie_{\mathbf{r}}\A_0$ on the right side. Note here that the first term in this direct sum is a unital subalgebra of $\A(\R)\bowtie_{\mathbf{r}}\A_0$. As we will see, subalgebras of this form for appropriate values of $m$ will correspond to the unital subalgebras that remain associative after twisting for the long-range deformations.

\subsubsection{Twist for the local deformation}\label{sec:twistlocaldeformation}
We can now consider the twist of the double-crossed product algebra $\A(\R)\bowtie_{\mathbf{r}} \A_0$ that corresponds to the local operator long-range deformation. In particular, let $k\geq 2$ and $\mathfrak{m} \in \mathrm{End}(V^{\otimes k})$. We consider a twisting function
\begin{equation}
    \gamma_\mathfrak{m}(\lambda) \coloneq \epsilon\otimes \epsilon + \lambda \gamma^{(1)}_{\mathfrak{m}} + \O(\lambda^2)\in \mathrm{Hom}\big((\A(\R)\bowtie_{\mathbf{r}} \A_0)\otimes (\A(\R)\bowtie_{\mathbf{r}} \A_0), \C \big)[\![\lambda]\!],
\end{equation}
where we expand $\gamma_{\mathfrak{m}}$ around $\lambda = 0$ up to first order in $\lambda$ and the zeroth order term is just the trivial twist, i.e. $\gamma^{(0)}_{\mathfrak{m}}(a,b) = \epsilon(a)\epsilon(b)$ for $a,b\in \A(\R)\bowtie_{\mathbf{r}} \A_0$. We then define the first order correction $\gamma^{(1)}_{\mathfrak{m}}$ to be given by
\begin{equation}\label{eq:twistlocaldeformation}
    \begin{split}
        &\gamma^{(1)}_{\mathfrak{m}}(\T_{\underline{\vec a}}\otimes \T_1\cdots \T_N, \T_{\underline{\vec b}}\otimes \T_{N+1}\cdots \T_M)\\
        \coloneq& \sum_{n=1}^{N}\R^{-1}_{\underline{\vec b}(N+1,...,M)}\mathfrak{m}_{n,...,n+k-1}\R_{\underline{\vec b}(N+1,...,M)} - \R^{-1}_{(1,...,N)\underline{\vec b}}\mathfrak{m}_{n,...,n+k-1}\R_{(1,...,N)\underline{\vec b}},
    \end{split}
\end{equation}
for $0 \leq N \leq M$, and where we used the notation $\T_{\underline i} \coloneq \T_{\underline i}(u_i)$ and $\T_{i}\coloneq \T_{i}(0)$. For this particular definition, we set a term $\mathfrak{m}_{n,...,n+k-1}$ in the sum to $0$ if it ``overflows'' the element $\T_{N+1}\cdots \T_M$, i.e. $\mathfrak{m}_{n,...,n+k-1} = 0$ for $n+k-1 > M$. Of course, we should check whether $\gamma^{(1)}_{\mathfrak{m}}$ is well-defined, i.e. is invariant under the RTT-relation \eqref{eq:RTTrelations}. Indeed, it can be readily verified that
\begin{equation}
    \label{eq:consistencygammamR}
    \begin{gathered}
        \R_{\underline{i}\underline{j}}\gamma_{\mathfrak{m}}^{(1)}\left(\T_{\underline 1}\cdots \T_{\underline i}\T_{\underline j}\cdots \T_{\underline N}\otimes b, c\otimes d\right) = \gamma_{\mathfrak{m}}^{(1)}\left(\T_{\underline 1}\cdots \T_{\underline j}\T_{\underline i}\cdots \T_{\underline N}\otimes b, c\otimes d\right)\R_{\underline{i}\underline{j}}\\
        \R_{\underline{i}\underline{j}}\gamma_{\mathfrak{m}}^{(1)}\left(a\otimes b, \T_{\underline 1}\cdots \T_{\underline i}\T_{\underline j}\cdots \T_{\underline N}\otimes d\right) = \gamma_{\mathfrak{m}}^{(1)}\left(a\otimes b, \T_{\underline 1}\cdots \T_{\underline j}\T_{\underline i}\cdots \T_{\underline N}\otimes d\right)\R_{\underline{i}\underline{j}}
    \end{gathered}
\end{equation}
and
\begin{equation}
    \label{eq:consistencygammamP}
    \begin{gathered}
        \P_{ij}\gamma_{\mathfrak{m}}^{(1)}\left(a\otimes \T_{1}\cdots \T_{i}\T_{j}\cdots \T_{N}, c\otimes d\right) = \gamma_{\mathfrak{m}}^{(1)}\left(a\otimes \T_{1}\cdots \T_{j}\T_{i}\cdots \T_{N}, c\otimes d\right)\P_{ij}\\
        \P_{ij}\gamma_{\mathfrak{m}}^{(1)}\left(a\otimes b, c\otimes \T_{1}\cdots \T_{i}\T_{j}\cdots \T_{N}\right) = \gamma_{\mathfrak{m}}^{(1)}\left(a\otimes b, c\otimes \T_{1}\cdots \T_{j}\T_{i}\cdots \T_{N}\right)\P_{ij},
    \end{gathered}
\end{equation}
for $a,c\in \A(\R)$ and $b,d\in \A_0$, and we suppressed the explicit dependence on the spectral parameters in \eqref{eq:consistencygammamR} in our notation, i.e. $\T_{\underline{i}} \coloneq\T_{\underline{i}}(u_i)$ and $\T_i\coloneq\T_i(0)$. Note here that the relations \eqref{eq:consistencygammamP} for the twist \eqref{eq:BQktwistelement} only hold precisely because the second term in the double-crossed product algebra corresponds to $\A_0$, i.e. its elements have a trivial spectral parameter, such that the braiding operator is just given by the permutation operator. In particular, a twist of the same form as \eqref{eq:BQktwistelement} could not have been (consistently) defined on the quantum double $\A(\R)\bowtie_{\mathbf{r}}\A(\R)$. Also note that $\gamma^{(1)}_\mathfrak{m}(a\otimes 1, c\otimes 1) = 0$ and $\gamma^{(1)}_\mathfrak{m}(1\otimes b, 1\otimes d) = 0$, such that $\A(\R)$ and $\A_0$ as subalgebras $\A(\R), \A_0 \subseteq \A(\R)\bowtie_{\mathbf{r}} \A_0$ remain undeformed.\\

Importantly, the twisted bialgebra $(\A(\R)\bowtie_{\mathbf{r}}\A_0)(\gamma_{\mathfrak{m}})$ is non-associative at first order in $\lambda$. In particular, let us also expand the Drinfeld associator in $\lambda$, i.e. write $\varphi_{\mathfrak{m}}(\lambda) = \varphi^{(0)}_{\mathfrak{m}} + \lambda \varphi^{(1)}_{\mathfrak{m}} + \O(\lambda^2)$ with $\varphi^{(0)}_{\mathfrak{m}}(a,b,c) = \epsilon(a)\epsilon(b)\epsilon(c)$. It can then be directly calculated that the non-zero terms in the first-order correction to the Drinfeld associator for the elements $1\otimes \T_{1}(0)$, $\T_{\underline{\vec a}}(u_\vec a)\otimes 1$ and $1\otimes \T_{2}(0)\cdots \T_{k}(0)$ are given by
\begin{equation}\label{eq:nonvanihsingDrinfeldlocaloperator}
    \begin{split}
        \varphi_{\mathfrak{m}}^{(1)}(1\otimes \T_{1}, \T_{\underline{\vec a}}\otimes 1, 1\otimes \T_{2}\cdots \T_{k}) = \gamma^{(1)}_{\mathfrak{m}}(1\otimes \T_1, \T_{\underline{\vec a}}\otimes \T_2\cdots \T_k)\neq 0,
    \end{split}
\end{equation}
where we used the notation $\T_{\underline i} \coloneq \T_{\underline i}(u_i)$ and $\T_{i}\coloneq \T_{i}(0)$. It therefore follows from \eqref{eq:nonassociativeYBE} that the Yang-Baxter equation does not hold in general. One might be worried whether this will then give rise to an integrable model. Luckily, the twisted bialgebra $(\A(\R)\bowtie_{\mathbf{r}}\A_0)(\gamma_{\mathfrak{m}})$ does still have a (large) associative substructure. In particular, as shown in Appendix~\ref{sec:proofsforlocaldeformation}, we have the following result:
\begin{proposition}\label{prop:vanishingDrinfeldassociatorlocaldef}
    Let $a,b,c\in \A_0 \oplus \left(\A^{(\geq 1)}(\R)\bowtie_{\mathbf{r}}\A_0^{(\geq k-1)}\right)$, then $\varphi_{\mathfrak{m}}^{(1)}(a, b, c) = 0$.
\end{proposition}

We can now use this result to construct the Lax operator and $R$-matrix and show that they satisfy the RLL relations, such that the integrability of the corresponding model is ensured. Firstly, we consider the vector $\vec b \coloneq (b_1,...,b_{k-1})$ and let $\T_{\vec b}(\vec 0) = \T_{b_1}(0)\cdots \T_{b_{k-1}}(0)$. It can then be directly calculated from \eqref{eq:twistlocaldeformation} that
\begin{equation}\label{eq:evaluationstwistlocaloperator}
    \begin{split}
        \gamma^{(1)}_{\mathfrak{m}}(\T_{\underline{a}}\otimes \T_{\vec b}, 1\otimes \T_1) &= 0,\\
        \gamma^{(1)}_{\mathfrak{m}}(1\otimes \T_1, \T_{\underline{a}}\otimes \T_{\vec b}) &= \R^{-1}_{\underline{a}\vec b}\mathfrak{m}_{1,b_1,...,b_{k-1}}\R_{\underline{a}\vec b} - \R^{-1}_{1\underline{a}}\mathfrak{m}_{1,b_1,...,b_{k-1}}\R_{1\underline{a}}.
    \end{split}
\end{equation}
In light of \eqref{eq:undeformedincreasedauxiliaryLaxoperator}, we define the deformed Lax operator on the increased auxiliary space to be given by
\begin{equation}\label{eq:Laxoperatorlocaldef}
    \L^{\gamma_{\mathfrak{m}}}_{(\underline{a},b_1,...,b_{k-1})1}(u) \coloneq \mathbf{r}^{\gamma_{\mathfrak{m}}}(\T_{\underline{a}}(u)\otimes \T_{b_1}(0)\cdots \T_{b_{k-1}}(0), 1\otimes \T_1(0)).
\end{equation}
We can now use \eqref{eq:evaluationstwistlocaloperator} to evaluate this Lax operator. In particular, using the definition of $\mathbf{r}^{\gamma}$ as given in \eqref{eq:deformedrform}, we get
\begin{equation}
    \begin{split}
        \L_{(\underline{a}\vec b)1}^{\gamma_\mathfrak{m}}=  \left(1+\lambda\Big(\R^{-1}_{\underline{a}\vec b}\mathfrak{m}_{1,b_1,...,b_{k-1}}\R_{\underline{a}\vec b}- \R^{-1}_{1\underline{a}}\mathfrak{m}_{1,b_1,...,b_{k-1}}\R_{1\underline{a}}\Big)\right)\L_{(\underline{a}\vec b)1} + \O(\lambda^2),
    \end{split}
\end{equation}
Note here that this Lax operator is exactly equal to the Lax operator obtained from the local operator long-range deformation as given in \eqref{eq:localoperatorLaxoperator}! Alternatively, we could have used the factorisation property \eqref{eq:nonassociativefactorisation} and the fact that $\T_{\underline{a}}\otimes \T_{\vec b} = (\T_{\underline{a}}\otimes 1)(1\otimes \T_{\vec b})$ to write
\begin{equation}\label{eq:localLaxoperatorwithDrinfeld}
    \begin{split}
        \L_{(\underline{a}\vec b)1}^{\gamma_\mathfrak{m}}  = \varphi_\mathfrak{m}(1\otimes \T_1, \T_{\underline{a}}\otimes 1, 1\otimes \T_{\vec b})\cdot \mathbf{r}^{\gamma_\mathfrak{m}}(\T_{\underline{\vec a}}\otimes 1, 1\otimes \T_1)\cdot \mathbf{r}^{\gamma_\mathfrak{m}}(1\otimes \T_{\vec b}, 1\otimes \T_1),
    \end{split}
\end{equation}
where in the second equality we used that $\mathbf{r}^{\gamma_{\mathfrak{m}}}(1\otimes \T_{\vec b}, 1\otimes \T_1) = \P_{\vec b,1}$ (note that $\A_0\subseteq \A(\R)\bowtie_\mathbf{r}\A_0$ remains undeformed under the twisting), such that the second and third Drinfeld associators in \eqref{eq:nonassociativefactorisation} cancel each other. In particular, it can be seen that the long-range term in the Lax operator is exactly given by the non-vanishing Drinfeld associator as calculated in \eqref{eq:nonvanihsingDrinfeldlocaloperator}. This shows that the long-range nature of this deformation is exactly due to the breaking of the associativity of the double-crossed product bialgebra $(\A(\R)\bowtie_\mathbf{r}\A_0)(\gamma_{\mathfrak{m}})$.\\

It can be noted now that the Lax operator \eqref{eq:Laxoperatorlocaldef} corresponds to a representation of the matrix corepresentation $\T_{\underline{a}}(u)\otimes \T_{\vec b}(0)$ of $\left(\A^{(\geq 1)}(\R)\bowtie_{\mathbf{r}}\A_0^{(\geq k-1)}\right)(\gamma_{\mathfrak{m}})$. Since it follows from Proposition~\ref{prop:vanishingDrinfeldassociatorlocaldef} that this subalgebra is associative up to first order in $\lambda$, this also means that there is an $R$-matrix that satisfies the Yang-Baxter equation, and such that the RLL-relation is satisfied. In particular, let also $\vec d = (d_1,..., d_{k-1})$, then the $R$-matrix is given by
\begin{equation}
    \R^{\gamma_{\mathfrak{m}}}_{(\underline{a}\vec b)(\underline{c}\vec d)}(u,v) \coloneq \mathbf{r}^{\gamma_{\mathfrak{m}}}(\T_{\underline{a}}(u)\otimes \T_{\vec b}(\vec 0), \T_{\underline{c}}(v)\otimes \T_{\vec d}(\vec 0)).
\end{equation}
It now follows immediately from \eqref{eq:nonassociativeYBE} and Proposition~\ref{prop:vanishingDrinfeldassociatorlocaldef} that
\begin{equation}\label{eq:localoperatorRLLrelation}
    \R^{\gamma_{\mathfrak{m}}}_{(\underline{a}\vec b)(\underline{c}\vec d)}(u,v)\L^{\gamma_{\mathfrak{m}}}_{(\underline{a}\vec b)1}(u)\L^{\gamma_{\mathfrak{m}}}_{(\underline{c}\vec d)1}(v) = \L^{\gamma_{\mathfrak{m}}}_{(\underline{c}\vec d)1}(v)\L^{\gamma_{\mathfrak{m}}}_{(\underline{a}\vec b)1}(u)\R^{\gamma_{\mathfrak{m}}}_{(\underline{a}\vec b)(\underline{c}\vec d)}(u,v) + \O(\lambda^2),
\end{equation}
and
\begin{equation}\label{eq:localoperatorYBE}
    \begin{split}
        &\R^{\gamma_\mathfrak{m}}_{(\underline{a}\vec b)(\underline{c}\vec d)}(u,v)\R^{\gamma_\mathfrak{m}}_{(\underline{a}\vec b)(\underline{e}\vec f)}(u,w)\R^{\gamma_\mathfrak{m}}_{(\underline{c}\vec d)(\underline{e}\vec f)}(v,w)\\
        = &\R^{\gamma_\mathfrak{m}}_{(\underline{c}\vec d)(\underline{e}\vec f)}(v,w)\R^{\gamma_\mathfrak{m}}_{(\underline{a}\vec b)(\underline{e}\vec f)}(u,w) \R^{\gamma_\mathfrak{m}}_{(\underline{a}\vec b)(\underline{c}\vec d)}(u,v)+ \O(\lambda^2),
    \end{split}
\end{equation}
where $\vec{f} = (f_1,...,f_{k-1})$ and $u,v,w\in \C$. The above equations now ensure the integrability of the corresponding long-range deformed spin chain up to first order in the deformation parameter.\\

Lastly, in order for the Lax operator \eqref{eq:Laxoperatorlocaldef} to be a proper representation of the algebra $(\A(\R)\bowtie_{\mathbf{r}}\A_0)(\gamma_\mathfrak{m})^{\text{cop}}$, it also better be true that a product of elements in the second coordinate factorises. In particular, it can be noted here that $\mathbf{r}^{\gamma_\mathfrak{m}}(\T_{\underline a}\otimes \T_{\vec b}, 1\otimes(\cdot))$ defines a linear map $\A_0\to \C$, which is just the $\ell$-functional for the twisted double-crossed product bialgebra. In light of Section~\ref{sec:coquasitriangularityandLfunctionals}, this function corresponds to an element in $\A_0^{*}$ if it satisfies the correct product factorisation property. Here, it follows immediately from Proposition~\ref{prop:vanishingDrinfeldassociatorlocaldef} and \eqref{eq:nonassociativefactorisation} that 
\begin{equation}\label{eq:factorisationlocaloperatorLax}
    \mathbf{r}^{\gamma_\mathfrak{m}}(\T_{\underline a}\otimes \T_{\vec b}, 1\otimes \T_{1}\T_2) = \mathbf{r}^{\gamma_\mathfrak{m}}(\T_{\underline a}\otimes \T_{\vec b}, 1\otimes \T_2)\mathbf{r}^{\gamma_\mathfrak{m}}(\T_{\underline a}\otimes \T_{\vec b}, 1\otimes \T_1) + \O(\lambda^2),
\end{equation}
such that indeed also the Lax operator factorises in this way at first order in $\lambda$. In fact, since all the Drinfeld associators corresponding to higher products in $\A_0$ vanish, the representation of $\T_{\underline{a}}(u)\otimes \T_{\vec b}(0)$ on a spin chain of length $L$ factorises into a product of Lax operators, similar to \eqref{eq:monodromymatrix}. Therefore, the monodromy matrix for the long-range models is also just given by a product of Lax operators.\\

We remark here that it is now not necessarily true that the $R$-matrix factorises into a product of Lax operators, since
\begin{equation}
    \label{eq:nonfactorisationRmatrix}
    \mathbf{r}^{\gamma_\mathfrak{m}}\left(\T_{\underline{a}}(u)\otimes \T_{\vec b}(\vec 0), \T_{\underline{c}}(\vec 0)\otimes \T_{\vec d}(\vec 0)\right)\neq \mathbf{r}^{\gamma_{\mathfrak{m}}}\left(\T_{\underline{a}}(u)\otimes \T_{\vec b}(\vec 0), 1\otimes\T_{c}(\vec 0)\T_{\vec d}(\vec 0)\right).
\end{equation}
Therefore, we see that our Lax operator and $R$-matrix do not have the same properties as described in \cite{Gombor2021, Gombor2022WrappingCorrections}. As also noted in Section~\ref{sec:LaxforBQ3derivation}, this is due to the fact that our Lax operators are not regular, i.e. do not reduce to the permutation operator at $u=0$, which is something that was assumed to be true in \cite{Gombor2021, Gombor2022WrappingCorrections}. The corresponding twisted $R$-matrices are regular and reduce the permutation operator for $u,v = 0$. Note that one could now in principle also use the twisted $R$-matrix as a Lax operator, i.e. define $\L^{\lambda}_{(\underline{a}\vec b)(\underline{1},2,...,k)}(u,v) \coloneq \R^{\gamma_{\mathfrak{m}}}_{(\underline{a} \vec{b})(\underline{1},2,..,k)}(u,v)$.\footnote{It is, however, not guaranteed that the corresponding integrable model continues to be integrable up to second order in $\lambda$. In fact, since the auxiliary space of the $R$-matrix is required to increase when going to higher-order corrections, such that also the value of $k$ increases, it is expected that the integrable model defined through this $R$-matrix as a Lax operator is only \textit{quasi-integrable}.} Since the factorisation property \eqref{eq:nonfactorisationRmatrix} does not hold, this then gives, in general, rise to inequivalent charges with respect to the ones obtained from $\L^{\gamma_\mathfrak{m}}_{(\underline{a}\vec b)1}(u)$. Also, one might wonder whether a Lax operator could alternatively be defined through an $\ell$-functional $\A(\R)\to \C$, i.e. set $\mathbf{r}^{\gamma_\mathfrak{m}}(\T_{\underline a}(u)\otimes \T_{\vec b}(0), \T_{\underline 1}(v)\otimes 1)$ to be the Lax operator. However, it can be easily verified that the corresponding Drinfeld associators do not vanish, such that the RLL-relations are not satisfied, and therefore no integrable model is obtained in this way.

\subsubsection{Twist for the boost deformation}\label{sec:twistboostdeformation}
Next, we will discuss the twist corresponding to the long-range deformation that is generated by the boost operator. A lot of its discussion will be similar to Section~\ref{sec:twistlocaldeformation}, so we will be more brief. For $k\geq 3$, we consider a twisting function
\begin{equation}
    \gamma_k(\lambda) \coloneq \epsilon\otimes \epsilon + \lambda \gamma^{(1)}_k + \O(\lambda^2)\in \mathrm{Hom}\big((\A(\R)\bowtie_{\mathbf{r}} \A_0)\otimes (\A(\R)\bowtie_{\mathbf{r}} \A_0), \C \big)[\![\lambda]\!],
\end{equation}
for which first-order term is given by
\begin{equation}
    \label{eq:BQktwistelement}
    \begin{split}
        &\gamma_k^{(1)}\left(\T_{\underline {\vec a}}\otimes \T_{1}\cdots \T_{N}, \T_{\underline{\vec b}}\otimes \T_{N+1}\cdots \T_{M}\right) \coloneq \sum_{n=1}^{N}\R^{-1}_{(n+1,...,N)\underline{\vec b}}{}^r\Q^{(k)}_{n,\underline{\vec b},(n+1,...,M)}\R_{(n+1,...,N)\underline{\vec b}},
    \end{split}
\end{equation}
for $0< N< M$, and where we used the notation $\T_{\underline i} \coloneq \T_{\underline i}(u_i)$ and $\T_{i}\coloneq \T_{i}(0)$, and ${}^r\Q^{(k)}$ is the reduced algebraic charge density as given in Definition~\ref{def:reducedalgebraicchargedensities}. Similar to the local operator deformation, this twist element satisfies the relations given in \eqref{eq:consistencygammamR} and \eqref{eq:consistencygammamP}, such that $\gamma^{(1)}_k$ is invariant under the RTT relations \eqref{eq:RTTrelations} and therefore well-defined. Note here again that one could only have such a well-defined twist of this form precisely because the second term in the double-crossed product algebra corresponds to $\A_0$, i.e. its elements have a trivial spectral parameter, such that the braiding operator is just given by the permutation operator. Also, remark here that $\gamma^{(1)}_k(a\otimes 1, c\otimes 1) = 0$ and $\gamma^{(1)}_k(1\otimes b, 1\otimes d) = 0$, such that $\A(\R)$ and $\A_0$ as subalgebras $\A(\R), \A_0 \subseteq \A(\R)\bowtie_{\mathbf{r}} \A_0$ remain undeformed.\\

The twisted algebra $(\A(\R)\bowtie_{\mathbf{r}}\A_0)(\gamma_k)$ is, again, in general non-associative at first order in $\lambda$, where we consider the expansion for the Drinfeld associator $\varphi_k = \varphi_k^{(0)} + \lambda\varphi_k^{(1)} + \O(\lambda^2)$ with $\varphi^{(0)}_k(a,b,c) = \epsilon(a)\epsilon(b)\epsilon(c)$. For example, it can be directly calculated that the non-zero terms in the Drinfeld associator for the elements $1\otimes \T_1$, $\T_{\underline{\vec a}}\otimes 1$ and $1\otimes \T_{2}\cdots \T_{k-1}$ are given by
\begin{equation}
    \label{eq:nontrivialDrinfeldAssociatorBoost}
    \begin{split}
        \varphi^{(1)}_k(1\otimes \T_1, \T_{\underline{\vec a}}\otimes 1, 1\otimes \T_{2}\cdots \T_{k-1}) = \gamma^{(1)}_k(1\otimes \T_1,\T_{\underline{\vec a}}\otimes \T_2\cdots \T_{k-1}) - \gamma^{(1)}_k(1\otimes \T_1, \T_{\underline{\vec a}}\otimes 1),
    \end{split}
\end{equation}
However, $(\A(\R)\bowtie_{\mathbf{r}}\A_0)(\gamma_k)$ still has a large associative subalgebra. Namely, as shown in Appendix~\ref{sec:proofsforboostdeformation}, we have the following result:
\begin{proposition}\label{prop:vanishingDrinfeldassociatorBQk}
    Let $a,b,c\in \A_0 \oplus \left(\A^{(\geq 1)}(\R)\bowtie_{\mathbf{r}}\A_0^{(\geq k-2)}\right)$, then $\varphi_k^{(1)}(a, b, c) = 0$.
\end{proposition}
This result can then be used to construct a Lax operator and an $R$-matrix that satisfy the RLL-relations and the Yang-Baxter equation. Firstly, let $\vec b = (b_1,...,b_{k-2})$ and define $\T_{\vec b}(\vec 0) = \T_{b_1}(0)\cdots \T_{b_{k-2}}(0)$, then it can be calculated that
\begin{equation}\label{eq:evaluationtwistboostoperator}
    \gamma^{(1)}_k\left(\T_{\underline{\vec a}}\otimes \T_{\vec b}, 1\otimes \T_{1}\right) = 0 \text{ and }\gamma^{(1)}_k\left(1\otimes \T_{1}, \T_{\underline{\vec{a}}}\otimes \T_{\vec b}\right) = {}^r\Q^{(k)}_{1,\underline{\vec a}, \vec b}.
\end{equation}
In light of \eqref{eq:undeformedincreasedauxiliaryLaxoperator}, the deformed Lax operator on the increased auxiliary space is defined to be equal to
\begin{equation}
    \label{eq:boostLaxoperatorfromdeformation}
    \begin{split}
         \L^{\gamma_k}_{(\underline{a},b_1,...,b_{k-2})1}(u) &\coloneq \mathbf{r}^{\gamma_k}(\T_{\underline{a}}(u)\otimes \T_{b_1}(0)\cdots \T_{b_{k-2}}(0), 1\otimes \T_1(0)).
    \end{split}
\end{equation}
We can now use \eqref{eq:evaluationtwistboostoperator} and the definition of $\mathbf{r}^{\gamma}$ as given in \eqref{eq:deformedrform} to obtain
\begin{equation}
    \begin{split}
        \L_{(\underline{a}\vec b)1}^{\gamma_k} =  \left(1+\lambda\cdot {}^r\Q^{(k)}_{1,\underline{a}, \vec b}\right)\L_{(\underline{a}\vec b)1} + \O(\lambda^2),
    \end{split}
\end{equation}
Note here that this Lax operator is exactly equal to the Lax operator obtained from the boost operator long-range deformation as given in \eqref{eq:boostoperatorLaxoperator}! Moreover, we could have also written this Lax operator in terms of the Drinfeld associator like in \eqref{eq:localLaxoperatorwithDrinfeld}, where it could then be seen that the long-range term in the Lax operator is exactly given by the non-vanishing Drinfeld associator as calculated in \eqref{eq:nontrivialDrinfeldAssociatorBoost}, which again shows that the long-range nature of this deformation is exactly due to the breaking of the associativity of the double-crossed product bialgebra $(\A(\R)\bowtie_\mathbf{r}\A_0)(\gamma_k)$.\\

Also here it can now be noted that the Lax operator corresponds to a representation of the matrix corepresentation $\T_{\underline{a}}(u)\otimes \T_{\vec b}(0)$ of $\left(\A^{(\geq 1)}(\R)\bowtie_{\mathbf{r}}\A_0^{(\geq k-2)}\right)(\gamma_k)$. Since it follows from Proposition~\ref{prop:vanishingDrinfeldassociatorBQk} that this subalgebra is associative up to first order in $\lambda$, there is also an $R$-matrix such that the RLL relation is satisfied. In particular, let also $\vec d = (d_1,...,d_{k-2})$, then the corresponding $R$-matrix is given by
\begin{equation}\label{eq:twistedRmatrixboost}
    \begin{split}
        \R^{\gamma_k}_{(\underline{a}\vec b)(\underline{c}\vec d)}(u,v)&\coloneq \mathbf{r}^{\gamma_k}\left(\T_{\underline{a}}(u)\otimes \T_{\vec b}(\vec 0), \T_{\underline{c}}(v)\otimes \T_{\vec d}(\vec 0)\right)
    \end{split}
\end{equation}
It now follows from Proposition~\ref{prop:vanishingDrinfeldassociatorBQk} that the RLL-relation and the Yang-Baxter equation, similar to \eqref{eq:localoperatorRLLrelation} and \eqref{eq:localoperatorYBE}, are satisfied for the Lax operator $\L^{\gamma_k}$ \eqref{eq:boostLaxoperatorfromdeformation} and the above $R$-matrix $\R^{\gamma_k}$. Moreover, a similar factorisation property for the Lax operator as given in \eqref{eq:factorisationlocaloperatorLax} holds for $\L^{\gamma_k}$, such that this Lax operator also defines a proper representation of $\T_{\underline{a}}\otimes \T_{\vec b}$ in the twisted algebra $(\A(\R)\bowtie_\mathbf{r}\A_0)(\gamma_k)$. In particular, the monodromy matrix for the boost deformation is just given by a product of Lax operators. Similar to \eqref{eq:nonfactorisationRmatrix}, the $R$-matrix does \textit{not} factorise into a product of Lax operators.

\subsubsection{Twist for the bilocal deformation}\label{sec:twistbilocaldeformation}
Lastly, we will discuss the twisting element for the long-range deformation that is generated by the bilocal operators. This discussion will again be very similar to the ones in Sections~\ref{sec:twistlocaldeformation} and \ref{sec:twistboostdeformation}, such that we will be more brief. For $2\leq k < \ell$, we consider a twisting function
\begin{equation}
    \gamma_{k|\ell}(\lambda) \coloneq \epsilon\otimes \epsilon + \lambda \gamma^{(1)}_{k|\ell} + \O(\lambda^2)\in \mathrm{Hom}\big((\A(\R)\bowtie_{\mathbf{r}} \A_0)\otimes (\A(\R)\bowtie_{\mathbf{r}} \A_0), \C \big)[\![\lambda]\!],
\end{equation}
for which first-order term is given by
\begin{equation}\label{eq:twistbilocal}
    \begin{split}
        &\gamma^{(1)}_{k | \ell}(\T_{\underline{\vec a}}\otimes \T_1\cdots \T_N, (1\otimes \T_{N+1}\cdots \T_M)(\T_{\underline{\vec b}}\otimes 1))\\
        &\coloneq \R^{-1}_{\underline{\vec a}(1,...,N+\ell-2)}{}^r\tilde{\Q}^{(k)}_{(1,...,N+\ell-2),\underline{\vec a},N+\ell-1}\R_{\underline{\vec a}(1,...,N+\ell-2)}\\
        &\qquad \qquad \qquad \qquad \qquad  \times\R_{\underline{\vec b}(N+3-\ell,...,M)}{}^r\Q^{(\ell)}_{N+2-\ell, \underline{\vec b},(N+3-\ell,...M)}\R^{-1}_{\underline{\vec b}(N+3-\ell,...,M)}\\
        &+\sum_{n=N+2,}^{N+\ell-1}\sum_{m=N+2-\ell}^{n-\ell}\R^{-1}_{\underline{\vec a}(1,...,M)}[\tilde{\mathfrak{q}}^{(k)}_{n-k+1,...,n}, \R_{\underline{\vec a}(1,...,M)}][\mathfrak{q}^{(\ell)}_{m,...,m+\ell-1}, \R_{\underline{\vec b}(1,...,M)}]\R^{-1}_{\underline{\vec b}(1,...,M)}\\
        &+\sum_{n=1}^{N}\R^{-1}_{\underline{\vec a}(1,...,M)}\tilde{\mathfrak{q}}^{(k)}_{n+\ell-k,...,n+\ell-1}\R_{\underline{\vec a}(1,....M)}\R_{\underline{\vec b}(n+1,...,M)}{}^r\Q^{(\ell)}_{n,\underline{\vec b},(n+1,...,M)}
        \R_{\underline{\vec b}(n+1,...,M)}^{-1}\\
        &+\sum_{\mathclap{n=N+1}}^{M}\ \R^{-1}_{\underline{\vec a}(1,...,n-1)}{}^r\tilde{\Q}^{(k)}_{(1,...,n-1),\underline{\vec a}, n}\R_{\underline{\vec a}(1,...,n-1)}\R_{\underline{\vec b}(1,...,M)}\mathfrak{q}^{(\ell)}_{n-\ell+1,...,n}\R^{-1}_{\underline{\vec b}(1,...,M)}.
    \end{split}
\end{equation}
for $0< N< M$, where we used the notation $\T_{\underline i} \coloneq \T_{\underline i}(u_i)$ and $\T_{i}\coloneq \T_{i}(0)$, $\tilde{\mathfrak{q}}^{(k)}$ and $\mathfrak{q}^{(\ell)}$ are the charge densities as defined in \eqref{eq:chargedensitiesfromalgebraic}, and ${}^r\tilde{\Q}^{(k)}$ and ${}^r\Q^{(\ell)}$ are the reduced algebraic charge density as given in Definition~\ref{def:reducedalgebraicchargedensities}. Note here the ordering in the second entry of the twist element, which we defined in terms of \begin{equation}
    (1\otimes \T_{N+1}\cdots \T_M)(\T_{\underline{\vec b}}\otimes 1) = \R^{-1}_{(N+1,...,M)\underline{\vec b}}\left(\T_{\underline{\vec b}}\otimes \T_{N+1}\cdots \T_M\right)\R_{(N+1,...,M)\underline{\vec b}},
\end{equation}
in order to shorten the notation for the (rather long) twisting element. Also, similar to the twist for the local operator, a term $\tilde{\mathfrak{q}}^{(k)}_{n-1+1,...,n}$ or $\mathfrak{q}^{(\ell)}_{m,...,m+\ell-1}$ might ``overflow'' the elements $\T_{1},..., \T_{M}$, in which case we set the overflowing term to 0. For example, we set $\mathfrak{q}^{(\ell)}_{m,...,m+\ell-1} = 0$ if $m+\ell-1 > M$ and $\tilde{\mathfrak{q}}^{(k)}_{n-k+1,...,n} = 0$ if $n-k+1 < 1$. Also, remark here that, under the substitution $\tilde{\mathfrak{q}}^{(k)} \to 1$ and ${}^r\tilde{\Q}^{(k)} \to 0$, this twisting element reduces to the twist $\gamma_\ell^{(1)}$ for the boost operator deformation as given in \eqref{eq:BQktwistelement}. Therefore, one might view the boost deformation as a special case of the bilocal deformation.\\

Similar to the local operator and boost operator deformations, this twist element satisfies the relations given \eqref{eq:consistencygammamR} and \eqref{eq:consistencygammamP}, such that the twist is invariant under the RTT-relations \eqref{eq:RTTrelations} and therefore well-defined. Moreover, the only reason one could define such a twist is that the second term in the double-crossed product algebra corresponds to $\A_0$, and therefore has trivial spectral parameters. Also, remark here that $\gamma^{(1)}_{k|\ell}(a\otimes 1, c\otimes 1) =0$ and $\gamma^{(1)}_{k|\ell}(1\otimes b, 1\otimes d) =0$, such that the subalgebra $\A(\R), \A_0\subseteq \A(\R)\bowtie_{\mathbf{r}}\A_0$ remain undeformed. We should note here that the representation of the twist element as given in \eqref{eq:twistbilocal} is not unique. For example, one can rewrite terms using the Sutherland equations as given in Proposition~\ref{prop:higherorderSutherlandEquations}.\\

Also here, the twisted algebra $(\A(\R)\bowtie_{\mathbf{r}}\A_0)(\gamma_{k|\ell})$ is in general not associative at first order in $\lambda$, where we similarly consider the expansion for the Drinfeld associator $\varphi_{k|\ell} = \varphi_{k|\ell}^{(0)} + \lambda\varphi_{k|\ell}^{(1)} + \O(\lambda^2)$ with $\varphi^{(0)}_{k|\ell}(a,b,c) = \epsilon(a)\epsilon(b)\epsilon(c)$. For example, it can be directly calculated that the Drinfeld associator for the elements $1\otimes \T_1$, $\T_{\underline{\vec a}}\otimes 1$ and $1\otimes \T_{2}\cdots \T_{k-1}$ is non-vanishing. However, $(\A(\R)\bowtie_{\mathbf{r}}\A_0)(\gamma_{k|\ell})$ still has a large associative subalgebra. Namely, as shown in Appendix~\ref{sec:proofsforbilocal}, we have the following result:
\begin{proposition}\label{prop:vanishingdrinfeldassociatorbilocal}
    Let $a,b,c\in \A_0 \oplus \left(\A^{(\geq 1)}(\R)\bowtie_{\mathbf{r}}\A_0^{(\geq \ell-1)}\right)$, then $\varphi_{k|\ell}^{(1)}(a, b, c) = 0$.
\end{proposition}
\noindent Under the substitution $\tilde{\mathfrak{q}}^{(k)} \to 1$ and ${}^r\tilde{\Q}^{(k)} \to 0$, one can then also view Proposition~\ref{prop:vanishingDrinfeldassociatorBQk} for the boost operator deformation as a special case of the above Proposition.\\

This result can then again be used to construct a Lax operator and an $R$-matrix that satisfy the RLL-relations and the Yang-Baxter equation. Firstly, let $\vec b = (b_1,...,b_{\ell-1})$ and define $\T_{\vec b}(\vec 0) = \T_{b_1}(0)\cdots \T_{b_{\ell-1}}(0)$, then it can be calculated that
\begin{equation}\label{eq:evaluationtwistbilocaloperator}
    \begin{split}
        \gamma^{(1)}_{k|\ell}\left(\T_{\underline{a}}\otimes \T_{\vec b}, 1\otimes \T_{1}\right) &= \R^{-1}_{\underline{a}\vec b}\ {}^r\tilde{\Q}^{(k)}_{\vec b, \underline{a}, 1} \R_{\underline{a}\vec b}\mathfrak{q}^{(\ell)}_{b_1,...,b_{\ell-1},1},\\
        \gamma^{(1)}_{k|\ell}\left(1\otimes \T_{1}, \T_{\underline{a}}\otimes \T_{\vec b}\right) &= \R^{-1}_{\underline{a}\vec b}\tilde{\mathfrak{q}}^{(k)}_{b_{\ell-k},...,b_{\ell-1}}\R_{\underline{a}\vec b}\ {}^r\Q^{(\ell)}_{1,\underline{a}, \vec b}.
    \end{split}
\end{equation}
In light of \eqref{eq:undeformedincreasedauxiliaryLaxoperator}, the deformed Lax operator on the increased auxiliary space is then defined to be given by
\begin{equation}
    \label{eq:bilocalLaxoperatorfromdeformation}
    \begin{split}
         \L^{\gamma_{k|\ell}}_{(\underline{a},b_1,...,b_{\ell-2})1}(u) &\coloneq \mathbf{r}^{\gamma_k}(\T_{\underline{a}}(u)\otimes \T_{b_1}(0)\cdots \T_{b_{\ell-1}}(0), 1\otimes \T_1(0)).
    \end{split}
\end{equation}
We can now use \eqref{eq:evaluationtwistbilocaloperator} and the definition of $\mathbf{r}^{\gamma}$ as given in \eqref{eq:deformedrform} to obtain
\begin{equation}
    \begin{split}
         \L_{(\underline{a}\vec b)1}^{\gamma_{k|\ell}} &= \L_{(\underline{a}\vec b)1} +  \O(\lambda^2)\\
        &+\lambda\left(\R^{-1}_{\underline{a}\vec b}\tilde{\mathfrak{q}}^{(k)}_{b_{\ell-k},...,b_{\ell-1}}\R_{\underline{a}\vec b}{}^r\Q^{(\ell)}_{1,\underline{a}, \vec b}\L_{(\underline{a}\vec b)1} - \L_{(\underline{a}\vec b)1}\R^{-1}_{\underline{a}\vec b}{}^r\tilde{\Q}^{(k)}_{\vec b, \underline{a}, 1}\R_{\underline{a}\vec b}\mathfrak{q}^{(\ell)}_{b_1,...,b_{\ell-1},1}\right),
    \end{split}
\end{equation}
Note here that this Lax operator is exactly equal to the Lax operator obtained from the bilocal operator long-range deformation as given in \eqref{eq:bilocaloperatorLaxoperator}!  Moreover, we could have also written this Lax operator in terms of the Drinfeld associator like in \eqref{eq:localLaxoperatorwithDrinfeld}, where it could then be seen that the long-range term in the Lax operator is exactly given by the non-vanishing Drinfeld associator.\\

Similar to before, the Lax operator corresponds to a representation of the matrix corepresentation $\T_{\underline{a}}(u)\otimes \T_{\vec b}(0)$ of $\left(\A^{(\geq 1)}(\R)\bowtie_{\mathbf{r}}\A_0^{(\geq \ell-2)}\right)(\gamma_{k|\ell})$. Again, it follows from Proposition~\ref{prop:vanishingdrinfeldassociatorbilocal} that this subalgebra is associative up to first order in $\lambda$, such that there is an $R$-matrix such that the RLL-relation is satisfied. In particular, let also $\vec d = (d_1,...,d_{\ell-1})$, then the corresponding $R$-matrix is given by
\begin{equation}\label{eq:twistedRmatrixbilocal}
    \begin{split}
        \R^{\gamma_{k|\ell}}_{(\underline{a}\vec b)(\underline{c}\vec d)}(u,v)&\coloneq \mathbf{r}^{\gamma_{k|\ell}}\left(\T_{\underline{a}}(u)\otimes \T_{\vec b}(\vec 0), \T_{\underline{c}}(v)\otimes \T_{\vec d}(\vec 0)\right)
    \end{split}
\end{equation}
It now follows from Proposition~\ref{prop:vanishingdrinfeldassociatorbilocal} that the RLL-relation and the Yang-Baxter equation, similarly to \eqref{eq:localoperatorRLLrelation} and \eqref{eq:localoperatorYBE}, are satisfied for the Lax operator $\L^{\gamma_{k|\ell}}$ \eqref{eq:bilocalLaxoperatorfromdeformation} and the above $R$-matrix $\R^{\gamma_{k|\ell}}$. Moreover, a similar factorisation property for the Lax operator as given in \eqref{eq:factorisationlocaloperatorLax} holds for $\L^{\gamma_{k|\ell}}$, such that this Lax operator also defines a proper representation of $\T_{\underline{a}}\otimes \T_{\vec b}$ in the twisted algebra. In particular, the monodromy matrix for the bilocal deformation is just given by a product of Lax operators. Similar to \eqref{eq:nonfactorisationRmatrix}, the $R$-matrix does \textit{not} factorise into a product of Lax operators.\\

As remarked in Sections~\ref{sec:classificationLongrangeDeformations} and \ref{sec:LaxforBQ3derivation}, if the spin chain has an additional Lie algebra symmetry, then the bilocal operators involving a Lie algebra generator also provide long-range deformations of the model \cite{BFLL2013IntegrableDeformations}. This should therefore also correspond to a twist of the corresponding quantum algebra. In order to write down such a twist, one should first derive the Sutherland-type equations for the Lie algebra generators, i.e. bring the commutator $[\Delta(J)_{12}, \L_{\underline{\vec a}(12)}]$ into the form of a Sutherland equation like in Proposition~\ref{prop:higherorderSutherlandEquations}. Then, the twist for the bilocal operator $[\mathbb{Q}_k | J]$ is simply the same as \eqref{eq:twistbilocal}, but where the charges $\mathfrak{q}^{(\ell)}$ are replaced by the Lie algebra generator $J$, and the reduced algebraic charge densities ${}^r\Q^{(\ell)}$ are replaced by the terms in the Sutherland equation for $J$. One can then follow the proof for Proposition~\ref{prop:vanishingdrinfeldassociatorbilocal} to find that there is also a subalgebra of the deformed quantum algebra for which the corresponding Drinfeld associator vanishes. Since such a Lie algebra symmetry, and therefore the corresponding long-range deformations, are model-dependent, we will not discuss them in more detail.

\subsection{Some examples}\label{sec:someexamples}
Up until now, we have described the long-range deformation of the corresponding quantum groups in full generality for any Yang-Baxter integrable spin chain. Let us now give some explicit examples of the Lax operators and $R$-matrices for the $\B[\mathbb{Q}_3]$ and $[\mathbb{Q}_2|\mathbb{Q}_3]$ long-range deformations of the XXX spin chain. In particular, in this section, we will solely focus on the FRT-bialgebra associated to the $R$-matrix for the Yangian $\mathcal{Y}(\mathfrak{gl}_N)$, which is given by \cite{MolevYangiansClassicalLie}
\begin{equation}\label{eq:RmatrixYangian}
    \R_{12}(u,v) = \frac{u-v-\P_{12}}{u-v-1} \in \mathrm{End}(\C^N\otimes \C^N)[\![u,v]\!],
\end{equation}
where $\P_{12} \colon e_i\otimes e_j \mapsto e_j\otimes e_i$ is the permutation operator. Since this $R$-matrix is given in terms of this permutation operator, it will follow that all the algebraic charge densities, and therefore the long-range deformations of the Lax operators and $R$-matrices, are also given in terms of the permutation operator. In particular, using the fact that $\R_{12}^{-1}(u,v) = \R_{21}(v,u)$, it can be readily calculated that the second algebraic charge density acting on two sites is given by
\begin{equation}
    \Q^{(2)}_{12}(u,v) = \R^{-1}_{12}(u,v)\frac{d}{du}\R_{12}(u,v) = \frac{\P_{12}-1}{(u-v)^2 - 1}.
\end{equation}
Note here that $\Q^{(2)}_{12}(0,0) = 1 - \P_{12}$, which is indeed the Hamiltonian density for the XXX spin chain. It now also readily follows that the third algebraic charge density acting on two sites is given by
\begin{equation}
    \Q^{(3)}_{12}(u,v) = \frac{d}{du}\Q^{(2)}_{12}(u,v) = -\frac{2(u-v)(\P_{12}-1)}{\left((u-v)^2-1\right)^2},
\end{equation}
where it can be noted that $\Q^{(3)}_{12}(0,0) = 0$, which is a direct consequence of the $R$-matrix being of difference form. Any higher algebraic charge density acting on two sites can be calculated by simply taking additional derivatives, i.e. $\Q^{(k)}_{12}(u,v) = \frac{d^{k-2}}{du^{k-2}}\Q^{(2)}_{12}(u,v)$.\\

Next, let us compute the higher coproducts of the algebraic charge densities. Recall here that the coproducts for the $R$-matrix on three sites are given by
\begin{equation}
    \R_{1(23)}(u;v,w) = \R_{13}(u,w)\R_{12}(u,v)\quad\text{and}\quad \R_{(12)3}(u,v;w) = \R_{13}(u,w)\R_{23}(v,w).
\end{equation}
Using \eqref{eq:coproductCharge2}, we now find that the coproduct for the second algebraic charge density acting on three sites is given by
\begin{equation}
    \begin{split}
        \Q^{(2)}_{1(23)}(u;v,w) &= \R^{-1}_{1(23)}(u;v,w)\frac{d}{du}\R_{1(23)}(u;v,w)\\
        &= \R^{-1}_{12}(u,v)\Q^{(2)}_{13}(u,w)\R_{12}(u,v) + \Q^{(2)}_{12}(u,v)\\
        &= \frac{\P_{12}-1}{(u-v)^2 - 1} + \frac{\P_{13}-1}{(u-w)^2 -1} + \frac{(u-v)[\P_{12}, \P_{13}] + \P_{13} - \P_{23}}{((u-v)^2 -1)((u-w)^2 - 1)}.
    \end{split}
\end{equation}
Here, it can be noted that $\Q^{(2)}_{1(23)}(0;0,0) = 1-\P_{12} + 1 - \P_{23}$, which is the sum over the two Hamiltonian densities of the XXX spin chain. The coproduct for the third algebraic charge density acting on three sites can now be directly calculated to be given by
\begin{equation}
    \begin{split}
        \Q^{(3)}_{1(23)}(u;v,w) &= \frac{d}{du}\Q^{(2)}_{1(23)}(u;v,w)\\
        &= -\frac{2(u-v)(\P_{12}-1)}{((u-v)^2-1)^2} - \frac{2(u-w)(\P_{13}-1)}{((u-w)^2-1)^2} + \frac{[\P_{12}, \P_{13}]}{((u-v)^2-1)((u-w)^2-1)}\\
        &\quad - 2\left(\frac{u-v}{(u-v)^2-1} + \frac{u-w}{(u-w)^2-1}\right)\frac{(u-v)[\P_{12},\P_{13}] + \P_{13} - \P_{23}}{((u-v)^2-1)((u-w)^2-1)}.
    \end{split}
\end{equation}
Remark that $\Q^{(3)}_{1(23)}(0;0,0) - \Q^{(3)}_{23}(0,0) = [\P_{12}, \P_{13}]$, which is indeed equal to the third charge of the XXX spin chain. Any higher algebraic charge densities acting on three sites can now be calculated by taking higher derivatives, i.e. $\Q^{(k)}_{1(23)}(u;v,w) = \frac{d^{k-2}}{du^{k-2}}\Q^{(2)}_{1(23)}(u;v,w)$. Similarly, higher coproducts of the algebraic charge densities can be calculated by taking derivatives of higher coproducts of the $R$-matrix. Also, all of these calculations can be easily performed using \texttt{Mathematica} for any arbitrary $R$-matrix.\\

We can now compute the Lax operator and the $R$-matrix for the $\B[\mathbb{Q}_3]$ generated long-range deformation of the XXX spin chain up to first order in the deformation parameter. First, let us consider the reduced algebraic charge density with a non-trivial spectral parameter in the second coordinate, which is given by
\begin{equation}\label{eq:reducedalgebraicchargedensityYangian}
    \begin{split}
        {}^r\Q^{(3)}_{1,\underline{2},3}(0,u,0) &= \Q^{(3)}_{1(\underline{2}3)}(0;u,0) - \R^{-1}_{1\underline{2}}(0,u)\Q^{(3)}_{13}(0,0)\R_{1\underline{2}}(0,u)\\
        &= \frac{2u}{(-1+u^2)^2}\left(\P_{12} +\P_{23} - \P_{13} - 1\right) + \frac{1+u^2}{(-1+u^2)^2}[\P_{12}, \P_{13}],
    \end{split}
\end{equation}
where we used the notation that an underline signifies a non-trivial insertion of a spectral parameter. Note here that $1,\P_{12}, \P_{23},\P_{13}$ and $[\P_{12}, \P_{13}]$ form a basis for the $\mathfrak{sl}_N$ invariant matrices in $\mathrm{End}(\C^N\otimes \C^N\otimes \C^N)$,\footnote{Recall that the Schur-Weyl dual of $\mathrm{SL}_N$ acting on $(\C^{N})^{\otimes k}$ is the symmetric group $S_k$, such that indeed all the $\mathfrak{sl}_N$-invariant matrices are given in terms of permutation operators.} such that the above form is the unique expression in terms of these basis elements. From \eqref{eq:boostLaxoperatorfromdeformation}, it then follows that the Lax operator for the $\B[\mathbb{Q}_3]$ deformation is given by
\begin{equation}\label{eq:BQ3LaxExample}
    \L^\lambda_{(\underline{1}2)3}(u) \coloneq \left(1 + \lambda\cdot {}^r\Q^{(3)}_{3,\underline{1},2}(0,u,0)\right)\R_{\underline{1}3}(u,0)\P_{23} + \O(\lambda^2),
\end{equation}
which is just a product of the matrices that we gave in \eqref{eq:RmatrixYangian} and \eqref{eq:reducedalgebraicchargedensityYangian}.\footnote{We remark that this Lax operator differs from the one found in e.g. \cite{deLeeuwRetore2023}. One has, in principle, some freedom in defining the Lax operator, as one can, for example, perform a similarity transformation on the auxiliary space and still obtain the same charges. Moreover, our Lax operator is not regular, as it also describes the deformation of the first charge $\mathbb{Q}_1$, which was not included in \cite{deLeeuwRetore2023}. However, one could construct a regular Lax operator by a reparametrisation $u\to u-\frac{\lambda}{u}$, by imposing a non-trivial spectral parameter $\frac{\lambda}{u}$ in the second coordinate of the auxiliary space, and by then removing the singularity in $u$ through a similarity transformation on the auxiliary space. That is, it can be verified that the Lax operator
\begin{equation}
   \L^{\text{reg}}_{(\underline{1}2)3}(u) \coloneq \L_{(\underline{1}2)3}^{\lambda}\left(u - \frac{\lambda}{u}\right) + \frac{\lambda}{u}\frac{d}{dv}\bigg|_{v=0}\R_{(\underline{1}2)3}(u,v;0) - \frac{\lambda}{u}[\Q^{(2)}_{\underline{1}2}(u,0), \R_{(\underline{1}2)3}(u,0;0)] + \O(\lambda^2)
\end{equation}
is regular. Moreover, it can be derived that this Lax operator provides the proper range-three Hamiltonian correction (i.e. without the $\mathbb{Q}_4$ term) and also has an associated $R$-matrix such that the RLL relations are satisfied. We remark that this regularised Lax operator is less natural to consider. Moreover, it is unknown whether such a regularisation can be done for arbitrary boost and bilocal deformations.} Similarly, we can use \eqref{eq:twistedRmatrixboost} to obtain the expression for the $R$-matrix. In particular, note here that the double coproduct for the undeformed $R$-matrix is given by
\begin{equation}
    \R_{(\underline{1}2)(\underline{3}4)}(u,v) = \R_{\underline{1}4}(u,0)\R_{\underline{13}}(u,v)\P_{24}\R_{2\underline{3}}(0,v).
\end{equation}
From \eqref{eq:twistedRmatrixboost}, it then follows that the deformed $R$-matrix for the $\B[\mathbb{Q}_3]$ deformation is given by
\begin{equation}
    \begin{split}
        \R^\lambda_{(\underline{1}2)(\underline{3}4)}(u,v) &\coloneq \R_{(\underline{1}2)(\underline{3}4)}(u,v)  + \O(\lambda^2)\\
        &\quad + \lambda\left({}^r\Q^{(3)}_{4,\underline{1},2}(0,u,0) \R_{(\underline{1}2)(\underline{3}4)}(u,v) - \R_{(\underline{1}2)(\underline{3}4)}(u,v){}^r\Q^{(3)}_{2,\underline{3},4}(0,v,0)\right),
    \end{split}
\end{equation}
where we can recognise the usual formula for a twisted $R$-matrix. It now follows from Proposition~\ref{prop:vanishingDrinfeldassociatorBQk} that the RLL-relation and the Yang-Baxter equation for this Lax operator and $R$-matrix are satisfied. In particular, we have
\begin{equation}\label{eq:RLLBq3deformationYangian}
    \R^\lambda_{(\underline{1}2)(\underline{3}4)}(u,v)\L^\lambda_{(\underline{1}2)5}(u)\L^\lambda_{(\underline{3}4)5}(v) = \L^\lambda_{(\underline{3}4)5}(v)\L^\lambda_{(\underline{1}2)5}(u)\R^\lambda_{(\underline{1}2)(\underline{3}4)}(u,v) + \O(\lambda^2),
\end{equation}
which can also be explicitly checked in, for example, \texttt{Mathematica}. This shows that we get an integrable model up to first order in $\lambda$. That is, if one considers the monodromy matrix $\T^\lambda_{(\underline{a}b)}(u) \coloneq \L^{\lambda}_{(\underline{a}b)L}(u) \cdots \L^{\lambda}_{(\underline{a}b)1}(u)$
and the corresponding transfer matrix $t^\lambda(u) \coloneq \mathrm{tr}_{\underline{a}b}\T^\lambda_{\underline{a}b}(u)$, then it follows from \eqref{eq:RLLBq3deformationYangian} that $[t^\lambda(u), t^\lambda(v)] = \O(\lambda^2)$, such that the transfer matrix generates a tower of commuting charges up to first order in the deformation parameter.  It can be directly calculated (using, for example, \texttt{Mathematica}) that the second charge for the $\B[\mathbb{Q}_3]$ deformation is given by
\begin{equation}\label{eq:examplecharge2}
    \begin{split}
        \mathbb{Q}_2(\lambda) = t^{\lambda}(0)^{-1}\frac{d}{du}t^\lambda(u)\Big|_{u=0} &= \sum_{n=1}^{L}(1-\P_{n,n+1}) + \lambda\Big[(1-\P_{n,n+2}) - 4(1-\P_{n,n+1})\Big] \\
        &\quad + \frac{\lambda }{2}\mathbb{Q}_4(0) + \O(\lambda^2).
    \end{split}
\end{equation}
On the right, one can recognise the two-loop correction to the dilatation operator in the $\mathfrak{su}(2)$ sector of planar $\mathcal{N}=4$ SYM \cite{Beisert2003Thedilatation, Beisert2004NovelLongRange}.\footnote{Remark that the explicit form in terms of permutation operators is different from \cite{Beisert2003Thedilatation, Beisert2004NovelLongRange}, as the $\mathfrak{su}(2)$ invariant expression in terms of products of permutation operators is not unique. However, one can explicitly calculate that our correction to the Hamiltonian is indeed equal to the two-loop correction to the dilatation operator as given in  \cite{Beisert2003Thedilatation, Beisert2004NovelLongRange}.} The additional $\mathbb{Q}_4$ charge can be simply removed by a change of basis of the charges, or via the reparametrisation $u\to u - \frac{\lambda}{u}$ of the spectral parameter.\footnote{Remark that this reparametrisation introduces a pole at $u=0$. In particular, the charges cannot be defined by evaluation at $u=0$ of the derivatives in \eqref{eq:definitioncharges}. Rather, the charges should be (equivalently) defined via contour integrals around $u=0$.} We have therefore explicitly constructed the Lax operator and the $R$-matrix for  the $\mathfrak{su}(2)$-sector of planar $\mathcal{N}=4$ up to two-loop. Remark here that the Lax operator and $R$-matrix are not equal to each other. Moreover, although the Lax operator is non-regular, the $R$-matrix is regular, i.e. reduces to the permutation operator at $u,v = 0$.\\

Next, let us consider the long-range deformation corresponding to the bilocal $[\mathbb{Q}_2 | \mathbb{Q}_3]$ operator. We will denote the charge densities by $\tilde{\mathfrak{q}}^{(2)}_{12} = 1-\P_{12}$ and $\mathfrak{q}^{(3)}_{123} = [\P_{12}, \P_{13}]$. Moreover, as remarked in Section~\ref{sec:algebraicchargedensities}, since the $R$-matrix is of difference form, we have $\tilde \Q^{(2)}_{12}(u,v) = \Q^{(2)}_{12}(u,v)$ and ${}^r\tilde{Q}^{(2)}_{1,\underline{2},3}(0,u,0) = \tilde{\Q}^{(2)}_{\underline{2}3}(u,0)$ by \eqref{eq:reducedAlgebraicchargedensityproperty}. Recall that the undeformed Lax operator on a tripled auxiliary space is given by $\L_{(\underline{1}23)4}(u) \coloneq \R_{\underline{1}4}(u,0)\P_{24}\P_{34}$. From \eqref{eq:bilocalLaxoperatorfromdeformation}, and also using the property for the reduced algebraic charge density as given in \eqref{eq:reducedAlgebraicchargedensityproperty}, we find that the Lax operator for the $[\mathbb{Q}_2|\mathbb{Q}_3]$ long-range deformation is given by
\begin{equation}
    \begin{split}
        \L^{\lambda}_{(\underline{1}23)4}(u) \coloneq \L_{(\underline{1}23)4}(u) +&\lambda\Big(\R^{-1}_{\underline{1}(23)}(u;0,0)\tilde{\mathfrak{q}}^{(2)}_{23}\R_{\underline{1}(23)}(u;0,0){}^r\Q^{(3)}_{4,\underline{1},2}(0,u,0) \L_{(\underline{1}23)4}(u)\\
        & - \L_{(\underline{1}23)4}(u) \R^{-1}_{\underline{1}(23)}(u;0,0)\tilde{\Q}^{(2)}_{\underline{1}4}(u,0)\R_{\underline{1}(23)}(u;0,0)\mathfrak{q}^{(3)}_{234}\Big) + \O(\lambda^2).
    \end{split}
\end{equation}
Similarly, we can use \eqref{eq:twistedRmatrixbilocal} to get the $R$-matrix for the bilocal deformation. Recall here that the undeformed $R$-matrix on a tripled auxiliary space is given by
\begin{equation}
    \R_{(\underline{1}23)(\underline{4}56)}(u,v) = \R_{\underline{1}6}(u,0)\R_{\underline{1}5}(u,0)\R_{\underline{14}}(u,v)\P_{26}\P_{25}\R_{2\underline{4}}(0,v)\P_{36}\P_{35}\R_{3\underline{4}}(0,v).
\end{equation}
Before writing down the deformation of this $R$-matrix, let us first write down the explicit twist element for the deformation. From \eqref{eq:twistbilocal}, it follows that the twist for the $[\mathbb{Q}_2|\mathbb{Q}_3]$ bilocal deformation is given by
\begin{equation}
    \begin{split}
        \F^{(1)}_{(\underline{1}23)(\underline{4}56)}(u,v) &\coloneq \R^{-1}_{\underline{4}(56)}\R^{-1}_{\underline{1}(235)}\tilde{\Q}^{(2)}_{\underline{1}6}\R_{\underline{1}(235)}\R_{\underline{4}(356)}{}^r\Q^{(3)}_{2,\underline{4},3}\R^{-1}_{\underline{4}3}\\
        &+\R^{-1}_{\underline{4}(56)}\R^{-1}_{\underline{1}(2356)}[\tilde{\mathfrak{q}}^{(2)}_{56}, \R_{\underline{1}(2356)}][\mathfrak{q}^{(3)}_{235},\R_{\underline{4}(2356)}]\R^{-1}_{\underline{4}(23)}\\
        &+ \R^{-1}_{\underline{4}(56)}\R^{-1}_{\underline{1}(2356)}\tilde{\mathfrak{q}}^{(2)}_{35}\R_{\underline{1}(2356)}\R_{\underline{4}(356)}{}^r\Q^{(3)}_{2,\underline{4}, 3}\R^{-1}_{\underline{4}3}\\
        &+  \R^{-1}_{\underline{4}(56)}\R^{-1}_{\underline{1}(2356)}\tilde{\mathfrak{q}}^{(2)}_{56}\R_{\underline{1}(2356)}\R_{\underline{4}(56)}{}^r\Q^{(3)}_{3,\underline{4}, 5}\\
        &+ \R^{-1}_{\underline{4}(56)}\R^{-1}_{\underline{1}(23)}\tilde{\Q}^{(2)}_{\underline{1}5}\R_{\underline{1}(23)}\R_{\underline{4}(2356)}\mathfrak{q}^{(3)}_{235}\R^{-1}_{\underline{4}(23)}\\
        &+\R^{-1}_{\underline{4}(56)}\R^{-1}_{\underline{1}(235)}\tilde{\Q}^{(2)}_{\underline{1}6}\R_{\underline{1}(235)}\R_{\underline{4}(2356)}\mathfrak{q}^{(3)}_{356}\R^{-1}_{\underline{4}(23)},
    \end{split}
\end{equation}
where we left out the explicit dependence on the spectral parameters on the right side. However, one should note that the coordinate $\underline{1}$ corresponds to an insertion of the spectral parameter $u$, while $\underline{4}$ corresponds to the spectral parameter $v$. In particular, we used the notation $\R_{\underline{1}(i_1,...,i_n)} \coloneq \R_{\underline{1},i_n}(u,0)\cdots \R_{\underline{1},i_1}(u,0)$, $\R_{\underline{4}(i_1,...,i_n)} \coloneq \R_{\underline{4},i_n}(v,0)\cdots \R_{\underline{4},i_1}(v,0)$ and $\tilde\Q^{(2)}_{\underline{1}, i} \coloneq \Q^{(2)}_{\underline{1}, i}(u,0)$ and ${}^r\Q^{(3)}_{i,\underline{4}, j} \coloneq {}^r\Q^{(3)}_{i,\underline{4}, j}(0,v,0)$. From \eqref{eq:twistedRmatrixbilocal}, it then follows that the $R$-matrix for the $[\mathbb{Q}_2|\mathbb{Q}_3]$ deformation is given by
\begin{equation}
    \begin{split}
        \R^\lambda_{(\underline{1}23)(\underline{4}56)}(u,v) &\coloneq \R_{(\underline{1}23)(\underline{4}56)}(u,v)  + \O(\lambda^2)\\
        &\quad + \lambda\left(\F^{(1)}_{(\underline{4}56)(\underline{1}23)}(v,u) \R_{(\underline{1}23)(\underline{4}56)}(u,v) - \R_{(\underline{1}23)(\underline{4}56)}(u,v) \F^{(1)}_{(\underline{1}23)(\underline{4}56)}(u,v)\right).
    \end{split}
\end{equation}
Similar to before, it follows from Proposition~\ref{prop:vanishingdrinfeldassociatorbilocal} that the RLL-relation and the Yang-Baxter equation for the above Lax operator and $R$-matrix for the $[\mathbb{Q}_2|\mathbb{Q}_3]$ deformation are satisfied, which can also be explicitly verified using, for example, \texttt{Mathematica}. In particular, this implies that the transfer matrix $t^\lambda (u) \coloneq \mathrm{tr}_{\underline{a}b_1b_2}\big(\L^\lambda_{(\underline{a}b_1b_2)L}\cdots \L^\lambda_{(\underline{a}b_1b_2)1}\big)$ generates a tower of commuting charges up to first order in $\lambda$, since $[t^\lambda(u), t^\lambda(v)] = \O(\lambda^2)$. For example, it can be calculated in \texttt{Mathematica} that the correction to the second charge for the bilocal $[\mathbb{Q}_2|\mathbb{Q}_3]$ long-range deformation of the $\mathrm{SU}(2)$ XXX Heisenberg spin chain is given by
\begin{equation}
    \begin{split}
        \mathbb{Q}_2(\lambda) = \sum_{n=1}^{L}&(1-\P_{n,n+1}) + \lambda\Big[4(\P_{n,n+1}-1) + (\P_{n,n+3}-1) - (\P_{n,n+2}-1)\\
        &+ (\P_{n,n+1,n+3} - \P_{n,n+3,n+1})-2(\P_{n,n+1}\P_{n+2,n+3}-1)  \\
        & + (\P_{n,n+2,n+3,n+1} - \P_{n,n+2,n+1,n+3})+ (\P_{n,n+1,n+3,n+2} - \P_{n,n+1,n+2,n+3})\Big]\\
        &+\O(\lambda^2),
    \end{split}
\end{equation}
where we used the notation $\P_{n_1,n_2,...,n_k} \coloneq \P_{n_1,n_k}\P_{n_1,n_{k-1}}\cdots \P_{n_1,n_2}$, which simply corresponds to the shift operator. It can be seen that the second charge gets a range-four correction at first order in $\lambda$ under the $[\mathbb{Q}_2|\mathbb{Q}_3]$ deformation. This term is also expected to give a contribution to the dilatation operator in the four-loop correction of the $\mathfrak{su}(2)$-sector of planar $\mathcal{N}=4$ SYM \cite{Beisert2003Thedilatation, Beisert2004NovelLongRange}. Of course, one can use the results of our paper to also compute the Lax operators, $R$-matrices and charges for general boost and bilocal deformations and for general FRT-bialgebras up to first order in $\lambda$.

\subsection{Long-range deformed Yangian}\label{sec:longrangedeformedYangian}
Up until now, we have described the long-range deformations in terms of the FRT-bialgebras. As described in Section~\ref{sec:dualdescription}, these FRT-bialgebras are dual to the quantised universal enveloping algebra constructions. In particular, we can now also describe the long-range deformations of these quantised universal enveloping algebras. For the purpose of this section, we will focus on the Yangian algebras, as it is the easiest example to consider. We will mostly follow the same steps as in Section~\ref{sec:deformationFRTalgebra}, but now applied to the bialgebra dual. We remark that the following discussion is expected to generalise naturally to more general quantum affine algebras.\\

In order to describe the long-range deformations of the Yangian $\mathcal{Y}_\hbar(\mathfrak{g})$ (see Section~\ref{sec:dualdescription}), we first need to consider a similar ``doubling'' of the quantum group, as was needed for the FRT-bialgebra. Let us first discuss the dual to the algebra $\A_0$. Let $(\rho_0,V)$ be a fundamental representation of $\mathcal{Y}_\hbar(\mathfrak{g})$, such that $\rho_0(x) = \tilde{x} \in \mathrm{End}(V)$ and $\rho_0(\J(x)) =0$ for $x\in \mathfrak{g}$, where $\tilde x$ denotes the fundamental representation of $x$. Next, we consider the linear map $\Phi_0 \colon \mathcal{Y}_{\hbar}(\mathfrak{g}) \to \A_0^*$ given by $\Phi_0(x)(a) \coloneq \langle a,x\rangle$ for $x\in \Y_\hbar(\mathfrak{g})$ and $a\in\A_0$, and where $\langle\cdot, \cdot\rangle$ is the bilinear map \eqref{eq:dualpairingYangian} for the representation $\rho_0$ and restricted to $\A_0 \subseteq \A(\R_{\rho_0})$. Remark that, by the properties of the dual pairing \eqref{eq:dualpairingprops}, $\Phi_0$ is automatically a bialgebra homomorphism.\footnote{One should remember here that $\A_0^*$ is not actually a bialgebra, since the image of the ``coproduct'' of $\A^*_0$ does not necessarily lie in $\A_0^*\otimes \A_0^*$ (see footnote~\ref{footnote:coproductdualalgebra}). However, since $\mathcal{Y}_\hbar(\mathfrak{g})$ is a bialgebra, this also ensures that $\mathrm{Im}(\Phi_0)\subseteq \A_0^*$ is a bialgebra. Therefore, $\Phi_0$, when restricting the codomain to its image, is a bialgebra homomorphism.} We then define the bialgebra\footnote{\label{footnote:antipodeforYinfty}If $S(\ker(\Phi_0)) \subseteq \ker(\Phi_0)$, then $\Y_{\infty} \cong \Y_{\hbar}(\mathfrak{g}) / \ker(\Phi_0)$ is a Hopf algebra, with the antipode for $\mathcal{Y}_\infty$ induced from $\Y_{\hbar}(\mathfrak{g})$. Although this is believed to be true, we do not provide a proof for this.}
\begin{equation}\label{eq:Yinftydef}
    \Y_{\infty} \coloneq \mathrm{Im}(\Phi_0)\subseteq \A_0^*,
\end{equation}
which one can think of as being the algebra of Yangian operators acting on the ``evaluation zero spin chains''. Here, we should explain our notation for $\Y_\infty$, which refers to $\hbar = \infty$. Recall that the algebra $\A_0$ corresponds to all matrix elements with spectral parameter $u = 0$. For the Yangian algebra, the spectral parameters have ``units of $\hbar$''. For example, one can reintroduce $\hbar$ in the $R$-matrix \eqref{eq:RmatrixYangian} for the Yangian by replacing $u\to u/\hbar$. In this sense, the $u\to 0$ limit corresponds to the $\hbar\to \infty$ limit. Similarly, since $\rho_0(\J(x)) = 0$, it can be noted from \eqref{eq:coproductYangian} that $(\rho_0\otimes \rho_0)(\Delta(\J(x))) = \frac{\hbar}{2}(\rho_0\otimes \rho_0)([x\otimes 1, t])$, i.e. only the term proportional to $\hbar$ survives for the coproduct. This coproduct structure can be emulated in \eqref{eq:coproductYangian} by ``taking $\hbar$ to be large''. In practice, this $\hbar \to \infty$ interpretation corresponds to the fact that neither $\A_0$ nor $\mathcal{Y}_\infty$ has an appropriate classical limit. Indeed, $\A_0$ is the free algebra generated by $t_{ij}(0)$, with an $R$-matrix given by the permutation operator $\P$. There is no limit in which $\P$ can be approximated around the identity, in which the elements $t_{ij}(0)$ become commutative. In that sense, one can also view $\Y_{\infty}$ as the ``pure quantum'' limit of the Yangian algebra, which cannot be approximated around a classical universal enveloping algebra. We emphasise here that $\Y_{\infty}$ is defined through \eqref{eq:Yinftydef}, and is not obtained by literally taking the $\hbar\to \infty$ limit of the Yangian algebra $\Y_{\hbar}(\mathfrak{g})$.\\

We can now consider the doubling of the Yangian algebra with this $\Y_\infty$, which is the dual description to Definition~\ref{def:doublecrossedproduct}. In particular, we would like to construct a double-crossed \textit{co}product algebra between $\Y_{\hbar}(\mathfrak{g})$ and $\mathcal{Y}_\infty$, which is obtained by twisting the coproduct structure of the tensor product bialgebra  $\Y_{\hbar}(\mathfrak{g})\otimes \Y_\infty$ by the universal $R$-matrix. Here, however, there are some difficulties, as the universal $R$-matrix for the Yangian is only defined in a topological completion. Therefore, the coproduct for the double-crossed coproduct algebra can also only be defined in such a completion. For the Yangian $\Y_{\hbar}(\mathfrak{g})$, we define the \textit{$u$-deformed coproduct}
\begin{equation}\label{eq:udeformedcoproductYangian}
    \Delta_{u,v}(x) \coloneq (\B_u\otimes \B_v)(\Delta(x)) \in \mathcal{Y}_\hbar(\mathfrak{g})(\!(u)\!)\otimes \Y_{\hbar}(\mathfrak{g})(\!(v)\!)
\end{equation}
for $x\in \Y_{\hbar}(\mathfrak{g})$, which satisfies $\R(u-v)\Delta_{u,v}(x) = \Delta_{v,u}^{\mathrm{cop}}(x)\R(u-v)$, and reduces to the ordinary coproduct in the limit $u,v\to 0$. See e.g. \cite{Hernandez2003QuantumAffinizations, Hernandez2007DrinfeldCoproduct} for more details on such a $u$-deformed coproduct structure. Next, we will introduce a similar notation as before for the FRT-bialgebras, to distinguish between elements in $\Y_{\hbar}(\mathfrak{g})$ and $\Y_\infty$. In particular, we define
\begin{equation}
    \begin{split}
        \R_{\underline{12}}(u-v) &\coloneq  \R(u-v)\in \Y_{\hbar}(\mathfrak{g})(\!(u)\!)\otimes \Y_{\hbar}(\mathfrak{g})(\!(v)\!),\\
        \R_{\underline{1}2}(u) &\coloneq (\id\otimes \Phi_0)(\R(u)) \in \Y_\hbar(\mathfrak{g})(\!(u)\!)\otimes \Y_\infty,\\
        \R_{1\underline{2}}(-v) &\coloneq (\Phi_0\otimes \id)(\R(-v)) \in \Y_\infty \otimes \Y_{\hbar}(\mathfrak{g})(\!(v)\!),\\
        \R_{12}(0) &\coloneq (\Phi_0\otimes \Phi_0)(\R(0)) \in \Y_\infty\otimes \Y_\infty,
    \end{split}
\end{equation}
where it should be remarked that, since $\Y_\infty$ corresponds to evaluation zero representations, the element $(\Phi_0\otimes \Phi_0)(\R(0))$ is well-defined in $\Y_\infty\otimes \Y_\infty$, and is in fact also the $R$-matrix for $\Y_\infty$. An underline in the coordinate signifies an element in $\Y_{\hbar}(\mathfrak{g})$, while no underline corresponds to an element in $\Y_\infty$.\\

We can now define the double-crossed coproduct bialgebra for the topological completion:
\begin{definition}
    The \textit{$u$-deformed double-crossed coproduct bialgebra} $\Y_\hbar(\mathfrak{g})\hat{\bowtie}_{\R} \Y_\infty$ is the algebra $\Y_\hbar(\mathfrak{g})\otimes \Y_\infty$, with counit induced from the tensor product bialgebra, and with a $u$-deformed coproduct given by
    \begin{equation}
        \Delta_{u,v}(x\otimes y) \coloneq \sum \R_{2\underline{3}}^{-1}(-v)\left[\B_u(x_{(1)})\otimes y_{(1)}\otimes\B_v(x_{(2)})\otimes y_{(2)}\right]\R_{2\underline{3}}(-v),
    \end{equation}
    for $x\in \Y_{\hbar}(\mathfrak{g})$ and $y\in \Y_{\infty}$.
\end{definition}
\noindent In this definition, one can recognise the dual construction of \eqref{eq:bicrossedproduct}. That is, there is a dual pairing
\begin{equation}
    \langle \cdot, \cdot\rangle \colon (\A(\R_\rho)\bowtie_{\mathbf{r}} \A_0)\times \left(\Y_\hbar(\mathfrak{g})\hat{\bowtie}_{\R} \Y_\infty\right) \to \C,
\end{equation}
given by $\langle a\otimes b, x\otimes y\rangle = \langle a,x\rangle y(b)$ for $a\otimes b\in \A(\R_\rho)\bowtie_{\mathbf{r}} \A_0$ and $x\otimes y \in \Y_\hbar(\mathfrak{g})\hat{\bowtie}_{\R} \Y_\infty$, and $\langle a, x\rangle$ is the dual pairing \eqref{eq:dualpairingYangian} between $\A(\R_\rho)$ and $\Y_{\hbar}(\mathfrak{g})$. For the FRT-bialgebra, we saw in \eqref{eq:bicrossedproduct} that the double-crossed product construction changes the product structure. Correspondingly, for the Yangian, it's the coproduct structure that gets altered. We emphasise, however, that $\Y_\hbar(\mathfrak{g})\hat{\bowtie}_{\R} \Y_\infty$ is not a bialgebra in the strict sense, as the coproduct is only well-defined in the appropriate topological completion.\footnote{It might be true that the $u,v\to 0$ limit of the coproduct $\Delta_{u,v}$ for the double-crossed coproduct bialgebra is actually well-defined, just as is the case for the coproduct \eqref{eq:udeformedcoproductYangian} of the (single) Yangian $\Y_\hbar(\mathfrak{g})$, such that $\Y_\hbar(\mathfrak{g})\hat{\bowtie}_{\R} \Y_\infty$ becomes a proper bialgebra. However, it is currently unknown whether this is indeed the case.} The double-crossed coproduct algebra $\Y_\hbar(\mathfrak{g})\hat{\bowtie}_{\R} \Y_\infty$ is pseudo-quasitriangular, with the universal $R$-matrix given by
\begin{equation}
    \R_{(\underline{1}2)(\underline{3}4)}(u,v) \coloneq \R^{-1}_{4\underline{1}}(-u)\R_{\underline{13}}(u-v)\R_{24}(0)\R_{2\underline{3}}(-v),
\end{equation}
such that $\R_{(\underline{1}2)(\underline{3}4)}(u,v)\Delta_{u,v}(x\otimes y) = \Delta_{v,u}^{\text{cop}}(x\otimes y)\R_{(\underline{1}2)(\underline{3}4)}(u,v)$ for $x\otimes y \in \Y_\hbar(\mathfrak{g})\hat{\bowtie}_{\R} \Y_\infty$. Remark that this universal $R$-matrix is dual to the $r$-form as given in \eqref{eq:rformdoubled}.\\

Next, we can consider twists of this double-crossed coproduct bialgebra. Here, we should now twist the $u$-deformed coproduct. In particular, we consider an invertible element\footnote{Remark that this twisting element is fully written in terms of algebraic charge densities, which correspond to derivations of the universal $R$-matrix and are well-defined algebra elements of the Yangian as discussed in Section~\ref{sec:algebraicchargedensities}.}
\begin{equation}
    \F_{(\underline{1}2)(\underline{3}4)}(u,v) \in \big[(\Y_\hbar(\mathfrak{g})\hat{\bowtie}_{\R} \Y_\infty)(\!(u)\!)\otimes  (\Y_\hbar(\mathfrak{g})\hat{\bowtie}_{\R} \Y_\infty)(\!(v)\!)\big][\![\lambda]\!],
\end{equation}
which is the twisting element dual to the twisting function for the boost or bilocal deformation as given in \eqref{eq:BQktwistelement} and \eqref{eq:twistbilocal}, such that 
\begin{equation}
    \lim_{u,v\to 0}\langle (a\otimes b)\otimes (c\otimes d),  \F_{(\underline{1}2)(\underline{3}4)}(u,v)\rangle = \gamma(a\otimes b, c\otimes d),
\end{equation}
for $a\otimes b, c\otimes d \in \A(\R_\rho)\bowtie_{\mathbf{r}} \A_0$, where it should be noted that the $u,v\to 0$ limit is well-defined under the dual pairing, as also the $R$-matrix in a given representation is well-defined at zero spectral parameter. We then have the following definition:
\begin{definition}
    The \textit{cotwisted $u$-deformed quasi-bialgebra} $(\Y_\hbar(\mathfrak{g})\hat{\bowtie}_{\R} \Y_\infty)(\F)$ is the algebra \linebreak $\Y_\hbar(\mathfrak{g})\hat{\bowtie}_{\R} \Y_\infty$, with a counit induced from $\Y_\hbar(\mathfrak{g})\hat{\bowtie}_{\R} \Y_\infty$, and with a cotwisted $u$-deformed coproduct given by
    \begin{equation}
        \Delta^\F_{u,v}(x\otimes y)\coloneq \F_{(\underline{1}2)(\underline{3}4)}(u,v)\Delta_{u,v}(x\otimes y) \F^{-1}_{(\underline{1}2)(\underline{3}4)}(u,v),
    \end{equation}
    for $x\otimes y \in \Y_\hbar(\mathfrak{g})\hat{\bowtie}_{\R} \Y_\infty$.
\end{definition}
\noindent In general, this cotwisted bialgebra is a quasi-bialgebra, i.e. the coproduct is non-coassociative (in the $u$-deformed sense).\footnote{Also here, it is unknown whether the cotwisted coproduct $\Delta^\F_{u,v}$ has a well-defined $u,v\to 0$ limit, which would make this $u$-deformed quasi-bialgebra into a proper quasi-bialgebra.} However, Propositions~\ref{prop:vanishingDrinfeldassociatorBQk} and \ref{prop:vanishingdrinfeldassociatorbilocal} imply that, after some number of successive applications of the coproduct, the corresponding Drinfeld coassociator becomes trivial up to first order in $\lambda$, such that the cotwisted algebra still has a (perturbatively) coassociative substructure. Moreover, the cotwisted algebra is pseudo-quasitriangular, with the cotwisted $R$-matrix given by
\begin{equation}
    \R^\F_{(\underline{1}2)(\underline{3}4)}(u,v) \coloneq \F_{(\underline{3}4)(\underline{1}2)}(v,u)\R_{(\underline{1}2)(\underline{3}4)}(u,v)\F^{-1}_{(\underline{1}2)(\underline{3}4)}(u,v),
\end{equation}
which is the dual of the twisted $r$-form as given in \eqref{eq:deformedrform}.\\

It is important to note that, for the Yangian (or rather its double-crossed coproduct bialgebra $\Y_\hbar(\mathfrak{g})\hat{\bowtie}_{\R} \Y_\infty$), it is the \textit{co}algebra structure that is deformed upon twisting. In particular, the algebra relations for the Yangian remain untouched.\footnote{This also implies that one cannot simply obtain the long-range deformed Yangian from a \textit{RTT-realisation} with the deformed $R$-matrix. Namely, the RTT relations define the algebra structure on the Yangian, which remains undeformed under the cotwisting. In this sense, the $R$-matrix that defines these algebra relations should also remain undeformed. The RTT relation involving the deformed $R$-matrix, therefore, only applies to the long-range deformed dual Yangian, i.e. the FRT-bialgebra. Of course, the $\ell$-functional $\bm{\ell}(a) \coloneq \mathbf{r}^\gamma(a,\cdot)$ does still define a bialgebra homomorphism between the (co-opposite) deformed FRT-bialgebra and the deformed Yangian algebra, such that in this way one may still obtain Yangian elements from the RTT realisation. } This is also expected for the long-range integrable models, which were shown to still have an underlying Yangian symmetry on the doubly-infinite spin chains \cite{Beisert2008YangianSymmetry, Bargheer2008BoostingNearestNeighbour, Bargheer2009}. The twist of the coproduct structure merely has the consequence that the representation of the Yangian on the spin chain of length $L$ gets deformed. Remark here that, for specific representations, the $u,v\to 0$ limits of the $u$-deformed coproducts will become well-defined, as the $R$-matrix in this limit has a well-defined representation, such that, in these representations, the coproduct can just be considered as an ordinary coproduct for a bialgebra.\\

Remember that the $\ell$-functionals as described in Section~\ref{sec:coquasitriangularityandLfunctionals} define representations for the FRT-bialgebra via the representations of the Yangian, and that these then define the Lax operator. For the long-range deformation, we have a similar story: the long-range deformed Lax operator corresponds to the map
\begin{equation}
    \bm{\ell}^\gamma : (\A(\R)\bowtie_{\mathbf{r}}\A_0)^{\text{cop}}(\gamma) \to \hat{\Y}_{\infty},\quad \bm{\ell}^\gamma(a\otimes b)(t_{ij}(0)) \coloneq \mathbf{r}^\gamma(a\otimes b, 1\otimes t_{ij}(0)),
\end{equation}
where $\hat \Y_\infty$ is the appropriate completion of $\Y_\infty \subseteq \A_0^*$ such that $\bm{\ell}^{\gamma}(a\otimes b)\in \hat{\Y}_\infty$ for \linebreak $a\otimes b\in (\A(\R)\bowtie_{\mathbf{r}}\A_0)^{\text{cop}}(\gamma)$. In particular, the representation for the FRT-bialgebra corresponding to the Lax operator is determined by the representation of $\Y_\infty$, which corresponds to the algebra of Yangian operators acting on the zero evaluation spin chains. In this way, we see that the long-range deformation of the spin chain is explicitly obtained through the deformation of the representation of the underlying Yangian algebra due to the twisting of the coproduct.

\section{Discussion and outlook}
In this paper, we have explicitly constructed the FRT-bialgebras relevant for the long-range deformations of homogeneous Yang-Baxter integrable spin chains up to first order in the deformation parameter $\lambda$. In Section~\ref{sec:algebraicchargedensities}, we introduced new algebra elements, the algebraic charge densities, whose algebraic properties allowed us to give a conjecture for the explicit expression of the charge densities, and to derive the generalised Sutherland equations in Section~\ref{sec:higherorderSutherlandequation}. These Sutherland equations were then used in Section~\ref{sec:LaxforBQ3derivation} to derive the Lax operators for the full family of long-range deformations up to first order in the deformation parameter. From the form of these Lax operators, it was observed that the appropriate FRT-bialgebra corresponding to the long-range deformation is given by the double-crossed product bialgebra $\A(\R)\bowtie_{\mathbf{r}}\A_0$ in order to describe the increased auxiliary space that is needed to describe the wrapping interactions. Moreover, we then gave the explicit twisting elements in Section~\ref{sec:deformationFRTalgebra} corresponding to the local, boost and bilocal long-range deformations. In particular, it was observed here that the twisted FRT-bialgebras were in general non-associative, and that the corresponding non-trivial Drinfeld associator encodes the long-range interaction terms of the Lax operator. The main results of this paper were then given in Propositions~\ref{prop:vanishingDrinfeldassociatorlocaldef}, \ref{prop:vanishingDrinfeldassociatorBQk} and \ref{prop:vanishingdrinfeldassociatorbilocal}, in which it is shown that the Drinfeld associator vanishes for a large subalgebra of the deformed FRT-bialgebra. This implies that the algebra still has a large associative substructure for which the Yang-Baxter equation is satisfied. In particular, this gave rise to explicit expressions for the Lax operators and $R$-matrices for the long-range deformations, which manifestly satisfy the RLL-relation and the Yang-Baxter equation. In Section~\ref{sec:someexamples}, we gave explicit examples of these Lax operators and $R$-matrices for the $\B[\mathbb{Q}_3]$ and $[\mathbb{Q}_2|\mathbb{Q}_3]$ long-range deformations of the XXX Heisenberg spin chain. Lastly, in Section~\ref{sec:longrangedeformedYangian}, we also derived the long-range deformation for the Yangian algebra. Here, we saw that the coproduct for the deformed algebra can only be defined in a topological completion. However, this topologically completed coproduct still possesses a pseudo-triangular universal $R$-matrix that braids the coproduct, and that corresponds to the $R$-matrix for the long-range deformed models under a given representation.\\

\noindent There are several interesting future research directions one can take from here:
\begin{itemize}[leftmargin = *]
    \item \textbf{Higher-order corrections.} We have now only described the quantum group deformations up to first order in the deformation parameter $\lambda$. The next obvious step would be to describe the deformed quantum groups up to all orders in $\lambda$, which are still expected to correspond to twists of the double-crossed product algebra $\A(\R)\bowtie_{\mathbf{r}}\A_0$. Moreover, it is then expected that the long-range nature of the perturbative integrable spin chain is encoded in the perturbatively associative structure of this deformed algebra. In particular, for every $n\geq 1$, it is expected that there is a subalgebra of $(\A(\R)\bowtie_{\mathbf{r}}\A_0)(\gamma)$ such that $\varphi^{(m)}$ vanishes for all $m\leq n$ in that subalgebra. In this way, the algebra would give rise to perturbative Lax operators and $R$-matrices which satisfy the RLL-relations and the Yang-Baxter equations up to order $\lambda^n$. We remark here that, although the Lax operators and $R$-matrices could then only be defined perturbatively, the corresponding quantum group could possibly be defined non-perturbatively, thus giving an exact quantum group structure. The family of long-range deformations defined by the local, boost and bilocal operators would then provide a moduli space of such long-range deformed quantum groups.

    \item \textbf{Classification.} We have now derived the twisted FRT-bialgebras by starting from the deformation equation described in \cite{Bargheer2008BoostingNearestNeighbour, Bargheer2009}. A natural question would be whether this family of long-range deformations does indeed give the full classification of long-range deformations, and therefore also the full set of ``long-range twists'' that one could perform on the FRT-bialgebra. It is currently unknown whether this family describes all the long-range deformations or whether more twists exist that give rise to a perturbatively associative substructure of the quantum algebra.

    \item \textbf{Algebraic Bethe Ansatz.} One of the motivations for understanding the quantum group structure of the long-range deformations is that it is expected to provide insights into the long-range Algebraic Bethe Ansatz. Indeed, our construction now allows us to, at the very least, identify the correct $N$-magnon creation operators. In particular, it can be remarked that, for a given twist $\gamma$, we can define a Lax operator with non-trivial multiple spectral parameters in the auxiliary space. In particular, for $\vec{a} = (a_1,..., a_N)$, we can define
    \begin{equation}\label{eq:laxmultiplespecralparams}
        \L^{\gamma}_{(\underline{\vec a}\vec b)1}(u_{\vec a}) \coloneq \L^{\gamma}_{(\underline{a}_1,...,\underline{a}_N,\vec b)1}(u_1,...,u_N) = \mathbf{r}^{\gamma}(\T_{\underline{a}_1}(u_1)\cdots \T_{\underline{a}_N}(u_N)\otimes \T_{b_1}\cdots \T_{b_{k-2}}, 1\otimes \T_1).
    \end{equation}
    In particular, when for example restricting to $6$-vertex models, such a Lax operator provides a representation for elements of the form $B(u_1)\cdots B(u_N)\otimes t_{i_1i_1}(0)\cdots t_{i_m,i_m}(0)$ on the spin chain, where $B(u) \coloneq t_{12}(u)\in \A(\R)[\![u]\!]$. Note here that these elements are manifestly symmetric in the rapidities $u_i$ and, when acting on the vacuum state $|\Omega\rangle \coloneq |{\downarrow}\rangle \otimes \cdots \otimes |{\downarrow}\rangle$, give states with $N$ flipped up-spins. Therefore, linear combinations of such operators provide a good candidate for the $N$-magnon creation operators that must appear in the long-range Algebraic Bethe Ansatz. In fact, in Appendix~\ref{sec:Thetamorphismrelation} in \eqref{eq:NmagnonstateTheta}, we indeed find that creation operators of this form correspond to the $N$-magnon states for the $B[\mathbb{Q}_3]$-deformed XXX spin chain. Namely, in Appendix~\ref{sec:Thetamorphismrelation} we are able to map these $N$-magnon states to the $N$-magnon states that were found in \cite{gromov2014theta} using the $\Theta$-morphism. Of course, we have not given a full proof that these are the $N$-magnon states, nor have we provided a proper Algebraic Bethe Ansatz that gives the eigenenergies and the Bethe equations for the long-range models. However, one should note that it follows from Proposition~\ref{prop:vanishingDrinfeldassociatorlocaldef}, \ref{prop:vanishingDrinfeldassociatorBQk} or \ref{prop:vanishingdrinfeldassociatorbilocal} that, for $\vec c = (c_1,...,c_N)$, there is an $R$-matrix
    \begin{equation}
        \R^{\gamma}_{(\underline{a}\vec b)(\underline{\vec c}\vec d)}(v, u_{\vec c}) \coloneq \mathbf{r}^{\gamma}(\T_{\underline a}\otimes \T_{\vec b}, \T_{\underline{\vec c}}\otimes \T_{\vec d}) =\mathbf{r}^{\gamma}(\T_{\underline a}(v)\otimes \T_{\vec b}, \T_{\underline{c}_1}(u_1)\cdots \T_{\underline{c}_N}(u_N)\otimes \T_{\vec d})
    \end{equation}
    that gives the braiding between the ordinary Lax operator $\L(v)$ and the Lax operator with multiple spectral parameters $\L(u_1,...,u_N)$ as given in \eqref{eq:laxmultiplespecralparams}. It is then expected that this braiding relation allows one to ``commute'' the creation operators of the form $B(u_1)\cdots B(u_N)\otimes t_{i_1i_1}(0)\cdots t_{i_m,i_m}(0)$ through the transfer matrix constituents to obtain a ``wanted'' and an ``unwanted'' term, thus providing a proper Algebraic Bethe Ansatz.

    \item \textbf{Reflection equation.} We have now provided a consistent method of getting Lax operators and $R$-matrices that correspond to the long-range deformation. However, this only gives deformations on the closed and (possibly quasi-) periodic spin chain. In order to describe open spin chains, one should additionally find solutions to the reflection equation for the twisted $R$-matrix \cite{Cherednik1984FactorizingParticles, Sklyanin1988BoundaryConditions, Kulish1992ReflectionEquation}. Remark that if $K_{\underline{a}}(u)$ is a solution to the reflection equation $\R_{\underline{ab}}(u,v)K_{\underline{a}}(u)\R_{\underline{ba}}(v,-u)K_{\underline{b}}(v) = K_{\underline{b}}(v)\R_{\underline{ab}}(u,-v)K_{\underline{a}}(u)\R_{\underline{ba}}(-v,-u)$, then $K_{\underline{a}}(u)$ is also a solution to the reflection equation for the (undeformed) $R$-matrix $\R_{(\underline{a}\vec b)(\underline{c}\vec d)}(u,v)$ on the increased auxiliary space. Importantly, after twisting the $R$-matrix, one finds that $K^\F(u)$ is a solution to the twisted reflection equation if (not necessarily only if) it satisfies
    \begin{equation}\label{eq:twistedreflectioncondition}
        \F_{(\underline{a}\vec{b})(\underline{c}\vec{d})}(u,v)K_{\underline{a}}(u) = K^\F_{\underline{a}\vec b}(u)\F_{(\underline{a}\vec b)(\underline{c}\vec d)}(-u,v). 
    \end{equation}
    For general twists, it is unknown whether there is such a matrix $K^\F(u)$ for any given $K(u)$. However, remark that the twists for the local and boost operators as given in \eqref{eq:twistlocaldeformation} and \eqref{eq:BQktwistelement}, respectively, act trivially on the coordinate $a$. In particular, this implies that $K_{\underline{a}\vec b}^\F(u) = K_{\underline{a}}(u)$ gives a solution to \eqref{eq:twistedreflectioncondition}, such that $K(u)$ is also a solution to the twisted reflection equation for the local and boost deformation. One might use this solution to the reflection equation to study the long-range deformations on an open spin chain for the local and boost operators \cite{Beisert2009OpenLongrange, Loebbert2012Recursionopen, JiangLoebbertZhong2022Irreleventdefs}. For the deformation corresponding to the bilocal operators, however, it is currently unknown what the solution to the reflection equation would be.

    \item \textbf{Proving the conjecture.} As a side result, we found a conjecture for the explicit expression of the charge densities in terms of the algebraic charge densities. A possible future research direction would therefore be to prove Conjecture~\ref{conj:chargedensities}. Here, it is expected that the description for the higher-order corrections to the deformation of the quantum group requires additional identities for the algebraic charge densities. These additional identities might already be enough to prove the statement. 

    \item \textbf{AdS/CFT.} The second charge in the $\B[\mathbb{Q}_3]$ long-range deformation of the XXX Heisenberg spin chain appears as the two-loop dilatation operator for the $\mathfrak{su}(2)$ sector of planar $\mathcal{N}=4$ SYM. We therefore also provide a quantum group theoretical description for this subsector up to two-loop. By the AdS/CFT correspondence, one might also expect this same quantum group structure to appear in the (classical) string theory. However, it is currently unknown how to make the connection to the corresponding (super) Yangian algebra on the string theory side \cite{Beisert:2006fmy, BeisertdeLeeuwRTTreal2014}. Moreover, it is known that, in the comparison between the gauge and string theories, one should be careful with the order of limits \cite{Klebanov2002NewEffects, Serban2004Planar}. In particular, the thermodynamic limit and the expansion in the coupling constant do not commute. Interestingly, a similar phenomenon is observed on the quantum group level. In particular, as discussed in Section~\ref{sec:LaxforBQ3derivation}, the deformed monodromy matrix on the doubly-infinite spin chain is given by $\L^{\lambda}_{\underline{a},(...)} = \L_{\underline{a},(...)} + \lambda [\mathcal{X}(0), \L_{\underline{a},(...)}] + \O(\lambda^2)$. Remark that, since the first-order term is given by a commutator, this monodromy matrix satisfies the RLL relation with the original undeformed $R$-matrix $\R(u,v)$. Alternatively, long-range deformations on finite spin chains are described by the deformed Lax operators $\L^{\lambda}_{(\underline{a}\vec b)n}$ for which there is a deformed $R$-matrix $\R^{\lambda}_{(\underline{a}\vec b)(\underline{c}\vec d)}$. We therefore immediately see two different descriptions for the long-range deformation on the doubly-infinite spin chain, depending on the order of the thermodynamic limit and the deformation. Namely, by first doing the thermodynamic limit and then the deformation, one obtains a monodromy matrix of which the corresponding $R$-matrix is given by the undeformed $\R(u,v)$. If one first performs the deformation and then takes the thermodynamic limit, one obtains a monodromy matrix of which the corresponding $R$-matrix is given by the deformed $\R^{\lambda}_{(\underline{a}\vec b)(\underline{c}\vec d)}$. It is expected that the second description is the correct one for the correspondence with string theory, as it is this description that also includes the wrapping corrections.
\end{itemize}

\section*{Acknowledgements}
We would like to thank N. Beisert, T. Gombor and A. Torrielli for comments on the draft of this paper. We would also like to thank M. Bartova, F. Loebbert and A.L. Retore for useful discussions and comments. MdL and KS are supported by ERC-2022-CoG-FAIM 101088193. MdL was also supported by  SFI and the Royal Society for funding under grants UF160578, RGF$\backslash $ R1$\backslash $181011, RGF$\backslash$8EA$\backslash $180167 and RF$\backslash $ ERE $\backslash $210373.

\begin{appendices}
    \section{Proofs}\label{sec:proofsfordeformation}
    In this appendix, we will provide all the proofs for the statements made in the main text in Sections~\ref{sec:algebraicchargedensities} and \ref{sec:deformationFRTalgebra}.

    \subsection{Proofs for the algebraic charge densities}\label{sec:proofsforalgebraicchargedensities}
    We will first consider the proofs of the various propositions and lemmas that are given in Section~\ref{sec:algebraicchargedensities}. Firstly, we consider the proof for Proposition~\ref{prop:coproductChargeDensities}:
    \begin{proof}[\textbf{Proof of Proposition~\ref{prop:coproductChargeDensities}}]
        First note that $\Q^{(2)}_{1(23)} = \R^{-1}_{12}\Q^{(2)}_{13}\R_{12} + \Q^{(2)}_{12}$ and $\tilde \Q^{(2)}_{(12)3} = \R^{-1}_{23}\tilde \Q^{(2)}_{13}\R_{23} + \tilde \Q^{(2)}_{23}$, such that the statement is true for $k=2$. The higher-order charges are then defined by $\Q^{(k)}_{12} = \D_1^{k-2}(\Q^{(2)}_{12})$ and $\tilde \Q^{(k)}_{12} = \D_2^{k-2}(\tilde \Q^{(2)}_{12})$ for $k\geq 3$. Therefore, we can calculate the coproducts of the higher-order charges by repeated application of the derivations on $\Q^{(2)}_{1(23)}$. Note here that
        \begin{equation*}
            \D_1\left(\R^{-1}_{12}\Q^{(k)}_{13}\R_{12}\right) = \R^{-1}_{12}\Q^{(k+1)}_{13}\R_{12} + [\R^{-1}_{12}\Q^{(k)}_{13}\R_{12}, \Q^{(2)}_{12}],
        \end{equation*}
        and similarly
        \begin{equation*}
            \D_3\left(\R^{-1}_{23}\tilde \Q^{(k)}_{13}\R_{23}\right) = \R^{-1}_{23}\tilde \Q^{(k+1)}_{13}\R_{23} + [\tilde \Q^{(2)}_{23}, \R^{-1}_{23}\tilde \Q^{(k)}_{13}\R_{23}].
        \end{equation*}
        By repeated application of the derivation, this then immediately shows that the coproducts $\Q^{(k)}_{1(23)}$ and $\tilde \Q^{(k)}_{(12)3}$ must be of the form
        \begin{equation*}
            \label{eq:intermediatecoproducuQnrule}
            \begin{split}
                \Q^{(k)}_{1(23)} &= \Q^{(k)}_{12} + \R^{-1}_{12}\Q^{(k)}_{13}\R_{12}\\
                &\quad + \sum_{n=2}^{k-1}\ \sum_{\ell_1+\cdots+\ell_n = k-1}c^{k-1}_{\ell_1,...,\ell_n}\left[\left[\cdots \left[\left[\R^{-1}_{12}\Q^{(\overline{\ell_n})}_{13}\R_{12}, \Q^{(\overline{\ell_1})}_{12}\right],\Q^{(\overline{\ell_2})}_{12}\right],\cdots \right], \Q^{(\overline{\ell_{n-1}})}_{12}\right],
            \end{split}
        \end{equation*}
        and
        \begin{equation*}
            \begin{split}
                \tilde \Q^{(k)}_{(12)3} &= \tilde \Q^{(k)}_{23} + \R^{-1}_{23}\tilde \Q^{(k)}_{13}\R_{23}\\
                &\quad + \sum_{n=2}^{k-1}\ \sum_{\ell_1+\cdots+\ell_n = k-1}\tilde c^{k-1}_{\ell_1,...,\ell_n}\left[\tilde \Q^{(\overline{\ell_{n-1}})}_{23}, \left[\tilde \Q^{(\overline{\ell_{n-2}})}_{23},\left[\cdots , \left[\tilde \Q_{23}^{(\overline{\ell_1})}, \R^{-1}_{23}\tilde \Q_{13}^{(\overline{\ell_n})}\R_{23}\right]\cdots\right]\right]\right],
            \end{split}
        \end{equation*}
        where $\overline{\ell_i} = \ell_i+1$ and $\ell_i \geq 1$ for some coefficients $c^k_{\ell_1,...,\ell_n} \in \Z_{\geq 0}$ and $\tilde c^k_{\ell_1,...,\ell_n} \in \Z_{\geq 0}$. These coefficients must now satisfy some recurrence relations. Firstly, note here that
        \begin{equation*}
            \begin{split}
                &\D_1\left(\left[\left[\cdots \left[\left[\R^{-1}_{12}\Q^{(\overline{\ell_n})}_{13}\R_{12}, \Q^{(\overline{\ell_1})}_{12}\right],\Q^{(\overline{\ell_2})}_{12}\right],\cdots \right], \Q^{(\overline{\ell_{n-1}})}_{12}\right]\right)\\
                &= \left[\left[\cdots \left[\left[\left[\R^{-1}_{12}\Q^{(\overline{\ell_n})}_{13}\R_{12}, \Q^{(\overline{1})}_{12}\right],\Q^{(\overline{\ell_1})}_{12}\right], \Q^{(\overline{\ell_2})}_{12}\right],\cdots \right], \Q^{(\overline{\ell_{n-1}})}_{12}\right]\\
                &\quad +\sum_{i=1}^{n}\left[\left[\cdots\left[\left[\cdots \left[\left[\R^{-1}_{12}\Q^{(\overline{\ell_n})}_{13}\R_{12}, \Q^{(\overline{\ell_1})}_{12}\right],\Q^{(\overline{\ell_2})}_{12}\right],\cdots \right],\Q^{(\overline{\ell_i}+1)}_{12}\right],\cdots\right], \Q^{(\overline{\ell_{n-1}})}_{12}\right].
            \end{split}
        \end{equation*}
        Using this relation, it can be seen that the coefficient $c^{k-1}_{\ell_1,..., \ell_n}$ can be obtained from summing over all coefficients $c^{k-2}_{\ell_1,..., \ell_i-1,...,\ell_n}$ for $1\leq i\leq n$ if $\ell_1 \neq 1$. If $\ell_1 = 1$, then there is an additional contribution to the coefficient $c^{k-1}_{1,\ell_2,\ell_3,...,\ell_n}$, which comes from the factor corresponding to the coefficient $c^{k-2}_{\ell_2,\ell_3,...,\ell_n}$ after taking the derivation in $\Q^{(k-1)}_{1(23)}$ or $\tilde \Q^{(k-1)}_{(12)3}$. This then shows that the coefficients must satisfy the recurrence relation 
        \begin{equation}
            \label{eq:coproductrecurrencerelation}
            c^k_{\ell_1,...,\ell_n} = \sum_{i=1}^{n}c^{k-1}_{\ell_1,...,\ell_i-1,...,\ell_n}.
        \end{equation}
        where $c^k_k = 1$, $c^k_{0,\ell_2,...,l_n} = c^{k}_{\ell_2,...,\ell_n}$ and $c^k_{\ell_1,...,\ell_n} = 0$ if $\ell_i = 0$ for some $2\leq i\leq n$. It can be easily seen that the same recurrence relation holds for the coefficients $\tilde c^{k-1}_{\ell_1,..., \ell_n}$, such that  $\tilde c^{k-1}_{\ell_1,..., \ell_n} = c^{k-1}_{\ell_1,..., \ell_n}$.\\

        We will now show that this recurrence relation is solved by
        \begin{equation}
            \label{eq:recurrencerelationsolution}
            c^k_{\ell_1,...,\ell_n} = \prod_{m=2}^n{\sum_{i=1}^{m}\ell_i  - 1\choose\sum_{i=1}^{m-1}\ell_i},
        \end{equation}
        for $n\geq 2$ and $\ell_i \geq 1$ for $2\leq i \leq n$. Remember here that $c^k_k = 1$. Firstly, since ${a\choose 0} = 1$ for any $a\in \N$, it can be directly noted that $c^{k}_{0,\ell_2...,\ell_n} = c^k_{\ell_2,...,\ell_n}$. Moreover, if there is some $2\leq m\leq n$ such that $\ell_m = 0$, then \eqref{eq:recurrencerelationsolution} contains a term
        \begin{equation}
            {\sum_{i=1}^{m}\ell_i  - 1\choose\sum_{i=1}^{m-1}\ell_i} = {\sum_{i=1}^{m-1}\ell_i  - 1\choose\sum_{i=1}^{m-1}\ell_i} = 0,
        \end{equation}
        since the top coefficient is now smaller than the bottom coefficient. Therefore, we also have $c^{k}_{\ell_1,..., \ell_n} = 0$ if $\ell_m = 0$ for some $2\leq m\leq n$. It can be noted that $c^k_{\ell_1,..., \ell_n} = {k-1\choose k-\ell_n}c^{k-\ell_n}_{\ell_1,..., \ell_{n-1}}$. Moreover, recall that the binomial coefficients satisfy Pascal's identity
        \begin{equation}
            \label{eq:binomialrecurrencerelation}
            {n\choose k} = {n-1 \choose k-1} + {n-1 \choose k}.
        \end{equation}
        we will now prove that \eqref{eq:recurrencerelationsolution} solves the recurrence relation \eqref{eq:coproductrecurrencerelation} using induction. Firstly, for $n=2$, we have $c^k_{\ell_1, \ell_2} = {k-1\choose \ell_1}$, with $\ell_1 + \ell_2 = k$, such that it follows immediately from \eqref{eq:binomialrecurrencerelation} that $c^{k}_{\ell_1, \ell_2} = c^{k-1}_{\ell_1-1, \ell_2} + c^{k-1}_{\ell_1, \ell_2-1}$. Next, let $n\geq 2$, and assume that the recurrence relation holds for all $m\leq n$. We then have
        \begin{equation*}
            \begin{split}
                c^k_{\ell_1,..., \ell_n} &= {k-1\choose k-\ell_n}c^{k-\ell_n}_{\ell_1,..., \ell_{n-1}} = \left[{(k-1)-1\choose (k-1)-\ell_n} + {(k-1)-1\choose (k-1)-(\ell_n-1)}\right]c^{k-\ell_n}_{\ell_1,..., \ell_{n-1}}\\
                &= \sum_{i=1}^{n-1}{(k-1)-1\choose (k-1)-\ell_n}c^{(k-1)-\ell_n}_{\ell_1,...,\ell_{i}-1, ...,\ell_{n-1}} + {(k-1)-1\choose (k-1)-(\ell_n-1)}c^{(k-1)-(\ell_n-1)}_{\ell_1,..., \ell_{n-1}}\\
                &= \sum_{i=1}^{n}c^{k-1}_{\ell_1,..., \ell_i-1,..., \ell_n}.
            \end{split}
        \end{equation*}
        Therefore, by induction, the coefficients $c^{k}_{\ell_1,..., \ell_n}$ as given in \eqref{eq:recurrencerelationsolution} indeed solve the recurrence relation \eqref{eq:coproductrecurrencerelation}, which then proves the proposition.
    \end{proof}

    Next, we give a proof of Lemma~\ref{lemma:simplifycharges}:
    
    \begin{proof}[\textbf{Proof of Lemma~\ref{lemma:simplifycharges}}]
        We will prove the Lemma using induction and only for $\Q^{(k)}$, as the case for $\tilde \Q^{(k)}$ is completely analogous. First, for $k = 3$, it follows from \eqref{eq:coproductsQ3} that
        \begin{equation*}
            \Q^{(3)}_{1(\underline{\vec a}2\underline{\vec b})} - \R^{-1}_{1\underline{\vec a}}\Q^{(3)}_{1(2\underline{\vec b})}\R_{1\underline{\vec a}} = [\R^{-1}_{1\underline{\vec a}}\Q^{(2)}_{1(2\underline{\vec b})}\R_{1\underline{\vec a}}, \Q^{(2)}_{1\underline{\vec a}}] + \Q^{(3)}_{1\underline{\vec a}}.
        \end{equation*}
        From \eqref{eq:coproductCharge2}, it follows that $\Q^{(2)}_{1(2\underline {\vec b})} = \Q^{(2)}_{12} + \R^{-1}_{12}\Q^{(2)}_{1\underline{\vec b}}\R_{12} = \Q^{(2)}_{12} + \Q^{(2)}_{2\underline {\vec b}}$, where we used that $\R_{12} = \P_{12}$ for evaluation-zero representations. Since $[\R^{-1}_{1\underline{\vec a}}\Q^{(2)}_{2\underline{\vec b}}\R_{1\underline{\vec a}}, \Q^{(2)}_{1\underline{\vec a}}] = 0$, it follows immediately that
        \begin{equation*}
            \Q^{(3)}_{1(\underline{\vec a}2\underline{\vec b})} - \R^{-1}_{1\underline{\vec a}}\Q^{(3)}_{1(2\underline{\vec b})}\R_{1\underline{\vec a}} = \Q^{(3)}_{1(\underline{\vec a}2)} - \R^{-1}_{1\underline{\vec a}}\Q^{(3)}_{12}\R_{1\underline{\vec a}}.
        \end{equation*}
        Now, let $k>3$ and assume that
        \begin{equation*}
            \Q^{(n)}_{1(\underline{\vec a},2,...,n-1,\underline{\vec b})} - \R^{-1}_{1\underline{\vec a}}\Q^{(n)}_{1(2,...,n-1, \underline{\vec b})}\R_{1\underline{\vec a}} = \Q^{(n)}_{1(\underline{\vec a},2,...,n-1)} - \R^{-1}_{1\underline{\vec a}}\Q^{(n)}_{1(2,...,n-1)}\R_{1\underline{\vec a}},
        \end{equation*}
        for all $3\leq n<k$. From Proposition~\ref{prop:coproductChargeDensities} it follows that
        \begin{equation*}
            \begin{split}
                &\Q^{(k)}_{1(\underline{\vec a},2,...,k-1,\underline{\vec b})} - \R^{-1}_{1\underline{\vec a}}\Q^{(k)}_{1(2,...,k-1, \underline{\vec b})}\R_{1\underline{\vec a}} = \\
                &\Q^{(k)}_{1\underline{\vec a}}+ \sum_{m=2}^{k-1}\ \sum_{\ell_1+\cdots+\ell_m = k-1}c^{k-1}_{\ell_1,...,\ell_m}\left[\left[\cdots \left[\R^{-1}_{1\underline{\vec a}}\Q^{(\overline{\ell_m})}_{1(2,...,k-1,\underline{\vec b})}\R_{1\underline{\vec a}}, \Q^{(\overline{\ell_1})}_{1\underline{\vec a}}\right],\cdots \right], \Q^{(\overline{\ell_{m-1}})}_{1\underline{\vec a}}\right].
            \end{split}
        \end{equation*}
        Note here that $\overline{\ell_m} < k$ in the sum on the right side. Therefore, by the induction hypothesis, we have
        \begin{equation*}
            \begin{split}
                \Q^{(\overline{\ell_m})}_{1(2,...,k-1,\underline{\vec b})} &= \Q^{(\overline{\ell_m})}_{1(2,...,k-1)} + \R^{-1}_{12}\left(\Q^{(\overline{\ell_m})}_{1(3,...,k-1,\underline{\vec b})} - \Q^{(\overline{\ell_m})}_{1(3,...,k-1)}\right)\R_{12}\\
                &= \Q^{(\overline{\ell_m})}_{1(2,...,k-1)} + \left(\Q^{(\overline{\ell_m})}_{2(3,...,k-1,\underline{\vec b})} - \Q^{(\overline{\ell_m})}_{2(3,...,k-1)}\right),
            \end{split}
        \end{equation*}
        where was again used that $\R_{12} = \P_{12}$. Note here that the last term in the above equation commutes with $\R_{1\underline{\vec a}}$ and with $\Q^{(\overline{\ell_2})}_{1\underline{\vec a}}$. In particular, this gives
        \begin{equation*}
            \left[\R^{-1}_{1\underline{\vec a}}\Q^{(\overline{\ell_m})}_{1(2,...,k-1,\underline{\vec b})}\R_{1\underline{\vec a}}, \Q^{(\overline{\ell_1})}_{1\underline{\vec a}}\right] = \left[\R^{-1}_{1\underline{\vec a}}\Q^{(\overline{\ell_m})}_{1(2,...,k-1)}\R_{1\underline{\vec a}}, \Q^{(\overline{\ell_1})}_{1\underline{\vec a}}\right].
        \end{equation*}
        It now immediately follows that also
        \begin{equation*}
            \Q^{(k)}_{1(\underline{\vec a},2,...,k-1,\underline{\vec b})} - \R^{-1}_{1\underline{\vec a}}\Q^{(k)}_{1(2,...,k-1, \underline{\vec b})}\R_{1\underline{\vec a}} = \Q^{(k)}_{1(\underline{\vec a},2,...,k-1)} - \R^{-1}_{1\underline{\vec a}}\Q^{(k)}_{1(2,...,k-1)}\R_{1\underline{\vec a}}.
        \end{equation*}
        The result for $\Q^{(k)}$ of the Lemma now follows from induction. The result for $\tilde \Q^{(k)}$ is completely analogous and follows mutatis mutandis.
    \end{proof}

    Next, we will consider an additional lemma for the reduced algebraic charge densities that will be used in the proofs for the long-range deformation of the FRT-bialgebra in Section~\ref{sec:deformationFRTalgebra}.
    
    \begin{lemma}\label{lem:simplifyreducedchargedensities}
        For $k\geq 2$ and $n\geq k-1$, we have
        \begin{equation*}
            \begin{split}
                {}^r\Q^{(k)}_{1,\underline{\vec a \vec b},(2,...,n)} &= \R^{-1}_{\underline{\vec b}(2,...,n)}{}^r\Q^{(k)}_{1,\underline{\vec a},(2,...,n)}\R_{\underline{\vec b}(2,...,n)} + \R^{-1}_{1\underline{\vec a}}{}^r\Q^{(k)}_{1,\underline{\vec b},(2,...,n)}\R_{1\underline{\vec a}},\\
                {}^r\tilde \Q^{(k)}_{(1,...,n-1),\underline{\vec a\vec b}, n} &= \R^{-1}_{(1,...,n-1)\underline{\vec a}}{}^r\tilde \Q^{(k)}_{(1,...,n-2),\underline{\vec b}, n}\R_{(1,...,n-1)\underline{\vec a}} + \R^{-1}_{\underline{\vec b}n}{}^r\tilde \Q^{(k)}_{(1,...,n-1),\underline{\vec a}, n}\R_{\underline{\vec b} n}.
            \end{split}
        \end{equation*}
        \begin{proof}
            Let $k\geq 2$ and $n\geq k-1$. It then follows from the definition of the reduced algebraic charge density ${}^r\Q^{(k)}$ as given in Definition~\ref{def:reducedalgebraicchargedensities} and Lemma~\ref{lemma:simplifycharges} that
            \begin{equation*}
                \begin{split}
                    {}^r\Q^{(k)}_{1,\underline{\vec a \vec b},(2,...,n)} &= \Q^{(k)}_{1(\underline{\vec a \vec b},2,...,n)} - \R^{-1}_{1(\underline{\vec a\vec b})}\Q^{(k)}_{1(2,...,n)}\R_{1(\underline{\vec a\vec b})}\\
                    &= \R^{-1}_{\underline{\vec b}(2,...,n)}\Q^{(k)}_{1(\underline{\vec a},2,...,n,\underline{\vec b})}\R_{\underline{\vec b}(2,...,n)} - \R^{-1}_{1(\underline{\vec a\vec b})}\Q^{(k)}_{1(2,...,n)}\R_{1(\underline{\vec a\vec b})}\\
                    &= R^{-1}_{\underline{\vec b}(2,...,n)}\left(\Q^{(k)}_{1(\underline{\vec a},2,...,n)} + \R^{-1}_{1\underline{\vec a}}\Q^{(k)}_{1(2,...,n,\underline{\vec b})}\R_{1\underline{\vec a}} - \R^{-1}_{1\underline{\vec a}}\Q^{(k)}_{1(2,....,n)}\R_{1\underline{\vec a}}\right)\R_{\underline{\vec b}(2,...,n)}\\
                    &\quad - \R^{-1}_{1(\underline{\vec a\vec b})}\Q^{(k)}_{1(2,...,n)}\R_{1(\underline{\vec a\vec b})}\\
                    &= \R^{-1}_{\underline{\vec b}(2,...,n)}\left(\Q^{(k)}_{1(\underline{\vec a},2,...,n)} - \R^{-1}_{1\underline{\vec a}}\Q^{(k)}_{1(2,....,n)}\R_{1\underline{\vec a}}\right)\R_{\underline{\vec b}(2,...,n)}\\
                    &\quad + \R^{-1}_{1\underline{\vec a}}\left(\Q^{(k)}_{1(\underline{\vec b},2,...,n)} - \R^{-1}_{1\underline{\vec b}}\Q^{(k)}_{1(2,...,n)}\R_{1\underline{\vec b}}\right)\R_{1\underline{\vec a}}\\
                    &= \R^{-1}_{\underline{\vec b}(2,...,n)}{}^r\Q^{(k)}_{1,\underline{\vec a},(2,...,n)}\R_{\underline{\vec b}(2,...,n)} + \R^{-1}_{1\underline{\vec a}}{}^r\Q^{(k)}_{1,\underline{\vec b},(2,...,n)}\R_{1\underline{\vec a}}.
                \end{split}
            \end{equation*}
            The derivation for the reduced algebraic charge density ${}^r\tilde{\Q}^{(k)}$ is completely analogous, thus proving the statement.
        \end{proof}
    \end{lemma}

    Lastly, we will also explicitly write down the proof for Corollary~\ref{cor:commutingcharges}:

    \begin{proof}[\textbf{Proof for Corollary~\ref{cor:commutingcharges}}]
        The proof follows immediately from Proposition~\ref{prop:higherorderSutherlandEquations}. In particular, one has
        \begin{equation*}
            \begin{split}
                \left[\mathbb{Q}^{(k)}, \left(\overleftarrow{\prod}_{m\in \Z} \L_{am}(u)\right)\right] &= \mathrm{tr}_a\left(\sum_{n \in \Z}[\mathfrak{q}^{(k)}_{n,...,n+k-1}, \overleftarrow{\prod}_{m\in \Z} \L_{am}(u)]\right)\\
                &= \sum_{n\in \Z}\mathrm{tr}_a\bigg(\L_{a(n+1,...)}{}^r\Q^{(k)}_{(n,...,n+k-3),a,n+k-2}\L_{a(...,n)}\\
                &\qquad \qquad - \L_{a(n+2,...)}{}^r\Q^{(k)}_{(n,...,n+k-2),a,n+k-1}\L_{a(...,n+1)}\bigg)= 0,
            \end{split}
        \end{equation*}
        where the sum vanishes due to it being a telescoping sum over $\Z$.
    \end{proof}

    \subsection{Proofs for the local deformation}\label{sec:proofsforlocaldeformation}
    This section will be devoted to proving Proposition~\ref{prop:vanishingDrinfeldassociatorlocaldef}. Before we can prove this proposition, we will first consider the following intermediate result:
    \begin{lemma}\label{lem:localdeftwistproperties}
        Let $k\geq 2$ and $\mathfrak{m} \in \mathrm{End}(V^{\otimes k})$, then the twist element $\gamma^{(1)}_\mathfrak{m}$ as defined in \eqref{eq:twistlocaldeformation} satisfies
        \begin{equation*}
            \begin{split}
                \gamma^{(1)}_\mathfrak{m}(a\otimes b\cdot b', c\otimes d) &=  \epsilon(b)\gamma^{(1)}_\mathfrak{m}(a\otimes b', c\otimes d)\\
                &\quad + \sum \mathbf{r}^{-1}(b'_{(1)}, c_{(1)})\gamma^{(1)}_\mathfrak{m}(a\otimes b, c_{(2)}\otimes b'_{(2)}\cdot d)\mathbf{r}(b'_{(3)}, c_{(3)}),\\
                \gamma^{(1)}_\mathfrak{m}(a\otimes b, c\cdot c'\otimes d) &= \sum \mathbf{r}^{-1}(b_{(1)}, c_{(1)})\gamma^{(1)}_\mathfrak{m}(a\cdot c_{(2)} \otimes b_{(2)}, c'\otimes d)\mathbf{r}(b_{(3)}, c_{(3)})\\
                &\quad + \sum \mathbf{r}^{-1}(c'_{(1)}, d_{(1)})\gamma^{(1)}_\mathfrak{m}(a\otimes b, c\otimes d_{(2)}) \mathbf{r}(c'_{(2)}, d_{(3)}),\\
                \gamma^{(1)}_\mathfrak{m}(a\otimes b, c\otimes d\cdot d') &= \gamma^{(1)}_{\mathfrak{m}}(a\otimes b, c\otimes d)\epsilon(d')
            \end{split}
        \end{equation*}
        for all $a,c,c'\in \A(\R)$ and $b,b,d,d' \in \A_0$ such that $c\otimes d \in \A_0 \oplus \left(\A^{(\geq 1)}(\R)\bowtie_{\mathbf{r}}\A_0^{(\geq k-1)}\right)$.
        \begin{proof}
            For the first equality, let $0< N < M < K$, then from the definition of $\gamma^{(1)}_{\mathfrak{m}}$ as given in \eqref{eq:twistlocaldeformation}, it follows that
            \begin{equation}
                \begin{split}
                    &\gamma^{(1)}_{\mathfrak{m}}(\T_{\underline{\vec a}}\otimes \T_1\cdots \T_M, \T_{\underline{\vec b}}\otimes \T_{M+1}\cdots \T_K)\\
                    &= \sum_{n=N+1}^{M}\R^{-1}_{\underline{\vec b}(M+1,...,K)}\mathfrak{m}_{n,...,n+k-1}\R_{\underline{\vec b}(M+1,...,K)} - \R^{-1}_{(1,...,M)\underline{\vec b}}\mathfrak{m}_{n,...,n+k-1}\R_{(1,...,M)\underline{\vec b}}\\
                    &\quad +\sum_{n=1}^{N}\R^{-1}_{\underline{\vec b}(M+1,...,K)}\mathfrak{m}_{n,...,n+k-1}\R_{\underline{\vec b}(M+1,...,K)} - \R^{-1}_{(1,...,M)\underline{\vec b}}\mathfrak{m}_{n,...,n+k-1}\R_{(1,...,M)\underline{\vec b}}\\
                    &= \gamma^{(1)}_{\mathfrak{m}}(\T_{\underline{\vec a}}\otimes \T_{N+1}\cdots\T_M, \T_{\underline{\vec b}}\otimes \T_{M+1}\cdots \T_K)\\
                    &\quad +\R^{-1}_{(N+1,...,M)\underline{\vec b}}\gamma^{(1)}_{\mathfrak{m}}(\T_{\underline{\vec a}}\otimes \T_{1}\cdots \T_N, \T_{\underline{\vec b}}\otimes \T_{N+1}\cdots \T_K)\R_{(N+1,...,M)\underline{\vec b}},
                \end{split}
            \end{equation}
            where we used that $[\mathfrak{m}_{n,...,n+k-1}, \R_{(1,...,N)\underline{\vec b}}] = 0$ for $n \geq N+1$ in the second equality. This then proves the first equality.\\

            For the second equality, let $N > 0$ and $M > N + k-2$, then
            \begin{equation}
                \begin{split}
                    &\gamma^{(1)}_{\mathfrak{m}}(\T_{\underline{\vec a}}\otimes \T_1\cdots \T_N, \T_{\underline{\vec b}}\T_{\underline{\vec c}}\otimes \T_{N+1}\cdots \T_M)\\
                    &= \sum_{n=1}^{N}\R^{-1}_{(\underline{\vec b}\underline{\vec c})(N+1,...,M)}\mathfrak{m}_{n,...,n+k-1}\R_{(\underline{\vec b}\underline{\vec c})(N+1,...,M)} - \R^{-1}_{(1,...,N)(\underline{\vec b}\underline{\vec c})}\mathfrak{m}_{n,...,n+k-1}\R_{(1,...,N)(\underline{\vec b}\underline{\vec c})}\\
                    &= \R^{-1}_{(1,...,N)\underline{\vec b}}\bigg[\sum_{n=1}^{N}\R^{-1}_{\underline{\vec c}(N+1,...,M)}\mathfrak{m}_{n,...,n+k-1}\R_{\underline{\vec c}(N+1,...,M)}\\
                &\qquad\qquad\qquad\qquad- \R^{-1}_{(1,...,N)\underline{\vec c}}\mathfrak{m}_{n,...,n+k-1}\R_{(1,...,N)\underline{\vec c}}\bigg]\R_{(1,...,N)\underline{\vec b}}\\
                &+\R^{-1}_{\underline{\vec c}(N+1,...,M)}\bigg[\sum_{n=1}^{N}\R^{-1}_{\underline{\vec b}(N+1,...,M)}\mathfrak{m}_{n,...,n+k-1}\R_{\underline{\vec b}(N+1,...,M)}\\
                &\qquad\qquad\qquad\qquad- \R^{-1}_{(1,...,N)\underline{\vec b}}\mathfrak{m}_{n,...,n+k-1}\R_{(1,...,N)\underline{\vec b}}\bigg]\R_{\underline{\vec c}(N+1,...,M)}\\
                &= \R^{-1}_{(1,...,N)\underline{\vec b}}\gamma^{(1)}_{\mathfrak{m}}(\T_{\underline{\vec a}}\T_{\underline{\vec b}}\otimes \T_{1}\cdots \T_N, \T_{\underline{\vec c}}\otimes \T_{N+1}\cdots \T_M)\R_{(1,...,N)\underline{\vec b}}\\
                &\quad+ \R^{-1}_{\underline{\vec c}(N+1,...,M)}\gamma^{(1)}_{\mathfrak{m}}(\T_{\underline{\vec a}}\otimes \T_{1}\cdots \T_N, \T_{\underline{\vec b}}\otimes \T_{N+1}\cdots \T_M)\R_{\underline{\vec c}(N+1,...,M)},
                \end{split}
            \end{equation}
            where in the second equality, it can be seen that the first term in the sum cancels the fourth term in the sum. Note here that the equality would still hold for $\T_{\underline{\vec b}}\to 1$ for any $M > N$, thus proving the second equality.\\

            For the last equality, let $N>0$, $M > N+k-2$ and $ K > M$, then we have
            \begin{equation}
                \begin{split}
                    &\gamma^{(1)}_{\mathfrak{m}}(\T_{\underline{\vec a}}\otimes \T_1\cdots \T_N, \T_{\underline{\vec b}}\otimes \T_{N+1}\cdots \T_K)\\
                    &= \sum_{n=1}^{N}\R^{-1}_{\underline{\vec b}(N+1,...,K)}\mathfrak{m}_{n,...,n+k-1}\R_{\underline{\vec b}(N+1,...,K)} - \R^{-1}_{(1,...,N)\underline{\vec b}}\mathfrak{m}_{n,...,n+k-1}\R_{(1,...,N)\underline{\vec b}}\\
                    &= \sum_{n=1}^{N}\R^{-1}_{\underline{\vec b}(N+1,...,M)}\mathfrak{m}_{n,...,n+k-1}\R_{\underline{\vec b}(N+1,...,M)} - \R^{-1}_{(1,...,N)\underline{\vec b}}\mathfrak{m}_{n,...,n+k-1}\R_{(1,...,N)\underline{\vec b}}\\
                    &=\gamma^{(1)}_{\mathfrak{m}}(\T_{\underline{\vec a}}\otimes \T_1\cdots \T_N, \T_{\underline{\vec b}}\otimes \T_{N+1}\cdots \T_M)\epsilon(\T_{M+1}\cdots \T_K),
                \end{split}
            \end{equation}
            where in the second equality we used that $n+k-1 \leq N+k-1 < M+1$, such that $[\mathfrak{m}_{n,...,n+k-1}, \R_{\underline{\vec b}(M+1,...,K)}] = 0$. Also, remark that the statement is trivially true under the substitution $\T_{\underline{\vec b}}\to 1$ and for any $0 < N < M <K$, thus proving the last equality and therefore the statement.
        \end{proof}
    \end{lemma}

    We can now prove Proposition~\ref{prop:vanishingDrinfeldassociatorlocaldef}.

    \begin{proof}[\textbf{Proof of Proposition~\ref{prop:vanishingDrinfeldassociatorlocaldef}}]
        Let $a\otimes b,c\otimes d,e\otimes f\in \A_0 \oplus \left(\A^{(\geq 1)}(\R)\bowtie_{\mathbf{r}}\A_0^{(\geq k-1)}\right)$, then
        \begin{equation}
            \begin{split}
                \varphi^{(1)}_\mathfrak{m}(a\otimes b, c\otimes d, e\otimes f) &= \gamma^{(1)}_\mathfrak{m}(c\otimes d, e\otimes f)\epsilon(a)\epsilon(b) + \gamma^{(1)}_\mathfrak{m}(a\otimes b, (c\otimes d)(e\otimes f))\\
                & \quad - \gamma^{(1)}_\mathfrak{m}((a\otimes b)(c\otimes d), e\otimes f) - \gamma^{(1)}_\mathfrak{m}(a\otimes b, c\otimes d)\epsilon(e)\epsilon(f).
            \end{split}
        \end{equation}
        From Lemma~\ref{lem:localdeftwistproperties}, it now follows that
        \begin{equation}
            \begin{split}
                &\gamma^{(1)}_\mathfrak{m}(a\otimes b, (c\otimes d)(e\otimes f)) = \sum \mathbf{r}^{-1}(d_{(1)}, e_{(1)})\gamma^{(1)}_\mathfrak{m}(a\otimes b, c\cdot e_{(2)}\otimes d_{(2)}\cdot f)\mathbf{r}(d_{(3)}, e_{(3)})\\
                &= \sum \mathbf{r}^{-1}(d_{(1)}, e_{(1)})\mathbf{r}^{-1}(b_{(1)}, c_{(1)})\gamma^{(1)}_\mathfrak{m}(a\cdot c_{(2)}\otimes b_{(2)}, e_{(2)}\otimes d_{(2)}f)\mathbf{r}(b_{(3)}, c_{(3)})\mathbf{r}(d_{(3)}, e_{(3)})\\
                &\quad +\sum\mathbf{r}^{-1}(e_{(1)}, f_{(1)})\gamma^{(1)}_\mathfrak{m}(a\otimes b, c\otimes d\cdot f_{(2)})\mathbf{r}(e_{(2)},f_{(3)})\\
                &= \sum \mathbf{r}^{-1}(d_{(1)}, e_{(1)})\mathbf{r}^{-1}(b_{(1)}, c_{(1)})\gamma^{(1)}_\mathfrak{m}(a\cdot c_{(2)}\otimes b_{(2)}, e_{(2)}\otimes d_{(2)}f)\mathbf{r}(b_{(3)}, c_{(3)})\mathbf{r}(d_{(3)}, e_{(3)})\\
                &\quad+\gamma^{(1)}_\mathfrak{m}(a\otimes b, c\otimes d)\epsilon(e)\epsilon(f),
            \end{split}
        \end{equation}
        where in the second equality we used that $\sum \mathbf{r}(e_{(1)}, d_{(1)})\mathbf{r}(d_{(2)}, e_{(2)}) = \epsilon(e)\epsilon(d)$ since $\mathbf{r}$ is cotriangular. Similarly, we have
        \begin{equation}
            \begin{split}
                &\gamma^{(1)}_\mathfrak{m}((a\otimes b)(c\otimes d), e\otimes f) = \sum \mathbf{r}^{-1}(b_{(1)}, c_{(1)})\gamma_\mathfrak{m}^{(1)}(a\cdot c_{(2)}\otimes b_{(2)}\cdot d, e\otimes f)\mathbf{r}(b_{(3)}, c_{(3)})\\
                &= \sum \mathbf{r}^{-1}(b_{(1)}, c_{(1)})\mathbf{r}^{-1}(d_{(1)},e_{(1)})\gamma_\mathfrak{m}^{(1)}(a\cdot c_{(2)}\otimes b_{(2)}, e_{(2)}\otimes d_{(2)}\cdot f)\mathbf{r}(d_{(3)}, e_{(3)})\mathbf{r}(b_{(3)}, c_{(3)})\\
                &\quad +\gamma_\mathfrak{m}^{(1)}(a\cdot c \otimes d, e\otimes f)\epsilon(b).
            \end{split}
        \end{equation}
        By combining all the terms, and using the fact that $\gamma^{(1)}_\mathfrak{m}(a\cdot c\otimes d, e\otimes f) = \gamma^{(1)}_\mathfrak{m}(c\otimes d, e\otimes f)\epsilon(a)$, we find that indeed the Drinfeld associator vanishes, i.e. $\varphi^{(1)}_\mathfrak{m}(a\otimes b, c\otimes d, e\otimes f) = 0$, thus proving the proposition.
    \end{proof}

    \subsection{Proofs for the boost deformation}\label{sec:proofsforboostdeformation}
    This section will be devoted to proving Proposition~\ref{prop:vanishingDrinfeldassociatorBQk}. In order to prove this proposition, we will use the following intermediate result: 

    \begin{lemma}\label{lemma:twistkproperties}
        Let $k\geq 3$, then the twist element $\gamma_k^{(1)}$ as defined in \eqref{eq:BQktwistelement} satisfies
        \begin{equation*}
            \begin{split}
                \gamma^{(1)}_k(a\otimes b\cdot b', c\otimes d) &=  \epsilon(b)\gamma^{(1)}_k(a\otimes b', c\otimes d)\\
                &\quad + \sum \mathbf{r}^{-1}(b'_{(1)}, c_{(1)})\gamma^{(1)}_k(a\otimes b, c_{(2)}\otimes b'_{(2)}\cdot d)\mathbf{r}(b'_{(3)}, c_{(3)}),\\
                \gamma^{(1)}_k(a\otimes b, c\cdot c'\otimes d) &= \sum \mathbf{r}^{-1}(b_{(1)}, c_{(1)})\gamma^{(1)}_k(a\cdot c_{(2)} \otimes b_{(2)}, c'\otimes d)\mathbf{r}(b_{(3)}, c_{(3)})\\
                &\quad + \sum \mathbf{r}^{-1}(c'_{(1)}, d_{(1)})\gamma^{(1)}_k(a\otimes b, c\otimes d_{(2)}) \mathbf{r}(c'_{(2)}, d_{(3)}),\\
                \gamma^{(1)}_k(a\otimes b, c\otimes d\cdot d') &= \gamma^{(1)}_k(a\otimes b, c\otimes d)\epsilon(d')
            \end{split}
        \end{equation*}
        for all $a,c,c'\in \A(\R)$ and $b,b,d,d' \in \A_0$ such that $c\otimes d \in \A_0 \oplus \left(\A^{(\geq 1)}(\R)\bowtie_{\mathbf{r}}\A_0^{(\geq k-2)}\right)$.
        \begin{proof}
            Let $0<N<M<K$. From the definition of the twist $\gamma_k^{(1)}$ as given in \eqref{eq:BQktwistelement}, it then follows that
            \begin{equation*}
                \begin{split}
                    &\gamma^{(1)}_k(\T_{\underline{\vec a}}\otimes \T_1\cdots \T_M, \T_{\underline{\vec b}}\otimes \T_{M+1}\cdots \T_K)\\
                    &=\sum_{n=N+1}^{M}\R^{-1}_{(n+1,...,M)\underline{\vec b}}{}^r\Q^{(k)}_{n,\underline{\vec b}, (n+1,...,K)}\R_{(n+1,...,M)\underline{\vec b}}\\
                    &\quad + \sum_{n=1}^{N}\R^{-1}_{(n+1,...,M)\underline{\vec b}}{}^r\Q^{(k)}_{n,\underline{\vec b}, (n+1,...,K)}\R_{(n+1,...,M)\underline{\vec b}}\\
                    & = \gamma_k^{(1)}(\T_{\underline{\vec a}}\otimes \T_{N+1}\cdots \T_M, \T_{\underline{\vec b}}\otimes \T_{M+1}\cdots \T_K) \\
                    &\quad + \R^{-1}_{(N+1,...,M)\underline{\vec b}}\sum_{n=1}^{N}\R^{-1}_{(n+1,...,N)\underline{\vec b}}{}^r \Q^{(k)}_{n,\underline{\vec b},(n+1,...,K)}\R_{(n+1,...,N)\underline{\vec b}}\R_{(N+1,...,M)\underline{\vec b}}\\
                    &= \gamma_k^{(1)}(\T_{\underline{\vec a}}\otimes \T_{N+1}\cdots \T_M, \T_{\underline{\vec b}}\otimes \T_{M+1}\cdots \T_K) \\
                    &\quad + \R^{-1}_{(N+1,...,M)\underline{\vec b}}\gamma^{(1)}_k(\T_{\underline{\vec a}}\otimes \T_{1}\cdots \T_N, \T_{\underline{\vec b}}\otimes \T_{N+1}\cdots \T_K)\R_{(N+1,...,M)\underline{\vec b}},
                \end{split}
            \end{equation*}
            which proves the first equality.\\
            
            For the second equality, let $N>0$ and $M> N+k-3$. Using Lemma~\ref{lem:simplifyreducedchargedensities}, we get
            \begin{equation}\label{eq:lemmaproofgammastep1}
                \begin{split}
                    &\gamma^{(1)}_k(\T_{\underline{\vec a}}\otimes \T_1\cdots \T_N, \T_{\underline{\vec b}}\T_{\underline{\vec c}}\otimes \T_{N+1}\cdots \T_M)\\
                    &= \sum_{n=1}^{N}\R^{-1}_{(n+1,...,N)(\underline{\vec b}\underline{\vec c})}{}^r\Q^{(k)}_{n, (\underline{\vec b\vec c}), (n+1,...,M)}\R_{(n+1,...,N)(\underline{\vec b \vec c})}\\
                    &= \R^{-1}_{(N+1,...,M)\underline{\vec c}}\sum_{n=1}^{N}\R^{-1}_{(n+1,...,N)\underline{\vec b}}{}^r\Q^{(k)}_{n, \underline{\vec b}, (n+1,...,M)}\R_{(n+1,...,N)\underline{\vec b}}\R_{(N+1,...,M)\underline{\vec c}}\\
                    &\quad + \R^{-1}_{(1,...,N)\underline{\vec b}}\sum_{n=1}^{N}\R^{-1}_{(n+1,...,N)\underline{\vec c}}{}^r\Q^{(k)}_{n, \underline{\vec c}, (n+1,...,M)}\R_{(n+1,...,N)\underline{\vec c}}\R_{(1,...,N)\underline{\vec b}}\\
                    &= \R^{-1}_{(N+1,...,M)\underline{\vec c}}\gamma^{(1)}_k(\T_{\underline{\vec a}}\otimes\T_{1}\cdots \T_N, \T_{\underline{\vec b}}\otimes \T_{N+1}\cdots \T_M)\R_{(N+1,...,M)\underline{\vec c}}\\
                    & \quad + \R^{-1}_{(1,...,N)\underline{\vec b}}\gamma^{(1)}_k(\T_{\underline{\vec a}}\T_{\underline{\vec b}}\otimes\T_{1}\cdots \T_N, \T_{\underline{\vec c}}\otimes \T_{N+1}\cdots \T_M)\R_{(1,...,N)\underline{\vec b}}.
                \end{split}
            \end{equation}
            It can be readily verified that the equality still holds under the substitution $\T_{\underline{\vec b}}\to 1$ and for all $0 < N < M$, thus proving the second equality.\\

            For the last equality, let $N  > 0$, $M > N+k-3$ and $K > M$. Using Lemma~\ref{lemma:simplifycharges}, it immediately follows that
            \begin{equation*}
                \begin{split}
                    &\gamma^{(1)}_k(\T_{\underline{\vec a}}\otimes \T_{1}\cdots \T_N, \T_{\underline{\vec b}}\otimes \T_{1}\cdots \T_K)\\
                    &= \sum_{n=1}^{N}\R^{-1}_{(n+1,...,N)\underline{\vec b}}{}^r\Q^{(k)}_{n, \underline{\vec b}, (n+1,...,K)}\R_{(n+1,...,N)\underline{\vec b}}\\
                    &=\sum_{n=1}^{N}\R^{-1}_{(n+1,...,N)\underline{\vec b}}{}^r\Q^{(k)}_{n, \underline{\vec b}, (n+1,...,M)}\R_{(n+1,...,N)\underline{\vec b}}\\
                    &= \gamma^{(1)}_k(\T_{\underline{\vec a}}\otimes \T_{1}\cdots \T_N, \T_{\underline{\vec b}}\otimes \T_{1}\cdots \T_M)\epsilon(\T_{M+1}\cdots \T_K).
                \end{split}
            \end{equation*}
            Note that this equality also holds trivially true in the case $\T_{\underline{\vec b}}\to 1$ with arbitrary $M > N$ and $K > M$. Thus, we prove the last equality and therefore the lemma.
        \end{proof}
    \end{lemma}

    We can now prove Proposition~\ref{prop:vanishingDrinfeldassociatorBQk}:

    \begin{proof}[\textbf{Proof of Proposition~\ref{prop:vanishingDrinfeldassociatorBQk}}]
        Let $a\otimes b,c\otimes d,e\otimes f\in \A_0 \oplus \left(\A^{(\geq 1)}(\R)\bowtie_{\mathbf{r}}\A_0^{(\geq k-2)}\right)$, then
        \begin{equation}
            \begin{split}
                \varphi^{(1)}_k(a\otimes b, c\otimes d, e\otimes f) &= \gamma^{(1)}_k(c\otimes d, e\otimes f)\epsilon(a)\epsilon(b) + \gamma^{(1)}_k(a\otimes b, (c\otimes d)(e\otimes f))\\
                & \quad - \gamma^{(1)}_k((a\otimes b)(c\otimes d), e\otimes f) - \gamma^{(1)}_k(a\otimes b, c\otimes d)\epsilon(e)\epsilon(f).
            \end{split}
        \end{equation}
        From Lemma~\ref{lemma:twistkproperties}, it now follows that
        \begin{equation}
            \begin{split}
                &\gamma^{(1)}_k(a\otimes b, (c\otimes d)(e\otimes f)) = \sum \mathbf{r}^{-1}(d_{(1)}, e_{(1)})\gamma^{(1)}_k(a\otimes b, c\cdot e_{(2)}\otimes d_{(2)}\cdot f)\mathbf{r}(d_{(3)}, e_{(3)})\\
                &= \sum \mathbf{r}^{-1}(d_{(1)}, e_{(1)})\mathbf{r}^{-1}(b_{(1)}, c_{(1)})\gamma^{(1)}_k(a\cdot c_{(2)}\otimes b_{(2)}, e_{(2)}\otimes d_{(2)}\cdot f)\mathbf{r}(b_{(3)}, c_{(3)})\mathbf{r}(d_{(3)}, e_{(3)})\\
                &\quad +\sum\mathbf{r}^{-1}(e_{(1)}, f_{(1)})\gamma^{(1)}_k(a\otimes b, c\otimes d\cdot f_{(2)})\mathbf{r}(e_{(2)},f_{(3)})\\
                &= \sum \mathbf{r}^{-1}(d_{(1)}, e_{(1)})\mathbf{r}^{-1}(b_{(1)}, c_{(1)})\gamma^{(1)}_k(a\cdot c_{(2)}\otimes b_{(2)}, e_{(2)}\otimes d_{(2)}\cdot f)\mathbf{r}(b_{(3)}, c_{(3)})\mathbf{r}(d_{(3)}, e_{(3)})\\
                &\quad+\gamma^{(1)}_k(a\otimes b, c\otimes d)\epsilon(e)\epsilon(f),
            \end{split}
        \end{equation}
        where in the second equality we used that $\sum \mathbf{r}(e_{(1)}, d_{(1)})\mathbf{r}(d_{(2)}, e_{(2)}) = \epsilon(e)\epsilon(d)$ since $\mathbf{r}$ is cotriangular. Similarly, we have
        \begin{equation}
            \begin{split}
                &\gamma^{(1)}_k((a\otimes b)(c\otimes d), e\otimes f) = \sum \mathbf{r}^{-1}(b_{(1)}, c_{(1)})\gamma_k^{(1)}(a\cdot c_{(2)}\otimes b_{(2)}\cdot d, e\otimes f)\mathbf{r}(b_{(3)}, c_{(3)})\\
                &= \sum \mathbf{r}^{-1}(b_{(1)}, c_{(1)})\mathbf{r}^{-1}(d_{(1)},e_{(1)})\gamma_k^{(1)}(a\cdot c_{(2)}\otimes b_{(2)}, e_{(2)}\otimes d_{(2)}\cdot f)\mathbf{r}(d_{(3)}, e_{(3)})\mathbf{r}(b_{(3)}, c_{(3)})\\
                &\quad +\gamma_k^{(1)}(a\cdot c \otimes d, e\otimes f)\epsilon(b).
            \end{split}
        \end{equation}
        By combining all the terms, and using the fact that $\gamma^{(1)}_k(a\cdot c\otimes d, e\otimes f) = \gamma^{(1)}_k(c\otimes d, e\otimes f)\epsilon(a)$, we find that indeed the Drinfeld associator vanishes, i.e. $\varphi^{(1)}_k(a\otimes b, c\otimes d, e\otimes f) = 0$, thus proving the proposition.
    \end{proof}

    \subsection{Proofs for the bilocal deformation}\label{sec:proofsforbilocal}
    This section will be devoted to proving Proposition~\ref{prop:vanishingdrinfeldassociatorbilocal}. In order to break up the proof in more digestible parts, we will write the twist $\gamma^{(1)}_{k|\ell}$ as given in \eqref{eq:twistbilocal} into the three parts
    \begin{equation}\label{eq:subtwistelementsbilocal}
        \begin{split}
            &{}^1\gamma^{(1)}_{k | \ell}(\T_{\underline{\vec a}}\otimes \T_1\cdots \T_N, (1\otimes \T_{N+1}\cdots \T_M)(\T_{\underline{\vec b}}\otimes 1))\\
            &\coloneq \R^{-1}_{\underline{\vec a}(1,...,N+\ell-2)}{}^r\tilde{\Q}^{(k)}_{(1,...,N+\ell-2),\underline{\vec a},N+\ell-1}\R_{\underline{\vec a}(1,...,N+\ell-2)}\\
            &\qquad \qquad \qquad \qquad \qquad  \times\R_{\underline{\vec b}(N+3-\ell,...,M)}{}^r\Q^{(\ell)}_{N+2-\ell, \underline{\vec b},(N+3-\ell,...M)}\R^{-1}_{\underline{\vec b}(N+3-\ell,...,M)}\\
            &+ \sum_{n=N+2,}^{N+\ell-1}\sum_{m=N+2-\ell}^{n-\ell}\R^{-1}_{\underline{\vec a}(1,...,M)}[\tilde{\mathfrak{q}}^{(k)}_{n-k+1,...,n}, \R_{\underline{\vec a}(1,...,M)}][\mathfrak{q}^{(\ell)}_{m,...,m+\ell-1}, \R_{\underline{\vec b}(1,...,M)}]\R^{-1}_{\underline{\vec b}(1,...,M)},\\
            &{}^2\gamma^{(1)}_{k | \ell}(\T_{\underline{\vec a}}\otimes \T_1\cdots \T_N, (1\otimes \T_{N+1}\cdots \T_M)(\T_{\underline{\vec b}}\otimes 1))\\
            &\coloneq \sum_{n=1}^{N}\R^{-1}_{\underline{\vec a}(1,...,M)}\tilde{\mathfrak{q}}^{(k)}_{n+\ell-k,...,n+\ell-1}\R_{\underline{\vec a}(1,....M)}\R_{\underline{\vec b}(n+1,...,M)}{}^r\Q^{(\ell)}_{n,\underline{\vec b},(n+1,...,M)}\R_{\underline{\vec b}(n+1,...,M)}^{-1},\\
            &{}^3\gamma^{(1)}_{k | \ell}(\T_{\underline{\vec a}}\otimes \T_1\cdots \T_N, (1\otimes \T_{N+1}\cdots \T_M)(\T_{\underline{\vec b}}\otimes 1))\\
            &\coloneq \sum_{\mathclap{n=N+1}}^{M}\ \R^{-1}_{\underline{\vec a}(1,...,n-1)}{}^r\tilde{\Q}^{(k)}_{(1,...,n-1),\underline{\vec a}, n}\R_{\underline{\vec a}(1,...,n-1)}\R_{\underline{\vec b}(1,...,M)}\mathfrak{q}^{(\ell)}_{n-\ell+1,...,n}\R^{-1}_{\underline{\vec b}(1,...,M)},
        \end{split}
    \end{equation}
    such that $\gamma^{(1)}_{k | \ell} = {}^1\gamma^{(1)}_{k | \ell} + {}^2\gamma^{(1)}_{k | \ell} + {}^3\gamma^{(1)}_{k | \ell}$. We now have the following intermediate result:
    \begin{lemma}\label{lem:usefulbilocaltwistidentities}
        For $2\leq k < \ell$, the twist elements ${}^{2,3}\gamma^{(1)}_{k|\ell}$ as defined in \eqref{eq:subtwistelementsbilocal} satisfy the following relations:
        \begin{itemize}[leftmargin=*]
            \item For $a,c,c'\in \A(\R)$, $b,b',d,d' \in \A_0$, such that $c\otimes d \in \A_0 \oplus \left(\A^{(\geq 1)}(\R)\bowtie_{\mathbf{r}}\A_0^{(\geq \ell-1)}\right)$:
            \begin{equation*}
                \begin{split}
                    {}^2\gamma^{(1)}_{k|\ell}(a\otimes b, c\cdot c' \otimes d) &= \sum \mathbf{r}^{-1}(c'_{(1)}, d_{(1)})\ {}^2\gamma^{(1)}_{k|\ell}(a\otimes b, c\otimes d_{(2)})\mathbf{r}(c'_{(2)}, d_{(3)})\\
                    &\quad + \sum \mathbf{r}^{-1}(b_{(1)}, c_{(1)})\ {}^2\gamma^{(1)}_{k|\ell}(a\cdot c_{(2)}\otimes b_{(2)}, c'\otimes d)\mathbf{r}(b_{(3)}, c_{(3)}),\\
                    {}^{2}\gamma^{(1)}_{k|\ell}(a\otimes b, c\otimes d\cdot d') &= {}^2\gamma^{(1)}_{k|\ell}(a\otimes b, c\otimes d)\epsilon(d'),\\
                    {}^2\gamma^{(1)}_{k|\ell}(a\otimes b\cdot b', c\otimes d) &= \sum \mathbf{r}^{-1}(b'_{(1)}, c_{(1)})\ {}^2\gamma^{(1)}_{k|\ell}(a\otimes b, c_{(2)}\otimes b'_{(2)}\cdot d)\mathbf{r}(b'_{(3)}, c_{(3)})\\
                    &\quad +\sum \mathbf{r}^{-1}(a_{(1)}, b_{(1)})\ {}^2\gamma^{(1)}_{k|\ell}(a_{(2)}\otimes b', c\otimes d)\mathbf{r}(a_{(3)}, b_{(2)}).
                \end{split}
            \end{equation*}
            \item For $a,a', c\in \A(\R)$, $b,b', d,d' \in \A_0$ such that $a\otimes b\in \A_0 \oplus \left(\A^{(\geq 1)}(\R)\bowtie_{\mathbf{r}}\A_0^{(\geq \ell-1)}\right)$:
            \begin{equation*}
                \begin{split}
                    {}^3\gamma^{(1)}_{k|\ell}(a\cdot a'\otimes b, c\otimes d) &= \epsilon(a)\ {}^3\gamma^{(1)}_{k|\ell}(a'\otimes b, c\otimes d)\\
                    &\quad + \sum \mathbf{r}^{-1}(a'_{(1)}, b_{(1)})\ {}^3\gamma^{(1)}_{k|\ell}(a\otimes b_{(2)}, a'_{(2)}\cdot c\otimes d)\mathbf{r}(a'_{(3)}, b_{(3)}),\\
                    {}^3\gamma^{(1)}_{k|\ell}(a\otimes b\cdot b', c\otimes d) &= \sum \mathbf{r}^{-1}(a_{(1)}, b_{(1)})\ {}^{3}\gamma^{(1)}_{k|\ell}(a_{(2)}\otimes  b', c\otimes d)\mathbf{r}(a_{(3)}, b_{(2)}),\\
                    {}^3\gamma^{(1)}_{k|\ell}(a\otimes b, c\otimes d\cdot d') &= {}^3\gamma^{(1)}_{k|\ell}(a\otimes b, c\otimes d)\epsilon(d')\\
                    &\quad+\sum \mathbf{r}^{-1}(c_{(1)}, d_{(1)})\ {}^3\gamma^{(1)}_{k|\ell}(a\otimes b\cdot d_{(2)}, c_{(2)}\otimes d')\mathbf{r}(c_{(3)}, d_{(3)}).
                \end{split}
            \end{equation*}
        \end{itemize}

        \begin{proof}
            Let us first consider the equalities for ${}^2\gamma^{(1)}_{k|\ell}$. Let $N>0$ and $M > N+\ell-2$, then
            \begin{equation*}
                \begin{split}
                    &{}^2\gamma^{(1)}_{k|\ell}(\T_{\underline{\vec a}}\otimes\T_1\cdots \T_N, (1\otimes \T_{N+1}\cdots\T_M)(\T_{\underline{\vec b}}\T_{\underline{\vec c}}\otimes 1))\\
                    &=\sum_{n=1}^{N}\R^{-1}_{\underline{\vec a}(1,...,M)}\tilde{\mathfrak{q}}^{(k)}_{n+\ell-k,...,n+\ell-1}\R_{\underline{\vec a}(1,....M)}\R_{(\underline{\vec b \vec c)}(n+1,...,M)}{}^r\Q^{(\ell)}_{n,\underline{\vec b\vec c},(n+1,...,M)}\R_{(\underline{\vec b \vec c})(n+1,...,M)}^{-1}\\
                    &= \sum_{n=1}^{N}\R^{-1}_{\underline{\vec a}(1,...,M)}\tilde{\mathfrak{q}}^{(k)}_{n+\ell-k,...,n+\ell-1}\R_{\underline{\vec a}(1,....M)}\R_{\underline{\vec b}(n+1,...,M)}{}^r\Q^{(\ell)}_{n,\underline{\vec b},(n+1,...,M)}\R_{\underline{\vec b}(n+1,...,M)}^{-1}\\
                    &+\sum_{n=1}^{N}\R^{-1}_{\underline{\vec a}(1,...,M)}\tilde{\mathfrak{q}}^{(k)}_{n+\ell-k,...,n+\ell-1}\R_{\underline{\vec a}(1,....M)}\R_{(\underline{\vec b \vec c})(n+1,...,M)}\\
                    &\qquad\qquad\qquad \qquad \qquad\qquad \qquad \qquad \times\R_{n\underline{\vec b}}^{-1}{}^r\Q^{(\ell)}_{n,\underline{\vec c},(n+1,...,M)}\R_{n\underline{\vec b}}\R_{(\underline{\vec b \vec c})(n+1,...,M)}^{-1},
                \end{split}
            \end{equation*}
             where we used Lemma~\ref{lem:simplifyreducedchargedensities} in the second equality. Using the fact that \begin{equation*}
                [\R_{\underline{\vec b}(1,...,n-1)}, \R_{\underline{\vec c}(n+1,....,M)}{}^r\Q^{(\ell)}_{n,\underline{\vec c},(n+1,...,M)}\R^{-1}_{\underline{\vec c}(n+1,...,M)}] = 0,
            \end{equation*}
            for all $n \leq N$, this then gives
            \begin{equation*}
                \begin{split}
                    &{}^2\gamma^{(1)}_{k|\ell}(\T_{\underline{\vec a}}\otimes\T_1\cdots \T_N, (1\otimes \T_{N+1}\cdots\T_M)(\T_{\underline{\vec b}}\T_{\underline{\vec c}}\otimes 1))\\
                    &= {}^2\gamma^{(1)}_{k|\ell}(\T_{\underline{\vec a}}\otimes\T_1\cdots \T_N, (1\otimes \T_{N+1}\cdots\T_M)(\T_{\underline{\vec b}}\otimes 1))\\
                    &\quad +\R^{-1}_{(1,...,M)\underline{\vec b}}\sum_{n=1}^{N}\R^{-1}_{(\underline{\vec a \vec b})(1,...,M)}\tilde{\mathfrak{q}}^{(k)}_{n+\ell-k,...,n+\ell-1}\R_{(\underline{\vec a\vec b})(1,...,M)}\\
                   & \qquad\qquad \qquad\qquad \times \R_{\underline{\vec c}(n+1,....,M)}{}^r\Q^{(\ell)}_{n,\underline{\vec c},(n+1,...,M)}\R^{-1}_{\underline{\vec c}(n+1,...,M)}\R^{-1}_{\underline{\vec b}(1,...,M)}\\
                   &= {}^2\gamma^{(1)}_{k|\ell}(\T_{\underline{\vec a}}\otimes\T_1\cdots \T_N, (1\otimes \T_{N+1}\cdots\T_M)(\T_{\underline{\vec b}}\otimes 1)) \\
                   &\quad + \R^{-1}_{(1,...,M)\underline{\vec b}}{}^2\gamma^{(1)}_{k|\ell}(\T_{\underline{\vec a}}\T_{\underline{\vec b}}\otimes\T_1\cdots \T_N, (1\otimes \T_{N+1}\cdots\T_M)(\T_{\underline{\vec c}}\otimes 1))\R_{(1,...,M)\underline{\vec b}},
                \end{split}
            \end{equation*}
            Note that the equality will still hold if we replace $\T_{\underline{\vec a}}\to 1$. Moreover, since the twist element satisfies ${}^2\gamma^{(1)}_{k|\ell}(a\otimes b, 1\otimes d) = 0$, the result is also true under the substitution $\T_{\underline{\vec b}}\to 1$ and/or $\T_{\underline{\vec c}}\to 1$ for any $M>N$, thus proving the first equality.\\

            For the second equality, let $N>0$, $M > N+\ell-2$ and $K > M$, then
            \begin{equation}\label{eq:bilocalproofstep1}
                \begin{split}
                    &{}^2\gamma^{(1)}_{k|\ell}(\T_{\underline{\vec a}}\otimes \T_1\cdots \T_N, (1\otimes \T_{N+1}\cdots \T_K)(\T_{\underline{\vec b}}\otimes 1))\\
                    &= \sum_{n=1}^{N}\R^{-1}_{\underline{\vec a}(1,...,K)}\tilde{\mathfrak{q}}^{(k)}_{n+\ell-k,...,n+\ell-1}\R_{\underline{\vec a}(1,....K)}\R_{\underline{\vec b}(n+1,...,K)}{}^r\Q^{(\ell)}_{n,\underline{\vec b},(n+1,...,K)} \R_{\underline{\vec b}(n+1,...,K)}^{-1}\\
                    &= \sum_{n=1}^{N}\R^{-1}_{\underline{\vec a}(1,...,M)}\tilde{\mathfrak{q}}^{(k)}_{n+\ell-k,...,n+\ell-1}\R_{\underline{\vec a}(1,....M)}\R_{\underline{\vec b}(M+1,...,K)}\\
                    &\qquad \qquad \qquad \qquad \times \R_{\underline{\vec b}(n+1,...,M)}{}^r\Q^{(\ell)}_{n,\underline{\vec b},(n+1,...,M)} \R_{\underline{\vec b}(n+1,...,M)}^{-1}\R^{-1}_{\underline{\vec b}(M+1,...,K)}\\
                    &= \R^{-1}_{(M+1,...,K)\underline{\vec b}}{}^2\gamma^{(1)}_{k|\ell}(\T_{\underline{\vec a}}\otimes \T_1\cdots \T_N, (1\otimes \T_{N+1}\cdots \T_M)(\T_{\underline{\vec b}}\otimes 1))\R_{(M+1,...,K)\underline{\vec b}},
                \end{split}
            \end{equation}
            where in the second equality we used the fact that $n+\ell-1 < N+\ell-1 < M+1$ such that
            \begin{equation}
                \R^{-1}_{\underline{\vec a}(1,...,K)}\tilde{\mathfrak{q}}^{(k)}_{n+\ell-k,...,n+\ell-1}\R_{\underline{\vec a}(1,....K)} = \R^{-1}_{\underline{\vec a}(1,...,M)}\tilde{\mathfrak{q}}^{(k)}_{n+\ell-k,...,n+\ell-1}\R_{\underline{\vec a}(1,....M)},
            \end{equation}
            and moreover, we could pull the term $\R_{\underline{\vec b}(M+1,...,K)}$ to the left from the second to the third equality in \eqref{eq:bilocalproofstep1}. Note that under the substitution $\T_{\underline{\vec b}}\to 1$, the twist element $^{3}\gamma^{(1)}_{k|\ell}$ in \eqref{eq:bilocalproofstep1} vanishes, such that the result in that case holds for any $K>M>N>0$, thus proving the second equality for $^{3}\gamma^{(1)}_{k|\ell}$.\\

            Lastly, let $0<N<M$ and $K > M+\ell-2$, then
            \begin{equation}
                \begin{split}
                    &{}^2\gamma^{(1)}_{k|\ell}(\T_{\underline{\vec a}}\otimes \T_1\cdots \T_M, (1\otimes \T_{M+1}\cdots \T_K)(\T_{\underline{\vec b}}\otimes 1))\\
                    &= \sum_{n=1}^{N}\R^{-1}_{\underline{\vec a}(1,...,K)}\tilde{\mathfrak{q}}^{(k)}_{n+\ell-k,...,n+\ell-1}\R_{\underline{\vec a}(1,....K)}\R_{\underline{\vec b}(n+1,...,M)}{}^r\Q^{(\ell)}_{n,\underline{\vec b},(n+1,...,K)}\R_{\underline{\vec b}(n+1,...,K)}^{-1}\\
                    &\quad+\sum_{n=N+1}^{M}\R^{-1}_{\underline{\vec a}(1,...,K)}\tilde{\mathfrak{q}}^{(k)}_{n+\ell-k,...,n+\ell-1}\R_{\underline{\vec a}(1,....K)}\R_{\underline{\vec b}(n+1,...,M)}{}^r\Q^{(\ell)}_{n,\underline{\vec b},(n+1,...,K)}\R_{\underline{\vec b}(n+1,...,K)}^{-1}\\
                    &= {}^2\gamma^{(1)}_{k|\ell}(\T_{\underline{\vec a}}\otimes \T_{1}\cdots \T_N, (1\otimes \T_{M+1}\cdots \T_K)(\T_{\underline{\vec b}}\otimes 1))\\
                    &\quad + \R^{-1}_{\underline{\vec a}(1,...,N)}{}^2\gamma^{(1)}_{k|\ell}(\T_{\underline{\vec a}}\otimes \T_{N+1}\cdots \T_M, (1\otimes \T_{M+1}\cdots \T_K)(\T_{\underline{\vec b}}\otimes 1))\R_{\underline{\vec a}(1,...,N)}.
                \end{split}
            \end{equation}
            Again, note here that ${}^2\gamma^{(1)}_{k|\ell}$ vanishes under substitution $\T_{\underline{\vec b}}\to 1$, such that the above equality in that case also holds for any $K>M$. This then proves the last equality for ${}^2\gamma^{(1)}_{k|\ell}$.\\

            All the equalities for ${}^3\gamma^{(1)}_{k|\ell}$ will follow in a completely analogous way to the equalities ${}^2\gamma^{(1)}_{k|\ell}$, which then proves the statement.
        \end{proof}
    \end{lemma}

    Secondly, we consider the following intermediate result:

    \begin{lemma}\label{lem:usefulidentitybilocal2}
        For $2\leq k < \ell$ and $a\otimes b, c\otimes d, e\otimes f \in \A_0 \oplus \left(\A^{(\geq 1)}(\R)\bowtie_{\mathbf{r}}\A_0^{(\geq \ell-1)}\right)$, we have
        \begin{equation*}
            \begin{split}
                &\epsilon(a\cdot b)\ {}^2\gamma^{(1)}_{k|\ell}(c\otimes d, e\otimes f) - \sum \mathbf{r}^{-1}(a_{(1)}, b_{(1)})\ {}^2\gamma^{(1)}_{k|\ell}(a_{(2)}\cdot c\otimes d, e\otimes f)\mathbf{r}(a_{(3)}, b_{(2)})\\
                &+ \sum \mathbf{r}^{-1}(d_{(1)}, e_{(1)})\ {}^3\gamma^{(1)}_{k|\ell}(a\otimes b, c\cdot e_{(2)}\otimes d_{(2)})\epsilon(f)\mathbf{r}(d_{(3)}, e_{(3)}) - {}^3\gamma^{(1)}_{k|\ell}(a\otimes b, c\otimes d)\epsilon(e\cdot f)\\
                &= \sum \mathbf{r}^{-1}(b_{(1)}, c_{(1)})\ {}^1 \gamma^{(1)}_{k|\ell}(a\cdot c_{(2)}\otimes b_{(2)} d, e\otimes f)\mathbf{r}(b_{(3)}, c_{(3)})-\epsilon(a\cdot b)\ {}^1\gamma^{(1)}_{k|\ell}(c\otimes d, e\otimes f)\\
                &+ \sum {}^1 \gamma^{(1)}_{k|\ell}(a\otimes b, c\otimes d)\epsilon(e\cdot f) -  \mathbf{r}^{-1}(d_{(1)}, e_{(1)})\ {}^1\gamma^{(1)}_{k|\ell}(a\otimes b, c e_{(2)}\otimes d_{(2)})\epsilon(f)\mathbf{r}(d_{(3)}, e_{(3)}).
            \end{split}
        \end{equation*}
        \begin{proof}
            Let $N > \ell-2$, $M> N+ \ell-2$ and $K > M +\ell-2$, then  
            \begin{equation*}
                \begin{split}
                    &{}^2\gamma^{(1)}_{k|\ell}(\T_{\underline{\vec a}}\T_{\underline{\vec b}}\otimes \T_{N+1}\cdots \T_{M}, (1\otimes \T_{M+1}\cdots \T_K)(1\otimes \T_{\underline{\vec c}}))\\
                    &= \sum_{\mathclap{n=N+1}}^{M}\R^{-1}_{(\underline{\vec a\vec b})(N+1,...,K)}\tilde{\mathfrak{q}}^{(k)}_{n+\ell-k,...,n+\ell-1}\R_{(\underline{\vec a \vec b})(N+1,....K)}\R_{\underline{\vec c}(n+1,...,K)}{}^r\Q^{(\ell)}_{n,\underline{\vec c},(n+1,...,K)}\R_{\underline{\vec c}(n+1,...,K)}^{-1}\\
                    &= \sum_{\mathclap{n=N+1}}^{M}\R^{-1}_{\underline{\vec b}(N+1,...,K)}\tilde{\mathfrak{q}}^{(k)}_{n+\ell-k,...,n+\ell-1}\R_{\underline{\vec b}(N+1,....K)}\R_{\underline{\vec c}(n+1,...,K)}{}^r\Q^{(\ell)}_{n,\underline{\vec c},(n+1,...,K)}\R_{\underline{\vec c}(n+1,...,K)}^{-1}\\
                    &\quad +\sum_{\mathclap{n=N+1}}^{M}\R^{-1}_{\underline{\vec b}(N+1,...,K)}\R^{-1}_{\underline{\vec a}(N+1,...,K)}[\tilde{\mathfrak{q}}^{(k)}_{n+\ell-k,...,n+\ell-1}, \R_{\underline{\vec a}(N+1,...,K)}]\R_{\underline{\vec b}(N+1,....K)}\\
                    &\qquad \qquad \qquad \qquad \qquad \qquad \qquad \times \R_{\underline{\vec c}(n+1,...,K)}{}^r\Q^{(\ell)}_{n,\underline{\vec c},(n+1,...,K)}\R_{\underline{\vec c}(n+1,...,K)}^{-1}.
                \end{split}
            \end{equation*}
            From the Sutherland equations as given in Proposition~\ref{prop:higherorderSutherlandEquations}, it follows that
            \begin{equation*}
                \begin{split}
                    \R^{-1}_{\underline{\vec a}(1,...,K)}[\tilde{\mathfrak{q}}^{(k)}_{n+\ell-k,...,n+\ell-1}, \R_{\underline{\vec a}(1,...,K)}] &= \R^{-1}_{\underline{\vec a}(1,...,n+\ell-3)}{}^r\tilde{\Q}^{(k)}_{(1,...,n+\ell-3),\underline{\vec a}, n+\ell-2}\R_{\underline{\vec a}(1,...,n+\ell-3)}\\
                    &- \R^{-1}_{\underline{\vec a}(1,...,n+\ell-2)}{}^r\tilde{\Q}^{(k)}_{(1,...,n+\ell-2),\underline{\vec a}, n+\ell-1}\R_{\underline{\vec a}(1,...,n+\ell-2)}.
                \end{split}
            \end{equation*}
            Combining these results, one finds that
            \begin{equation}\label{eq:proofbilocalstep2}
                \begin{split}
                    &{}^2\gamma^{(1)}_{k|\ell}(\T_{\underline{\vec b}}\otimes \T_{N+1}\cdots \T_M, \T_{\underline{\vec c}}\otimes \T_{M+1}\cdots \T_K)\\
                    &\quad -\R^{-1}_{\underline{\vec a}(1,...,N)}{}^2\gamma^{(1)}_{k|\ell}(\T_{\underline{\vec a}}\T_{\underline{\vec b}}\otimes \T_{N+1}\cdots \T_M, \T_{\underline{\vec c}}\otimes \T_{M+1}\cdots \T_K)\R_{\underline{\vec a}(1,...,N)}\\
                    &= \sum_{\mathclap{n=N+1}}^{M}\R^{-1}_{\underline{\vec c}(M+1,...,K)}\R^{-1}_{\underline{\vec b}(N+1,...,K)}\R^{-1}_{\underline{\vec a}(1,...,n+\ell-2)}{}^r\tilde{\Q}^{(k)}_{(1,...,n+\ell-2),\underline{\vec a}, n+\ell-1}\R_{\underline{\vec a}(1,...,n+\ell-2)}\\
                    &\qquad \qquad \qquad \qquad \qquad \qquad \qquad \times \R_{\underline{\vec b}(N+1,....K)}\R_{\underline{\vec c}(n+1,...,K)}{}^r\Q^{(\ell)}_{n,\underline{\vec c},(n+1,...,K)}\R_{\underline{\vec c}(n+1,...,M)}^{-1}\\
                    & - \sum_{\mathclap{n=N+1}}^{M}\R^{-1}_{\underline{\vec c}(M+1,...,K)}\R^{-1}_{\underline{\vec b}(N+1,...,K)}\R^{-1}_{\underline{\vec a}(1,...,n+\ell-3)}{}^r\tilde{\Q}^{(k)}_{(1,...,n+\ell-3),\underline{\vec a}, n+\ell-2}\R_{\underline{\vec a}(1,...,n+\ell-3)}\\
                    &\qquad \qquad \qquad \qquad \qquad \qquad \qquad \times \R_{\underline{\vec b}(N+1,....K)}\R_{\underline{\vec c}(n+1,...,K)}{}^r\Q^{(\ell)}_{n,\underline{\vec c},(n+1,...,K)}\R_{\underline{\vec c}(n+1,...,M)}^{-1}.
                \end{split}
            \end{equation}
            In an analogous way, it can be found that
            \begin{equation}\label{eq:proofbilocalstep3}
                \begin{split}
                    &\R^{-1}_{(N+1,...,M)\underline{\vec c}}{}^3\gamma^{(1)}_{k|\ell}(\T_{\underline{\vec a}}\otimes \T_{1}\cdots \T_N, \T_{\underline{\vec b}}\T_{\underline{\vec c}}\otimes \T_{N+1,...,M})\R_{(N+1,...,M)\underline{\vec c}}\\
                    &\quad - {}^3\gamma^{(1)}_{k|\ell}(\T_{\underline{\vec a}}\otimes \T_{1}\cdots \T_N, \T_{\underline{\vec b}}\otimes \T_{N+1}\cdots \T_{M})\\
                    &= \sum_{n=N+1}^{M}\R^{-1}_{\underline{\vec b}(N+1,...,M)}\R^{-1}_{\underline{\vec a}(1,...,n-1)}{}^r\tilde{\Q}^{(k)}_{(1,...,n-1), \underline{\vec a}, n}\R_{\underline{\vec a}(1,...,n-1)}\R_{\underline{\vec b}(1,...,M)}\\
                    &\qquad \qquad\qquad \qquad  \times \R_{\underline{\vec c}(n-\ell+3,...,M)}{}^r\Q^{(\ell)}_{n-\ell+2,\underline{\vec c}, (n-\ell+3,..., M)}\R^{-1}_{\underline{\vec c}(n-\ell+3,...,M)}\R^{-1}_{\underline{\vec b}(1,...,N)}\\
                    &-\sum_{n=N+1}^{M}\R^{-1}_{\underline{\vec b}(N+1,...,M)}\R^{-1}_{\underline{\vec a}(1,...,n-1)}{}^r\tilde{\Q}^{(k)}_{(1,...,n-1), \underline{\vec a}, n}\R_{\underline{\vec a}(1,...,n-1)}\R_{\underline{\vec b}(1,...,M)}\\
                    &\qquad \qquad \qquad \qquad \times \R_{\underline{\vec c}(n-\ell+2,...,M)}{}^r\Q^{(\ell)}_{n-\ell+1,\underline{\vec c}, (n-\ell+2,..., M)}\R^{-1}_{\underline{\vec c}(n-\ell+2,...,M)}\R^{-1}_{\underline{\vec b}(1,...,N)}.
                \end{split}
            \end{equation}
            Importantly, if we would replace $n\to n+\ell-2$ and $n\to n+\ell-1$ in the first and second sum of \eqref{eq:proofbilocalstep3}, respectively, then these terms cancel most of the terms in the sum in \eqref{eq:proofbilocalstep2} up to some boundary terms if we sum up \eqref{eq:proofbilocalstep2} and \eqref{eq:proofbilocalstep3}. In particular, we find that
            \begin{equation*}
                \begin{split}
                    &{}^2\gamma^{(1)}_{k|\ell}(\T_{\underline{\vec b}}\otimes \T_{N+1}\cdots \T_M, \T_{\underline{\vec c}}\otimes \T_{M+1}\cdots \T_K)\\
                    &\quad -\R^{-1}_{\underline{\vec a}(1,...,N)}{}^2\gamma^{(1)}_{k|\ell}(\T_{\underline{\vec a}}\T_{\underline{\vec b}}\otimes \T_{N+1}\cdots \T_M, \T_{\underline{\vec c}}\otimes \T_{M+1}\cdots \T_K)\R_{\underline{\vec a}(1,...,N)}\\
                    &+\R^{-1}_{(N+1,...,M)\underline{\vec c}}{}^3\gamma^{(1)}_{k|\ell}(\T_{\underline{\vec a}}\otimes \T_{1}\cdots \T_N, \T_{\underline{\vec b}}\T_{\underline{\vec c}}\otimes \T_{N+1,...,M})\R_{(N+1,...,M)\underline{\vec c}}\\
                    &\quad - {}^3\gamma^{(1)}_{k|\ell}(\T_{\underline{\vec a}}\otimes \T_{1}\cdots \T_N, \T_{\underline{\vec b}}\otimes \T_{N+1}\cdots \T_{M})\\
                    &= A+B,
                \end{split}
            \end{equation*}
            where
            \begin{equation*}
                \begin{split}
                    A &= \sum_{\mathclap{n=M+2-\ell}}^{M}\R^{-1}_{\underline{\vec c}(M+1,...,K)}\R^{-1}_{\underline{\vec b}(N+1,...,K)}\R^{-1}_{\underline{\vec a}(1,...,n+\ell-2)}{}^r\tilde{\Q}^{(k)}_{(1,...,n+\ell-2),\underline{\vec a}, n+\ell-1}\R_{\underline{\vec a}(1,...,n+\ell-2)}\\
                    &\qquad \qquad \qquad \qquad \qquad \qquad \times \R_{\underline{\vec b}(N+1,....K)}\R_{\underline{\vec c}(n+1,...,K)}{}^r\Q^{(\ell)}_{n,\underline{\vec c},(n+1,...,K)}\R_{\underline{\vec c}(n+1,...,M)}^{-1}\\
                    & - \sum_{\mathclap{n=M+3-\ell}}^{M}\R^{-1}_{\underline{\vec c}(M+1,...,K)}\R^{-1}_{\underline{\vec b}(N+1,...,K)}\R^{-1}_{\underline{\vec a}(1,...,n+\ell-3)}{}^r\tilde{\Q}^{(k)}_{(1,...,n+\ell-3),\underline{\vec a}, n+\ell-2}\R_{\underline{\vec a}(1,...,n+\ell-3)}\\
                    &\qquad \qquad \qquad \qquad \qquad \qquad \times \R_{\underline{\vec b}(N+1,....K)}\R_{\underline{\vec c}(n+1,...,K)}{}^r\Q^{(\ell)}_{n,\underline{\vec c},(n+1,...,K)}\R_{\underline{\vec c}(n+1,...,M)}^{-1}
                \end{split}
            \end{equation*}
            and
            \begin{equation*}
                \begin{split}
                    B &= \sum_{n=N+1}^{N+\ell-2}\R^{-1}_{\underline{\vec b}(N+1,...,M)}\R^{-1}_{\underline{\vec a}(1,...,n-1)}{}^r\tilde{\Q}^{(k)}_{(1,...,n-1), \underline{\vec a}, n}\R_{\underline{\vec a}(1,...,n-1)}\R_{\underline{\vec b}(1,...,M)}\\
                    &\qquad \qquad\qquad \qquad  \times \R_{\underline{\vec c}(n-\ell+3,...,M)}{}^r\Q^{(\ell)}_{n-\ell+2,\underline{\vec c}, (n-\ell+3,..., M)}\R^{-1}_{\underline{\vec c}(n-\ell+3,...,M)}\R^{-1}_{\underline{\vec b}(1,...,N)}\\
                    &-\sum_{n=N+1}^{N+\ell-1}\R^{-1}_{\underline{\vec b}(N+1,...,M)}\R^{-1}_{\underline{\vec a}(1,...,n-1)}{}^r\tilde{\Q}^{(k)}_{(1,...,n-1), \underline{\vec a}, n}\R_{\underline{\vec a}(1,...,n-1)}\R_{\underline{\vec b}(1,...,M)}\\
                    &\qquad \qquad \qquad \qquad \times \R_{\underline{\vec c}(n-\ell+2,...,M)}{}^r\Q^{(\ell)}_{n-\ell+1,\underline{\vec c}, (n-\ell+2,..., M)}\R^{-1}_{\underline{\vec c}(n-\ell+2,...,M)}\R^{-1}_{\underline{\vec b}(1,...,N)}.
                \end{split}
            \end{equation*}
            Note here that $A$ and $B$ now contain sums over only $\ell-1$ terms. Let us evaluate $A$. Firstly, it can be seen that the two sums in the domain $n=M+3-\ell$ to $M$ can be combined, in which one can recognise the Sutherland equation as given in Proposition~\ref{prop:higherorderSutherlandEquations}. In particular, this gives
            \begin{equation*}
                \begin{split}
                    A &= \R^{-1}_{\underline{\vec c}(M+1,...,K)}\R^{-1}_{\underline{\vec b}(N+1,...,K)}\R^{-1}_{\underline{\vec a}(1,...,M)}{}^r\tilde{\Q}^{(k)}_{(1,...,M),\underline{\vec a}, M+1}\R_{\underline{\vec a}(1,...,M)}\\
                    &\qquad \qquad \qquad \times \R_{\underline{\vec b}(N+1,...,K)}\R_{\underline{\vec c}(M+3-\ell,...,K)}{}^r\Q^{(\ell)}_{M+2-\ell,\underline{\vec c}, (M+3-\ell,...,K)}\R^{-1}_{\underline{\vec c}(M+3-\ell,...,M)}\\
                    &+\sum_{n=M+3-\ell}^{M}\R^{-1}_{\underline{\vec c}(M+1,...,K)}\R^{-1}_{\underline{\vec b}(N+1,...,K)}\R^{-1}_{\underline{\vec a}(1,...,K)}[\R_{\underline{\vec a}(1,...K)}, \tilde{\mathfrak{q}}^{(k)}_{n+\ell-k,...,n+\ell-1}]\\
                    &\qquad \qquad \qquad \times  \R_{\underline{\vec b}(N+1,....K)}\R_{\underline{\vec c}(n+1,...,K)}{}^r\Q^{(\ell)}_{n,\underline{\vec c},(n+1,...,K)}\R_{\underline{\vec c}(n+1,...,M)}^{-1}
                \end{split}
            \end{equation*}
            Next, by repeated application of the Sutherland equations as given in Proposition~\ref{prop:higherorderSutherlandEquations}, it follows that
            \begin{equation*}
                \begin{split}
                    \R^{-1}_{\underline{\vec a}(1,...,M)}{}^r\tilde{\Q}^{(k)}_{(1,...,M),\underline{\vec a}, M+1}\R_{\underline{\vec a}(1,...,M)} &= \R^{-1}_{\underline{\vec a}(1,...,M+\ell-2)}{}^r\tilde{\Q}^{(k)}_{(1,...,M+\ell-2),\underline{\vec a}, M+\ell-1}\R_{\underline{\vec a}(1,...,M+\ell-2)}\\
                    &+\sum_{n=M+3-\ell}^{M}\R^{-1}_{\underline{\vec a}(1,...,K)}[\tilde{\mathfrak{q}}^{(k)}_{n+\ell-k,...,n+\ell-1}, \R_{\underline{\vec a}(1,...,K)}],
                \end{split}
            \end{equation*}
            and similarly, for $n\geq M+3-\ell$,
            \begin{equation*}
                \begin{split}
                    &\R_{\underline{\vec c}(M+3-\ell,...,K)}{}^r\Q^{(\ell)}_{M+2-\ell,\underline{\vec c},(M+3-\ell,...,K)}\R^{-1}_{\underline{\vec c}(M+3-\ell,...,K)}\\
                    &= \R^{-1}_{\underline{\vec c}(n+1,...,K)}{}^r\Q^{(\ell)}_{n,\underline{\vec c},(n+1,...,K)}\R^{-1}_{\underline{\vec c}(n+1,...,K)} +\sum_{m=M+2-\ell}^{n-1}[\mathfrak{q}^{(\ell)}_{m,...,m+\ell-1}, \R_{\underline{\vec c}(1,...,K)}]\R^{-1}_{\underline{\vec c}(1,...,K)}.
                \end{split}
            \end{equation*}
            Combining these facts then gives
            \begin{equation*}
                \begin{split}
                    A &= \R^{-1}_{\underline{\vec c}(M+1,...,K)}\R^{-1}_{\underline{\vec b}(N+1,...,K)}\R^{-1}_{\underline{\vec a}(1,...,M+\ell-2)}{}^r\tilde{\Q}^{(k)}_{(1,...,M+\ell-2),\underline{\vec a}, M+\ell-1}\R_{\underline{\vec a}(1,...,M+\ell-2)}\\
                    &\qquad \qquad \qquad \times \R_{\underline{\vec b}(N+1,...,K)}\R_{\underline{\vec c}(M+3-\ell,...,K)}{}^r\Q^{(\ell)}_{M+2-\ell,\underline{\vec c}, (M+3-\ell,...,K)}\R^{-1}_{\underline{\vec c}(M+3-\ell,...,M)}\\
                    &+\sum_{n=M+3-\ell}^{M}\sum_{\ m = M+2-\ell}^{n-1}\R^{-1}_{\underline{\vec c}(M+1,...,K)}\R^{-1}_{\underline{\vec b}(N+1,...,K)}\R^{-1}_{\underline{\vec a}(1,...,K)}[\tilde{\mathfrak{q}}^{(k)}_{n+\ell-k,...,n+\ell-1},\R_{\underline{\vec a}(1,...K)}]\\
                    &\qquad \qquad \qquad \times  \R_{\underline{\vec b}(N+1,....K)}[\mathfrak{q}^{(\ell)}_{m,...,m+\ell-1}, \R_{\underline{\vec c}(1,...,K)}]\R^{-1}_{\underline{\vec c}(1,...,M)}.
                \end{split}
            \end{equation*}
            Using Lemma~\ref{lem:simplifyreducedchargedensities}, it then immediately follows from the definitions of the twist element ${}^1\gamma^{(1)}_{k|\ell}$ as given in \eqref{eq:subtwistelementsbilocal} that
            \begin{equation*}
                \begin{split}
                    A &= \R^{-1}_{(1,...,N)\underline{\vec b}}\ {}^1\gamma^{(1)}_{k|\ell}(\T_{\underline{\vec a}}\T_{\underline{\vec b}}\otimes \T_{1}\cdots \T_{M}, \T_{\underline{\vec c}}\otimes \T_{M+1}\cdots \T_K)\R_{(1,...,N)\underline{\vec b}} \\
                    &\qquad \qquad  - {}^1\gamma^{(1)}_{k|\ell}(\T_{\underline{\vec b}}\otimes \T_{N+1}\cdots \T_M, \T_{\underline{\vec c}}\otimes \T_{M+1}\cdots \T_K)\big.
                \end{split}
            \end{equation*}
            In a completely analogous way, it can be shown that
            \begin{equation*}
                \begin{split}
                    B &= {}^1\gamma^{(1)}_{k|\ell}(\T_{\underline{\vec a}}\otimes \T_{1}\cdots \T_N, \T_{\underline{\vec b}}\otimes \T_{N+1}\cdots \T_M)\\
                &\qquad \qquad  - \R^{-1}_{(N+1,...,M)\underline{\vec c}}{}^1\gamma^{(1)}_{k|\ell}(\T_{\underline{\vec a}}\otimes \T_{1}\cdots \T_N, \T_{\underline{\vec b}}\T_{\underline{\vec c}}\otimes \T_{N+1}\cdots \T_K)\R_{(N+1,...,M)\underline{\vec c}},
                \end{split}
            \end{equation*}
            thus proving the statement for $a\otimes b, c\otimes d,e\otimes f \in \left(\A^{(\geq 1)}(\R)\bowtie_{\mathbf{r}}\A_0^{(\geq \ell-1)}\right)$. Note here that the results still hold if we substitute $\T_{\underline{\vec a}}\to 1$ with $N>0$, and/or $\T_{\underline{\vec b}}\to 1$ with $M>N$, and/or $\T_{\underline{\vec c}}\to 1$ with $K>M$, which then also proves the case where $a\otimes b$, $c\otimes d$ and/or $e\otimes f$ are elements of $\A_0$, thus also proving the lemma.
        \end{proof}
    \end{lemma}

    We now have all the ingredients to prove Proposition~\ref{prop:vanishingdrinfeldassociatorbilocal}:

    \begin{proof}[\textbf{Proof for Proposition~\ref{prop:vanishingdrinfeldassociatorbilocal}}]
        Let us first write $\varphi^{(1)}_{k|\ell} = {}^1\varphi^{(1)}_{k|\ell} + {}^2\varphi^{(1)}_{k|\ell} + {}^3\varphi^{(1)}_{k|\ell}$, where ${}^i\varphi^{(1)}_{k|\ell}$ is the Drinfeld associator corresponding to the twist ${}^i\gamma^{(1)}_{k|\ell}$ as defined in \eqref{eq:subtwistelementsbilocal}. Next, let $a\otimes b, c\otimes d, e\otimes f \in \A_0 \oplus \left(\A^{(\geq 1)}(\R)\bowtie_{\mathbf{r}}\A_0^{(\geq \ell-1)}\right)$, then
        \begin{equation}
            \begin{split}
                {}^i\varphi^{(1)}_{k|\ell}(a\otimes b, c\otimes d, e\otimes f) &= \epsilon(a\cdot b)\ {}^i\gamma^{(1)}_{k|\ell}(c\otimes d, e\otimes f) + {}^i\gamma^{(1)}_{k|\ell}(a\otimes b, (c\otimes d)(e\otimes f))\\
                & \quad - {}^i\gamma^{(1)}_{k|\ell}((a\otimes b)(c\otimes d), e\otimes f) - {}^i\gamma^{(1)}_{k|\ell}(a\otimes b, c\otimes d)\epsilon(e\cdot f)
            \end{split}
        \end{equation}
        for $i = 1,2,3$. Let us first consider ${}^2\varphi^{(1)}_{k|\ell}$. From Lemma~\ref{lem:usefulbilocaltwistidentities}, it follows that
        \begin{equation*}
            \begin{split}
                &{}^2\gamma^{(1)}_{k|\ell}(a\otimes b, (c\otimes d)(e\otimes f)) = \sum \mathbf{r}^{-1}(d_{(1)}, e_{(1)})\ {}^2\gamma^{(1)}_{k|\ell}(a\otimes b, c\cdot e_{(2)}\otimes d_{(2)}\cdot f)\mathbf{r}(d_{(3)}, e_{(3)})\\ 
                &=\sum\mathbf{r}^{-1}(e_{(1)}, f_{(1)})\gamma^{(1)}_k(a\otimes b, c\otimes d\cdot f_{(2)})\mathbf{r}(e_{(2)},f_{(3)})\\
                &\quad + \sum \mathbf{r}^{-1}(d_{(1)}, e_{(1)})\mathbf{r}^{-1}(b_{(1)}, c_{(1)})\gamma^{(1)}_k(a\cdot c_{(2)}\otimes b_{(2)}, e_{(2)}\otimes d_{(2)}\cdot f)\mathbf{r}(b_{(3)}, c_{(3)})\mathbf{r}(d_{(3)}, e_{(3)})\\
                &= {}^2\gamma^{(1)}_{k|\ell}(a\otimes b, c\otimes d)\epsilon(e\cdot f)\\
                &+ \sum \mathbf{r}^{-1}(d_{(1)}, e_{(1)})\mathbf{r}^{-1}(b_{(1)}, c_{(1)})\ {}^2\gamma^{(1)}_{k|\ell}(a\cdot c_{(2)}\otimes b_{(2)}, e_{(2)}\otimes d_{(2)}\cdot f)\mathbf{r}(b_{(3)}, c_{(3)})\mathbf{r}(d_{(3)}, e_{(3)}).
            \end{split}
        \end{equation*}
        where in the second equality we used that $\sum \mathbf{r}(e_{(1)}, d_{(1)})\mathbf{r}(d_{(2)}, e_{(2)}) = \epsilon(e)\epsilon(d)$ since $\mathbf{r}$ is cotriangular. Similarly, we have
        \begin{equation*}
            \begin{split}
                &{}^2 \gamma^{(1)}_{k|\ell}((a\otimes b)(c\otimes d), e\otimes f) = \sum \mathbf{r}^{-1}(b_{(1)}, c_{(1)})\ {}^2\gamma_{k|\ell}^{(1)}(a\cdot c_{(2)}\otimes b_{(2)}\cdot d, e\otimes f)\mathbf{r}(b_{(3)}, c_{(3)})\\
                &= \sum \mathbf{r}^{-1}(b_{(1)}, c_{(1)})\mathbf{r}^{-1}(d_{(1)},e_{(1)})\ {}^2\gamma_{k|\ell}^{(1)}(a\cdot c_{(2)}\otimes b_{(2)}, e_{(2)}\otimes d_{(2)}\cdot f)\mathbf{r}(d_{(3)}, e_{(3)})\mathbf{r}(b_{(3)}, c_{(3)})\\
                &\quad +\sum \mathbf{r}^{-1}(a_{(1)}, b_{(1)})\ {}^2\gamma_{k|\ell}^{(1)}(a_{(2)}\cdot c \otimes d, e\otimes f)\mathbf{r}(a_{(3)}, b_{(3)}).
            \end{split}
        \end{equation*}
        Combining the terms now gives
        \begin{equation*}
            \begin{split}
                {}^2\varphi^{(1)}_{k|\ell}(a\otimes b, c\otimes d, e\otimes f) &= \epsilon(a\cdot b)\ {}^2\gamma^{(1)}_{k|\ell}(c\otimes d, e\otimes f)\\
                &-\sum \mathbf{r}^{-1}(a_{(1)}, b_{(1)})\ {}^2\gamma_{k|\ell}^{(1)}(a_{(2)}\cdot c \otimes d, e\otimes f)\mathbf{r}(a_{(3)}, b_{(3)}).
            \end{split}
        \end{equation*}
        Similarly, it can be derived using Lemma~\ref{lem:usefulbilocaltwistidentities} that
        \begin{equation*}
            \begin{split}
                {}^3\varphi^{(1)}_{k|\ell}(a\otimes b, c\otimes d, e\otimes f) &= \sum \mathbf{r}^{-1}(d_{(1)}, e_{(1)})\ {}^3\gamma^{(1)}_{k|\ell}(a\otimes b, c\cdot e_{(2)}\otimes d_{(2)})\epsilon(f)\mathbf{r}(d_{(3)}, e_{(3)})\\
                &- {}^3\gamma^{(1)}_{k|\ell}(a\otimes b, c\otimes d)\epsilon(e\cdot f).
            \end{split}
        \end{equation*}
        It now follows immediately from Lemma~\ref{lem:usefulidentitybilocal2} that
        \begin{equation*}
            \begin{split}
                \varphi^{(1)}_{k|\ell}(a\otimes b, c\otimes d, e\otimes f) = \left({}^1\varphi^{(1)}_{k|\ell}+{}^2\varphi^{(1)}_{k|\ell}+{}^3\varphi^{(1)}_{k|\ell}\right)(a\otimes b, c\otimes d, e\otimes f) = 0,
            \end{split}
        \end{equation*}
        thus proving the statement.
    \end{proof}

    \section{Relation to the $\Theta$-morphism}\label{sec:Thetamorphismrelation}
    In this section, we will discuss the relation between the $\B[\mathbb{Q}_3]$ long-range deformation and the \textit{$\Theta$-morphism} that was described in \cite{gromov2014theta}. In that paper, the $\Theta$-morphism, which is a second-order derivative on the inhomogeneities of the spin chain, was used to derive the transfer matrix and the corresponding $N$-magnon states for the $B[\mathbb{Q}_3]$ long-range deformation. Here, we will show that our transfer matrix derived from the deformed Lax operator is equal to the one found in that paper. Moreover, we then show that our deformed FRT-bialgebra contains operators that correspond to the $N$-magnon creation operators, by mapping these operators to the $N$-magnon states that were found in \cite{gromov2014theta}. We emphasise here that we do not have full proofs for some of the claims we make. However, all claims have been verified explicitly in \texttt{Mathematica}. This section is therefore mainly for the illustrative purpose of identifying the correct $N$-magnon creation operators that would be needed in a proper long-range Algebraic Bethe Ansatz.\\

    We consider the $R$-matrix corresponding to the Yangian $\mathcal{Y}(\mathfrak{gl}_2)$, which is given by
    \begin{equation}\label{eq:RmatrixYangiantheta}
        \R_{12}(u,v) = \frac{u-v-\P_{12}}{u-v-1} \in \mathrm{End}(\C^2\otimes \C^2)[\![u,v]\!],
    \end{equation}
    which can be related to the more usual Lax operator for the XXX spin chain (see e.g. \cite{Faddeev1995AlgebraicAspects}) under the reparametrisation $\hat u = \frac{i}{2} - iu$ and $\hat v = \frac{i}{2}-iv$. To make the relation between the $\B[\mathbb{Q}_3]$-deformed Lax operator as given in \eqref{eq:BQ3LaxExample} and the $\Theta$-morphism, we will first rewrite the Lax operator in terms of differential operators. We will consider a non-trivial spectral parameter in coordinate $\underline{1}$, and trivial spectral parameters in the coordinates $2$, $3$. We then get
    \begin{equation}
        \begin{split}
            \Q^{(3)}_{3(\underline{1}2)}\R_{(\underline{1}2)3} &= \R^{-1}_{\underline{1}2}\Q^{(3)}_{3(2\underline{1})}\R_{(2\underline{1})3}\R_{\underline{1}2} = \R^{-1}_{\underline{1}2}\left(\Q^{(3)}_{2\underline{1}} + [\Q^{(2)}_{2\underline{1}}, \Q^{(2)}_{32}]\right)\R_{\underline{1}(32)}\R_{23}\\
            &= \R^{-1}_{\underline{1}2}\left(-\R_{\underline{1}2}\Q^{(3)}_{\underline{1}2}\R_{\underline{1}3} + \Q^{(2)}_{2\underline{1}}\Q^{(2)}_{32}\R_{\underline{1}(32)} - \Q^{(2)}_{32}\Q^{(2)}_{2\underline{1}}\R_{\underline{1}(32)}\right)\R_{23},
        \end{split}
    \end{equation}
    where we used that $\R_{23} = \P_{23}$, $\Q^{(3)}_{23} = 0$ and $\R_{\underline{1}2}\Q^{(3)}_{\underline{1}2} = -\Q^{(3)}_{2\underline{1}}\R_{\underline{1}2}$, since $\R$ is of difference form.\footnote{Indeed, from braiding unitarity $\R_{21}\R_{12} = 1$ and difference form $(\D_1+\D_2)(\R_{12}) = 0$, it follows that $\R_{12}\Q^{(2)}_{12} = \Q^{(2)}_{21}\R_{12}$. By taking the derivative of this equation, the identity for $\Q^{(3)}_{12}$ follows.} Next, from the definition of the third algebraic charge density, we have 
    \begin{equation}
        \Q^{(3)}_{\underline{1}2} = -\left(\Q^{(2)}_{\underline{1}2}\right)^2 + \R^{-1}_{\underline{1}2}\D_{\underline{1}}^2(\R_{\underline{1}2}),
    \end{equation} 
    and by the Sutherland equation, we have 
    \begin{equation}
        \Q^{(2)}_{32}\R_{\underline{1}(32)} = \R_{\underline{1}(32)}\Q^{(2)}_{32} + \R_{\underline{1}2}\Q^{(2)}_{3\underline{1}}\R_{13} - \Q^{(2)}_{2\underline{1}}\R_{\underline{1}(32)}.
    \end{equation}
    Using the fact that $\R$ is of difference form such that $\R_{\underline{1}2}\Q^{(2)}_{\underline{1}2} = \Q^{(2)}_{2\underline{1}}\R_{\underline{1}2}$, we get 
    \begin{equation}
        \begin{split}
            &\Q^{(3)}_{3(\underline{1}2)}\R_{(\underline{1}2)3}\\
            &=\R^{-1}_{\underline{1}2}\left(-\D_1^2(\R_{\underline{1}2})\R_{\underline{1}3} + \Q^{(2)}_{2\underline{1}}\R_{\underline{1}2}\Q^{(3)}_{3\underline{1}}\R_{\underline{1}3} + \Q^{(2)}_{2\underline{1}}\R_{\underline{1}2}\R_{\underline{1}3}\Q^{(2)}_{32} - \Q^{(2)}_{32}\Q^{(2)}_{2\underline{1}}\R_{\underline{1}(32)}\right)\R_{23}
        \end{split}
    \end{equation}
    Next, using $\Q^{(2)}_{2\underline{1}}\R_{\underline{1}2}\R_{\underline{1}3} = -[\Q^{(2)}_{32}, \R_{\underline{1}2}\R_{\underline{1}3}] + \R_{\underline{1}2}\Q^{(2)}_{3\underline{1}}\R_{\underline{1}3}$, we get
    \begin{equation}
        \begin{split}
            \Q^{(3)}_{3(\underline{1}2)}\R_{(\underline{1}2)3} &=\R^{-1}_{\underline{1}2}\left(-\D_1^2(\R_{\underline{1}2})\R_{\underline{1}3} + \Q^{(2)}_{2\underline{1}}\R_{\underline{1}2}\Q^{(3)}_{3\underline{1}}\R_{\underline{1}3}-[\Q^{(2)}_{32}, \R_{\underline{1}(32)}]\Q^{(2)}_{32}\right)\R_{23}\\
            & + \Q^{(2)}_{3\underline{1}}\R_{\underline{1}3}\Q^{(2)}_{32}\R_{23} - \R^{-1}_{\underline{1}2}\Q^{(2)}_{32}\Q^{(2)}_{2\underline{1}}\R_{\underline{1}2}\R_{\underline{1}3}\R_{23}.
        \end{split}
    \end{equation}
    Note now that $\Q^{(2)}_{\underline{2}1}\R_{\underline{1}2} = \D_1(\R_{\underline{1}2}) = - \D_2(\R_{\underline{1}2})$ since $\R$ is of difference form, and where one should interpret the derivative in the second coordinate as $\D_2(\R_{\underline{1}2}) = \frac{d}{dv}\R_{\underline{1}2}(u,v)\big|_{v=0}$, i.e. the derivative evaluated at a trivial spectral parameter. The above formula can now be completely written in terms of these kind of derivations. Namely, for some (matrix-valued) functions $g(v_1,...,v_n) \in \text{End}(V^{\otimes L})$, we define the operator
    \begin{equation}\label{eq:xioperator}
        \xi_{n,m}(g) \coloneq \left(-\D_m^2 g + \D_n\D_m g - [\mathfrak{q}^{(2)}_{nm}, g]\mathfrak{q}^{(2)}_{nm}\right)\bigg|_{v_i = 0},
    \end{equation}
    where $\D_i$ denotes the derivative in the $i$-th coordinate. Moreover, it can be calculated using \texttt{Mathematica} that there is a matrix $\mathfrak{m}$ such that
    \begin{equation}
        \Q^{(2)}_{3\underline{1}}\R_{\underline{1}3}\Q^{(2)}_{32}\R_{23} - \R^{-1}_{\underline{1}2}\Q^{(2)}_{32}\Q^{(2)}_{2\underline{1}}\R_{\underline{1}2}\R_{\underline{1}3}\R_{23} = [\mathfrak{m}_{\underline{1}2}, \R_{(\underline{1}2)3}].
    \end{equation}
    Combining all the terms now gives that the Lax operator corresponding to the $\B[\mathbb{Q}_3]$ deformation \eqref{eq:BQ3LaxExample} is given by
    \begin{equation}
        \L^{\lambda}_{(\underline{\vec a}b)n} = \R_{(\underline{\vec a}b)n} + \lambda\left(\R^{-1}_{\underline{\vec a}b}\xi_{nb}(\R_{\underline{\vec a}(nb)})\R_{bn} + [\mathfrak{m}_{\underline{a}b}, \R_{(\underline{a}b)n}]\right) + \O(\lambda^2).
    \end{equation}
    If we now calculate the monodromy matrix $\T^\lambda_{\underline{\vec a}b} \coloneq  \L^{\lambda}_{(\underline{\vec a}b)L}\cdots\L^{\lambda}_{(\underline{\vec a}b)1}$ for $L \geq 2$, it follows that the first order correction term is given by
    \begin{equation}\label{eq:monodromyMatrixtheta}
        \begin{split}
            \T^{(1)}_{\underline{\vec a}b} &= \sum_{n=1}^{L}\xi_{n,n+1}\left(\T^{(0)}_{\underline{\vec a}}\right)\cdot\T^{(0)}_b + \R^{-1}_{\underline{\vec a}b}\left[\T^{(0)}_b\cdot\xi_{L-1,L}\left(\T^{(0)}_{\underline{\vec a}}\right)\right]\R_{\underline{\vec a}b} - \xi_{L,1}\left(\T^{(0)}_{\underline{\vec a}}\right)\cdot \T^{(0)}_b\\
            &+ [\mathfrak{m}_{\underline{\vec a}b}, \T^{(0)}_{\underline{\vec a} b}],
        \end{split}
    \end{equation}
    where $\T^{(0)}_{\underline{\vec a}} \coloneq \R_{\underline{\vec a}L}\cdots \R_{\underline{\vec a}1}$ and $\T^{(0)}_{b} \coloneq \P_{bL}\cdots \P_{b1}$ are the zero-th order terms of the monodromy matrix. This is a particularly nice closed-form formula completely in terms of the operator $\xi$ as given in \eqref{eq:xioperator}. Remember here that the above form only holds for $L\geq 2$, i.e. does not involve any wrapping effects. If we take the trace over the auxiliary space, it follows from the fact that $\mathrm{tr}_b(\T^{(0)}_{b}) = \U$, with $\U$ the shift operator, that only the first term in \eqref{eq:monodromyMatrixtheta} survives. In particular, the transfer matrix is given by
    \begin{equation}
        \begin{split}
            t^{(1)}(u)\cdot \U^{-1} &\coloneq  \mathrm{tr}_{\underline{a}b}\T^{(1)}_{\underline{a}b}(u)\cdot \U^{-1} = \sum_{n=1}^{L}\xi_{n,n+1}(t^{(0)}(u)) \\
            &= \sum_{n=1}^{L}\left(-\frac{1}{2}(\D_n - \D_{n+1})^2(t^{(0)}(u)) - [\mathfrak{q}^{(2)}_{n,n+1}, t^{(0)}(u)]\mathfrak{q}^{(2)}_{n,n+1}\right),
        \end{split}
    \end{equation}
    where $t^{(0)}(u)$ is the undeformed transfer matrix. One can now recognise the same form for the first order correction to the transfer matrix as was found in \cite[(6.25)]{gromov2014theta} using the $\Theta$-morphism!\\

    Next, we discuss the $N$-magnon creation operators for the long-range deformation. Here, we assume the reader is already familiar with the Algebraic Bethe Ansatz for the undeformed spin chain, see e.g. \cite{faddeev1996algebraicbetheansatzwork, Retore2021} for introductions. For the undeformed spin chain, the $N$-magnon creation operators are given by $B(u_1)B(u_2)\cdots B(u_N)$, where $B(u_i) \coloneq t_{12}(u_i)$, which act the vacuum state $|\Omega\rangle \coloneq |{\downarrow}\rangle\otimes \cdots \otimes |{\downarrow}\rangle$. For the deformed spin chain, however, we have to take into account the doubling of the auxiliary space of the FRT-bialgebra. A natural generalisation, therefore, is that the $N$-magnon creation operators on the $\B[\mathbb{Q}_3]$-deformed spin chain are of the form $B(u_1)\cdots B(u_N)\otimes^\lambda\left[t_{11}(0) + f(u_1,...,u_N)t_{22}(0)\right]$ for some symmetric function $f(u_1,...,u_N)$. Here, the notation $\otimes ^\lambda$ means that we view this tensor product as elements in the deformed double-crossed product algebra $(\A(\R)\bowtie_{\mathbf{r}}\A_0)(\gamma_{3})$. Indeed, since the $B(u_i)$'s mutually commute, it is ensured that these operators are symmetric in the rapidities $u_i$. Moreover, an $N$-tuple of products of $B$'s also ensure that these create states with $N$ spin-up excitations. By making the correspondence to the $\Theta$-morphism, we will show that the creation operators of this form are indeed the correct $N$-magnon creation operators.\\

    We will denote the matrix elements of the monodromy matrix $\T_{\underline{a}b}^\lambda(u_1,...,u_N)$ as given in \eqref{eq:monodromyMatrixtheta} by $t_{i_1,j_1}(u_1)\cdots t_{i_N,j_N}(u_N)\otimes t_{k,\ell}(0)$, where we will also use the notation $A(u)\coloneq t_{11}(u)$, $B(u)\coloneq t_{12}(u)$, $C(u) \coloneq t_{21}(u)$ and $D(u) \coloneq t_{22}(u)$. Recall that the coproduct of the $R$-matrix is given by
    \begin{equation}
        \R_{(\underline{a}_1,...,\underline{a}_N)b} = \R_{\underline{a}_1,b}\cdots \R_{\underline{a}_N,b}.
    \end{equation}
    Applying this to the braiding $\R^{-1}_{\underline{\vec a}b}\left[\T^{(0)}_b\cdot\xi_{L-1,L}\left(\T^{(0)}_{\underline{\vec a}}\right)\right]\R_{\underline{\vec a}b}$, it can be calculated from \eqref{eq:monodromyMatrixtheta} that the first order corrections to the operators acting on the spin chain are given by
    \begin{equation}\label{eq:BBAtheta}
        \begin{split}
            B(u_1)\cdots B(u_N)\otimes^{(1)} t_{11}(0) &= \sum_{n=1}^{L-1}\xi_{n,n+1}(B(u_1)\cdots B(u_N))A(0) \\
            &+ f^{-}_0(\vec u)A(0)\xi_{L-1,L}(B(u_1)\cdots B(u_N))\\
            &+ \sum_{i=1}^{N}f^{-}_i(\vec u)B(0)\xi_{L-1,L}(B(u_1)\cdots A(u_i)\cdots B(u_N)) +...,
        \end{split}
    \end{equation}
    and
    \begin{equation}\label{eq:BBDtheta}
        \begin{split}
            B(u_1)\cdots B(u_N)\otimes^{(1)}t_{22}(0) &= \sum_{n=1}^{L-1}\xi_{n,n+1}(B(u_1)\cdots B(u_N))D(0)\\
            &+ f_0^{+}(\vec u)D(0)\xi_{L-1,L}(B(u_1)\cdots B(u_N))\\
            &+ \sum_{i=1}^Nf_i^{+}(\vec u)B(0)\xi_{L-1,L}(B(u_1)\cdots D(u_i)\cdots B(u_N)) +...,
        \end{split}
    \end{equation}
    where $A(u), D(u), B(u)$ are the undeformed matrix elements $t_{11}(u), t_{22}(u)$ and $t_{12}(u)$ of the undeformed monodromy matrix $\T_{\underline{a}}(u) = \R_{\underline{a}L}(u)\cdots \R_{\underline{a}1}$, respectively, the functions $f_i^{\pm}$ are given by
    \begin{equation}
        f_i^{\pm}(\vec u) \coloneq f_i^{\pm}(u_1,...,u_N) = \begin{cases}
            \prod_{n=1}^N \frac{u_n}{u_n\pm 1} & \text{for }i = 0,\\
            \pm\frac{1}{u_{i}}\prod_{n=i}^{N}\frac{u_n}{u_n\pm 1} & \text{for }1\leq i \leq N,
        \end{cases}
    \end{equation}
    and lastly the ``$...$'' in \eqref{eq:BBAtheta} and \eqref{eq:BBDtheta} include the terms corresponding to the similarity transformation by $\mathfrak{m}$ of the monodromy matrix \eqref{eq:monodromyMatrixtheta}, which just corresponds to a change of basis of the algebra generators. We now want to combine \eqref{eq:BBAtheta} and \eqref{eq:BBDtheta} such that the $A(0)$ and $D(0)$ in the second terms add to the shift operator, and the $N$-magnon state will contain a term $\U \cdot \xi_{L-1,L}\sim \xi_{L,1}\cdot \U$. Therefore, we define the $N$-magnon state to be given by
    \begin{equation}\label{eq:NmagnonstateTheta}
        |u_1,...,u_N\rangle^\lambda \coloneq \left(B(u_1)\cdots B(u_N)\otimes^\lambda\left[t_{11}(0) + \frac{f^-_0(u_1,...,u_N)}{f^+_0(u_1,...,u_N)} t_{22}(0)\right]\right)|\Omega\rangle.
    \end{equation}
    Note that this state is per construction symmetric in the rapidities $u_1,..., u_N$ and that it flips $N$ down-spins into up-spins. Since $A(0)|\Omega\rangle = |\Omega\rangle$ and  $D(0)|\Omega\rangle = 0$, it follows that the zero-th order term of this state is simply given by $|u_1,...,u_N\rangle^{(0)} = B(u_1)\cdots B(u_N)|\Omega\rangle$, which is indeed the zero-th order $N$-magnon state when the rapidities $(u_1,...,u_N)$ are on-shell, i.e. satisfy the Bethe equations
    \begin{equation}
        \left[\frac{\hat u_n + \frac{i}{2}}{\hat u_n - \frac{i}{2}}\right]^L = \prod_{m\neq n}\frac{\hat u_n - \hat u_m + i}{\hat u_n - \hat u_m - i},
    \end{equation}
    for $\hat u_n \coloneq \frac{i}{2} - iu_n$. The first-order correction to this state is given by
    \begin{equation}\label{eq:bulktermNmagnontheta}
        \begin{split}
            |u_1,...,u_N\rangle^{(1)} &= \sum_{n=1}^{L}\xi_{n,n+1}\left(B(u_1)\cdots B(u_N)\right)|\Omega\rangle + \text{extra terms},
        \end{split}
    \end{equation}
    where the extra terms are given by
    \begin{equation}
        \begin{split}
            &\text{extra terms}\\
            &=\bigg[f_0^{-}(u_1,...,u_N)\U\cdot  \xi_{L-1,L}\left(B(u_1)\cdots B(u_N)\right)-\xi_{L,1}(B(u_1)\cdots B(u_N))\bigg]|\Omega\rangle\\
            &+\sum_{i=1}^{N}B(0)\cdot \xi_{L-1,L}\left(B(u_1)\cdots \left[f^{-}_i(\vec u)A(u_i) + \frac{f_0^{-}(\vec u)}{f_0^{+}(\vec u)} f_i^{+}(\vec u)D(u_i)\right]\cdots B(u_N)\right)|\Omega\rangle + ....
        \end{split}
    \end{equation}
    In \eqref{eq:bulktermNmagnontheta}, one can now see that the bulk term of the first-order correction to the $N$-magnon state is given by second-order derivatives in the inhomogeneities of the spin chain, similar to the $N$-magnon states that were derived in \cite{gromov2014theta}. The (highly non-trivial) claim is now that, for $(u_1,...,u_N)$ on-shell, there is a function $g(u_1,...,u_N)\in \C$ such that the extra terms are given by
    \begin{equation}
        \text{extra terms} = \left[-\mathfrak{q}^{(2)}_{L,1}\cdot \mathbb{Q}_2 + g(u_1,...,u_N)\right]|u_1,...,u_N\rangle^{(0)}.
    \end{equation}
    We remark that we do not have an explicit proof for this statement. However, it can be verified to be true using \texttt{Mathematica} (for small values of $L$ and $N$). In particular, using the fact that $\mathfrak{q}^{(2)}_{n,n+1}|\Omega\rangle = 0$, it then follows that the state \eqref{eq:NmagnonstateTheta}, for $(u_1,...,u_N)$ on-shell, is given by
    \begin{equation}
        |u_1,...,u_N\rangle^{\lambda} = \left(1 - \frac{\lambda}{2}\sum_{n=1}^{L}(\D_n - \D_{n+1})^2 -\lambda  \mathfrak{q}^{(2)}_{L,1}\cdot \mathbb{Q}_2 + \lambda g(u_1,...,u_N)\right)\left(|u_1,...,u_N\rangle^{(0)}\right),
    \end{equation}
    where one can now recognise the exact same $N$-magnon state as was found in \cite[(6.23)]{gromov2014theta} using the $\Theta$-morphism.\footnote{Note here that the extra factor of $ g(u_1,...,u_N)$ just corresponds to the rescaling of the state given by $|u_1,...,u_N\rangle^\lambda \to (1+\lambda g(u_1,...,u_N))|u_1,...,u_N\rangle^\lambda  + \O(\lambda^2)$.} In that paper, it was shown that these states are indeed eigenvectors for the long-range transfer matrix for $L\geq 2$, i.e. without wrapping effects. In particular, this shows that the state \eqref{eq:NmagnonstateTheta} is indeed the correct $N$-magnon creation operator, which is expected to also be the creation operator when wrapping effects are included.\footnote{For the $\B[\mathbb{Q}_3]$ deformation, this is a rather trivial statement, as wrapping only occurs on the spin chain of length $L=1$, for which the state \eqref{eq:NmagnonstateTheta} (almost) trivially is an eigenvector of the transfer matrix. For more general deformations, the creation operators are expected to be of similar forms, also on spin chains that include wrapping effects.} We emphasise again that we have not given any formal proof in our current discussion, as we bypassed some calculations using \texttt{Mathematica}. However, the important takeaway is that we are now able to identify the correct $N$-magnon creation operators in our algebra, which therefore brings us a small step closer to understanding the long-range Algebraic Bethe Ansatz. It is expected that the above discussion generalises to all the long-range deformations, and also to higher-order corrections.
\end{appendices}

\bibliographystyle{JHEP}
\bibliography{bibliography}

\end{document}